\documentclass[11pt]{article}

\usepackage[utf8]{inputenc}
\usepackage{palatino}
\usepackage{titling}
\usepackage{amsmath,amssymb,amsfonts,amsthm, bbm}
\usepackage{mathtools}
\usepackage{geometry}
\usepackage{graphicx}
\usepackage{enumitem}
\usepackage[colorlinks=true]{hyperref}
\usepackage{braket}
\usepackage{thm-restate}
\usepackage{indentfirst}
\usepackage{multirow}
\usepackage{xcolor}
\usepackage[most]{tcolorbox}

\hypersetup{pdfpagemode=UseNone}

\newtheorem{theorem}{Theorem}

\newtheorem{lemma}[theorem]{Lemma}
\newtheorem{claim}[theorem]{Claim}
\newtheorem{corollary}[theorem]{Corollary}
\newtheorem{remark}[theorem]{Remark}
\newtheorem{definition}{Definition}

\newcommand{\A}{\mathrm{A}}
\newcommand{\B}{\mathrm{B}}

\newcommand{\clA}{\mathcal{A}}
\newcommand{\clB}{\mathcal{B}}

\newcommand{\clP}{\mathcal{P}}
\newcommand{\clL}{\mathcal{L}}
\newcommand{\clX}{\mathcal{X}}
\newcommand{\clY}{\mathcal{Y}}

\newcommand{\ba}{\mathbf{a}}
\newcommand{\bA}{\mathbf{A}}
\newcommand{\bb}{\mathbf{b}}
\newcommand{\bB}{\mathbf{B}}
\newcommand{\bx}{\mathbf{x}}
\newcommand{\bX}{\mathbf{X}}
\newcommand{\by}{\mathbf{y}}
\newcommand{\bY}{\mathbf{Y}}
\newcommand{\bs}{\mathbf{s}}

\newcommand{\tZ}{\widetilde{Z}}

\newcommand{\bbC}{\mathbb{C}}
\newcommand{\bbE}{\mathop{\mathbb{E}}}
\newcommand{\bbR}{\mathbb{R}}

\newcommand{\sfP}{\mathsf{P}}
\newcommand{\sfQ}{\mathsf{Q}}
\newcommand{\sfU}{\mathsf{U}}

\DeclareMathOperator{\D}{D}
\DeclareMathOperator{\R}{R}
\DeclareMathOperator{\Q}{Q}

\newcommand{\eps}{\varepsilon}

\newcommand{\Tr}{\operatorname{Tr}}
\newcommand{\ketbra}[2]{\ket{#1}\!\!\bra{#2}}
\newcommand{\state}[1]{\ketbra{#1}{#1}}

\newcommand{\poly}{\mathrm{poly}}
\newcommand{\Var}{\mathrm{Var}}
\newcommand{\bias}{\mathrm{bias}}
\newcommand{\chsh}{\mathrm{CHSH}}
\newcommand{\wchsh}{\mathrm{WorstCaseCHSH}}
\newcommand{\MC}{\mathrm{MC}}
\newcommand{\HC}{\mathrm{HC}}
\newcommand{\col}{\mathrm{col}}
\newcommand{\norm}{\mathrm{norm}}
\newcommand{\pub}{\mathrm{pub}}
\newcommand{\dep}{\mathrm{dep}}
\newcommand{\uni}{\mathrm{uni}}
\newcommand{\era}{\mathrm{era}}
\newcommand{\res}{\mathrm{res}}
\newcommand{\IP}{\mathrm{RealIP}}
\newcommand{\EQ}{\mathrm{EQ}}
\newcommand{\diag}{\mathrm{diag}}

\newcommand{\bits}{\{0,1\}}

\title{Non-local games and communication complexity with noisy entanglement}
\author{Srijita Kundu \\
\small Quantum Computing Research Centre \\
\small Hon Hai (Foxconn) Research Institute \\
\small \texttt{srijita.kundu@foxconn.com.sg} \\
\and Olivier Lalonde \\
\small Institute for Quantum Computing \\
\small University of Waterloo \\
\small \texttt{olalonde@uwaterloo.ca}
}
\date{}

\begin{document}

\maketitle
\begin{abstract}
 We study the impact of noise on the theories of quantum nonlocality and entanglement-assisted communication complexity. We consider non-local games and entanglement-assisted communication complexity in a model where the players Alice and Bob may share arbitrarily many noisy EPR pairs. We study four noise models: depolarizing noise, unital noise, biased reset noise, and erasure noise. Our results are as follows:
 \begin{enumerate}
 \item We upper bound the value of the CHSH game under all these noise models in terms of the noise parameter, without any assumptions on the measurements used in the strategy.
 \item We prove a parallel repetition theorem for general non-local games under all noise models except biased reset noise; we prove an improved parallel repetition theorem for unique games. Our parallel repetition rate is smaller than the quantum parallel repetition rate for CHSH in a nontrivial noise regime.
 \item Using our unique-game parallel repetition theorem and a relation defined by the CHSH game, we prove a separation between communication complexity with noisy vs noiseless entanglement (the communication is classical in both cases). This implies an $\Omega(n)$ two-way communication lower bound for distilling $n$ EPR pairs from noisy EPR pairs, in the same nontrivial noise regime.
 \item We show that any interactive entanglement-assisted communication protocol can be simulated by an SMP communication protocol with noisy shared randomness, with an exponential blowup in communication. This generalizes a known result on the simulation of noiseless shared randomness with noisy shared randomness.
 \item We show a polynomial lower bound on the number of copies of noisy EPR pairs required to compute the Equality function with constant communication, under all four noise models. The previous result gives a matching upper bound, and moreover, this answers (in the negative) an open question in the literature on whether logarithmically many noisy shared bits suffice for communication.
 \end{enumerate}
\end{abstract}

\section{Introduction}
\subsection{Background}
Quantum entanglement has long been identified as one of the most distinctive features of quantum mechanics, and is used as a resource in a vast number of quantum information-theoretic tasks. Such tasks include the design of non-local strategies for non-local games \cite{CHTW10} and efficient communication protocols \cite{BCvD01}, for example using quantum teleportation \cite{BBC+93}. Although these tasks have been studied extensively, most known strategies and protocols assume pure shared entanglement, and little is known about the weaker model in which the shared entanglement is noisy, despite its likely relevance for near-term practical applications.

The particular model of noisy entanglement we will be studying is one in which the shared resource is an arbitrarily large number of independent copies of a given noisy entangled state $\Psi_{E^\A E^\B}$. For many such states of interest, it is possible to extract near-perfect pure entangled states (e.g. EPR pairs) from this resource using local operations and classical communication (LOCC) between the parties holding the two halves of the entangled state \cite{BBP+96,WTB17}. This process is called \emph{entanglement distillation}, and known protocols for it can distill $n$ near-perfect EPR pairs using $O(n)$ communication, with the hidden constant depending on the state $\Psi$ and on the desired fidelity. Despite this, in many settings where one might want to use entanglement, the noisy model may be strictly weaker than the noiseless model, even when the noise is low enough that the states are entangled and have distillable entanglement. This is because communication may either be forbidden or counted as a resource in such a setting, and entanglement distillation costs communication.

\paragraph{Non-local games.}
A non-local game is a scenario in which two spatially separated parties, traditionally called Alice and Bob, are given classical inputs and must provide outputs satisfying a given predicate. The largest probability with which they can do so in a given model is called the value of the game. This probability can be strictly larger when the parties share entanglement. For example, the CHSH game has classical value $3/4$ and quantum value $\cos^2\left(\frac{\pi}{8}\right) \approx 0.85$. In the usual quantum model, Alice and Bob may share an entangled state of arbitrary finite dimension.

Recent work \cite{Yao19,QY21,QY25,DFN+25} has established that noise has a dramatic effect on the complexity of computing the entangled value of a game. For certain fixed choices of $\Psi$, such as qubit Werner states, approximating this value was shown to be NP-complete, as in the unentangled setting. This stands in stark contrast to the case where Alice and Bob share an arbitrarily large number of perfect EPR pairs, where approximating the entangled value of a non-local game to any fixed nontrivial precision is undecidable \cite{JNVWY20}. 

The values of several specific non-local games in the noisy-entanglement model have also been studied recently \cite{FQXY26}, under additional assumptions on both the noisy state and the players' strategies. Their result for the CHSH game is as follows:
\begin{theorem}[\cite{FQXY26}]\label{thm:FQXY-CHSH}
Suppose that $\Psi$ is a qubit Werner state, namely, an EPR pair which is replaced with the maximally mixed state with probability $1-\mu$. If Alice and Bob's observables are traceless, their winning probability at CHSH is at most $\frac12+\frac1{2\sqrt{2}}\mu$.\footnote{Technically, this result also works for any noise model such that the marginals of EPR pairs remain maximally mixed, i.e., unital noise. But this is not explicitly stated in \cite{FQXY26}.}
\end{theorem}

\paragraph{Communication complexity.}
Communication complexity studies settings in which Alice and Bob receive classical inputs and communicate to produce outputs satisfying a given predicate. The central quantity is the least amount of communication needed to succeed with high probability. The two most commonly studied models in quantum communication complexity are the one with quantum communication without prior entanglement, introduced by Yao \cite{Yao93}, and the one with classical communication assisted by prior entanglement, introduced by Buhrman, Cleve, and van Dam \cite{BCvD01}.\footnote{By quantum teleportation, additionally allowing quantum messages in the entanglement-assisted model reduces the communication by a factor of at most $2$.} Raz gave the first exponential separation between classical and quantum communication complexity without prior entanglement for a partial Boolean function \cite{Raz99}. Subsequent works established related exponential separations in more restrictive communication models and for total search problems \cite{Gav08,KR11,Gav16,GGJL25}.

Unlike the noisy version of Yao's model, in which the channel used to transmit quantum messages is noisy, which has been studied before \cite{BNT+14}, as far as we know, the model in which Alice and Bob share noisy entanglement has never been studied in the literature. However, its classical analogue, in which Alice and Bob share imperfectly correlated randomness, has received some attention \cite{BGI15,CGMS15}. Bavarian, Gavinsky, and Ito \cite{BGI15} showed that the Equality function on $n$ bits has a worst-case bounded-error simultaneous-message-passing (SMP) protocol with constant communication when Alice and Bob share sufficiently many independent realizations of dependent (but not perfectly correlated) random variables $R^\A$ and $R^\B$, respectively; their protocol uses $O(2^n)$ realizations. Canonne, Guruswami, Meka, and Sudan \cite{CGMS15} studied the special case in which $R^\A$ and $R^\B$ are $\rho$-correlated bits with uniform marginals, and obtained the following result.
\begin{theorem}[\cite{CGMS15}] \label{thm:canonne}
Any Boolean function with $n$-bit inputs that has a $C$-bit two-way classical protocol with perfect shared randomness also has a one-way classical protocol using $\poly_\rho(n, 2^C)$ many $\rho$-correlated uniform bits and $2^{O_\rho(C)}$ communication. 
\end{theorem}
This result shows that every communication problem having a constant-communication public-coin protocol also has a constant-communication protocol using imperfectly correlated randomness of this type. In particular, when specialized to the Equality function on $n$ bits, this gives a constant-communication protocol with $\poly(n)$ shared samples, thereby improving exponentially on the cost of the protocol of \cite{BGI15}. They asked whether an analogue of Newman's theorem \cite{New91} could be shown to hold for imperfect shared randomness, which could further reduce the shared randomness cost to $\Theta(\log n)$, as is known to be possible with perfect shared randomness. 

\cite{CGMS15} also exhibited a partial Boolean function which can be computed with $k$ bits of communication when given access to perfect shared randomness, but requires $2^{\Omega_\rho(k)}$ communication when given access to $\rho$-correlated uniform bits, for every $\rho<1$. Note, however, that in the imperfect-shared-randomness model, Alice and Bob still have access to private randomness, and Newman's result \cite{New91} shows that randomized communication with private and public randomness differs by at most $O(\log n)$. Thus, the separation in \cite{CGMS15} is not exponential in the input size\footnote{$k$ is a parameter defining a sparsity promise on their inputs. The communication only depends on this sparsity promise, and is independent of the input size, in both cases.}, whereas separations between communication with noisy and noiseless entanglement could potentially be exponential in the input size.

\subsection{Our results}
In this work, we further the study of non-local games and initiate the study of communication complexity in the presence of noisy entanglement. As mentioned previously, our model of noisy entanglement is one in which Alice and Bob share an unbounded number of copies of a fixed mixed quantum state $\Psi_{E^\A E^\B}$. Our results illustrate that this is a strictly weaker model than the one in which Alice and Bob share noiseless entanglement for a number of tasks of interest, even though the amount of shared entanglement can be unbounded.

For many of our results, in line with the previous literature, we assume that $\Psi$ is the result of passing one half of a perfect EPR pair through a noisy channel. We consider the following noise models, which are described in more detail in Section~\ref{sec:noisyentanglement}:
\begin{itemize}
\item Depolarizing noise;
\item Unital noise, which is a very general class of noise containing depolarizing and dephasing noise;
\item Biased reset noise, which is a different generalization of depolarizing noise;
\item Erasure noise.
\end{itemize}

\subsubsection{Results about non-local games}

\paragraph{Upper bounds on the noisy entangled values of the CHSH game.} Our first result consists of new upper and lower bounds on the value of the CHSH game under the aforementioned noise models. Unlike \cite{FQXY26}, our result requires no assumptions on Alice or Bob's observables. The upper bounds and best-known lower bounds under these models are summarized in Table~\ref{tab:chsh-ub-lb}.
\begin{table}[!ht]
\centering
\begin{tabular}{|c|c|c|}
\hline
Noise model & Upper bound & Lower bound \\
\hline
 & & \\
Depolarizing & \multirow{2}{*}{$\min\left\{\frac12+\frac1{4\sqrt{1-\mu^2}}, \frac12 + \frac{\sqrt{1+\mu}}{4} \right\}$} & \multirow{2}{*}{$\frac12+\frac\mu{2\sqrt{2}}$} \\
Parameters: $\mu$ & & \\
& & \\
Unital & \multirow{2}{*}{$\min\left\{\frac12+\frac1{4\sqrt{1-s_1^2}}, \frac12 + \frac{\sqrt{1+s_1}}{4}\right\}$} & \multirow{2}{*}{$\frac12+\frac14\sqrt{s_1^2+s_2^2}$} \\
 Parameters: $s_1, s_2$ & & \\
 & & \\
 Reset & \multirow{2}{*}{$\min\left\{\frac12+\frac14\sqrt{\frac{1-(1-\lambda)|b|}{1-(1-\lambda)|b|-\lambda^2}}, \frac12+\frac1{2\sqrt{2}}\right\}$}  & \multirow{2}{*}{$\frac12+\frac\lambda{2\sqrt{2}}$} \\
 Parameters: $\lambda, b$ & & \\
 & & \\
 Erasure & \multirow{2}{*}{$\min\left\{\frac12+\frac1{4\sqrt{\eps}}, \frac12+\frac1{2\sqrt{2}}\right\}$} & \multirow{2}{*}{$\frac12+\frac{1-\eps}{2\sqrt{2}}$}\\
 Parameters: $\eps$ & & \\
 & & \\
\hline
\end{tabular}
\caption{Upper and lower bounds on the noisy quantum value of the CHSH game under depolarizing, unital, reset, and erasure noise. The lower bounds are shown in Corollary~\ref{cor:mu-chsh} and Lemma~\ref{lem:chsh-uni-lb}; the upper bounds are shown in Corollary~\ref{cor:chsh-comb}. A comparison of the two upper bounds, along with the lower bound, is given in Figure~\ref{fig:chsh-bounds} in Section~\ref{sec:game-ub}.}
\label{tab:chsh-ub-lb}
\end{table}

Our upper bounds are derived in terms of the maximal correlation of the noisy entangled state, a quantum generalization of the classical notion introduced by Beigi in \cite{Bei13}. The maximal correlation of EPR pairs under depolarizing noise was calculated in \cite{QY21}; we compute the maximal correlation of EPR pairs under the other noise models in Lemmas~\ref{lem:MC-unital}, \ref{lem:MC-reset}, and \ref{lem:MC-erasure}.

For unital noise (including depolarizing noise), we have a better upper bound which is the minimum of two functions
\[ \frac12+\frac1{4\sqrt{1-\rho^2}}\qquad \text{and} \qquad \frac12 + \frac{\sqrt{1+\rho}}{4} \]
of the maximal correlation $\rho$. For other noise models, it is the minimum of the first function (which diverges as $\rho\to1$) and the noiseless quantum value. This is because, for unital noise, we can upper bound the value of a game by an NPA local level 1 SDP, using an NPA local level 1 characterization of maximal correlation. This result holds for any game, not just CHSH, and is shown in Theorem~\ref{thm:noisy-NPA}. On the other hand, the upper bound in terms of the divergent function works for all of the noise models considered and uses special properties of the CHSH game; this is shown in Theorem~\ref{thm:chsh-c}.

The lower bounds all come from the obvious strategy of playing the original CHSH strategy on a single noisy EPR pair. It is unclear whether the upper bounds should match these lower bounds, since it is possible that more than one copy of the noisy EPR pair could help. However, the upper bound cannot be optimal in the depolarizing-noise case: depolarized EPR pairs are unentangled for $\mu \leq \frac13$, and the noisy quantum value at these values of $\mu$ should reduce to the classical value $0.75$, which our upper bound does not. The bound from \cite{FQXY26} stated in Theorem~\ref{thm:FQXY-CHSH} for depolarizing noise does match the depolarizing lower bound, but it requires the extra assumption about traceless observables.

\begin{remark}
All the upper bounds obtained in these results are in terms of the maximal correlation of the shared entangled state (which is the same for $m$ copies of a state as for a single copy). The fact that we started with EPR pairs is only used to calculate this maximal correlation under the various noise models --- so in principle the bound in terms of maximal correlation could hold for arbitrary noisy entangled states.
\end{remark}

\paragraph{Noisy parallel repetition theorems for non-local games.}
Our next set of results is parallel repetition theorems for non-local games with noisy entanglement. A parallel repetition theorem for a non-local game $G$ in a specific model (classical, quantum, noisy quantum) shows that if the value of a single copy of $G$ in that model is less than $1$, then the value of $n$ copies of $G$ played in parallel, denoted by $G^n$, goes down exponentially in $n$. This is of course true if the optimal strategy for $G^n$ is to just play the optimal strategy for $G$ $n$ times in parallel, but that is not true for all games. A parallel repetition theorem for the quantum value of general two-player games was only shown very recently \cite{OpenAI}, although the analogous classical theorem has been known for some time \cite{Raz95}.

We prove our noisy parallel repetition results for all of the noise models introduced except for biased reset noise. We first prove a weaker theorem that holds for any non-local game, and upper bounds the noisy quantum value of a parallel-repeated game $G^n$ in terms of the classical value of the parallel-repeated game. Then we prove a stronger result that works only for unique games, for which we define a quantity called a collision value $\col_p(G)$ (see Definition~\ref{def:col-value}) that depends on the winning predicate of the game as well as the input distribution. The upper bound for unique games is in terms of the classical value of $G^n$ as well as $\col_p(G)$. The two results are combined in the following theorem.
\begin{theorem}[Combination of Theorems~\ref{thm:parrep}, \ref{thm:parrep-gen}, \ref{thm:u-parrep}, \ref{thm:u-parrep-gen}]\label{thm:parrep-comb}
Let $G$ be a non-local game with input sets $\clX\times\clY$ and output sets $\clA\times\clB$, such that the classical value of $G^n$ is at most $w_G^n$. Then the quantum value of $G^n$ using noisy EPR pairs under depolarizing, unital, or erasure noise with maximal correlation $\rho$ is at most
\[ \left(2^{\rho\log(|\clA|\cdot|\clB|)}\cdot w_G\right)^{\frac{n}{1+\rho}}. \]

If $G$ is a unique game, then the noisy quantum value of $G^n$ is also upper bounded by
\[ O(n^2)\cdot\left((\col_{1+\rho}(G))^{4\rho}\cdot w_G^{1-\rho}\right)^{\frac{n}{1+3\rho}}. \]
\end{theorem}

We note that these results are different from the most standard form of parallel repetition theorems in that the base of the exponent depends on the single-copy \emph{classical} value of $G$ rather than its single-copy noisy quantum value. However, we think this dependence makes sense: at a high enough noise level, the noisy EPR pairs in fact become unentangled, and therefore the noisy quantum value should become the classical value. For example, for depolarizing noise, the states are unentangled if $\rho = \mu \leq \frac13$ (where $\mu$ is the sole parameter characterizing the depolarizing noise channel). None of the cases of Theorem~\ref{thm:parrep-comb} are sharp enough to give the classical bound $w_G^n$ exactly at low enough $\rho$, but the base of the exponent depending on $w_G$ and a noise parameter does seem correct.

However, the base of the exponent in the general-game case of Theorem~\ref{thm:parrep-comb} may be larger than $1$, in which case the result is trivial. This is difficult to check for arbitrary games, since we usually do not know the exact exponential rate in classical parallel repetition. But we do know the classical rate for the CHSH game \cite{DS14,Amb26}; substituting this value, $2^{\rho\log(|\clA|\cdot|\clB|)}\cdot w_\chsh$ is smaller than $1$ for some values of $\rho\in[0,1]$.

A greater test of the nontriviality of Theorem~\ref{thm:parrep-comb} is checking whether the upper bounds on the noisy quantum value of $G^n$ are smaller than known upper bounds on the noiseless quantum value of $G^n$. Moreover, this has to happen at some noise level small enough that the state is not already unentangled, since at such a noise level, the noisy quantum value should be equal to the classical value. Putting in the values of the parameters for CHSH in the general game case of Theorem~\ref{thm:parrep-comb}, the base of the exponent is also smaller than $\cos^2(\pi/8)$ at $\rho\approx 0.03$, which is far smaller than the threshold $\rho=\mu=\frac13$ at which depolarized EPR pairs become unentangled.

This is why we prove the unique games case of our result, which gives a much sharper bound. First, the base of the exponent for unique games is always smaller than $1$ ($\rho$ is in $[0,1]$ and so is $\col_{1+\rho}(G)$ --- see Definition~\ref{def:col-value}). Moreover, computing the value of $\col_{1+\rho}(G)$ for the CHSH game, we see that the base of the exponent for CHSH in the unique games upper bound is smaller than $\cos^2(\pi/8)$ as long as $\mu\leq 0.36$. Thus we have the following theorem for the CHSH game.
\begin{theorem}[Combination of Theorems~\ref{thm:chsh-parrep} and \ref{thm:chsh-parrep-gen}]\label{thm:chsh-parrep-comb}
For depolarizing, unital, and erasure noise, when the maximal correlation of EPR pairs under these noise models is $\rho$, the quantum value of $\chsh^n$ with noisy EPR pairs is at most $\alpha_\rho^n$, for some $\alpha_\rho$ that is less than $\cos^2(\pi/8)$ for $\rho\leq0.36$.
\end{theorem}
The above theorem only gives a parallel repetition rate smaller than the noiseless rate in the range $\rho\in (0.33,0.36]$, rather than the whole range $(0.33, 1]$ where EPR pairs remain entangled under depolarizing noise. Despite this limitation, we consider this a significant result, since getting the exact base of the exponent in parallel repetition is quite difficult.

\begin{remark}
Maximal correlation is not definitionally used in deriving Theorems~\ref{thm:parrep-comb} and \ref{thm:chsh-parrep-comb}. The relevant parameters here are actually exponents that appear in hypercontractivity theorems for all three noise models, and Theorems~\ref{thm:parrep}--\ref{thm:chsh-parrep-gen} are actually stated in terms of these. It can be shown that the exponents in these hypercontractivity inequalities must at least be the maximal correlation, but for these three models, it so happens that these exponents are actually equal to the maximal correlation; see Theorems~\ref{thm:depolarizing-hypercontractivity}, \ref{thm:unital-hypercontractivity} and \ref{thm:erasure-hypercontractivity}.
\end{remark}

\subsubsection{Communication complexity results}
We now turn to describing our results regarding the power of noisy entanglement in communication complexity. We present two separations with the perfect entanglement model, one concerning the communication complexity of a relational problem and another concerning the amount of noisy entanglement required to compute the Equality function with constant communication. 

\paragraph{Separation between communication with noisy vs noiseless entanglement.} Our first result for communication complexity is a separation between the quantum communication complexity models where Alice and Bob are allowed to share perfect EPR pairs, and where they are allowed to share noisy EPR pairs. We allow only classical communication in both models. Unlike in non-local games, sharing EPR pairs is sufficient (up to a small multiplicative factor) for communication, by a result of \cite{CH19}. Therefore, we consider this a sufficiently general model for quantum communication with noisy entanglement.

\begin{theorem}[Restated in Theorem~\ref{thm:Q*-Qnoise}]\label{thm:comm-sep}
Let $\Psi$ be the result of applying depolarizing, unital or erasure noise to a perfect EPR pair, and suppose $\MC(\Psi) < 0.36$. There exists a relational problem on $n$-bit inputs that can be solved with no communication if Alice and Bob share $n$ perfect EPR pairs, but requires $\Theta(n)$ two-way classical communication if Alice and Bob share arbitrarily many copies of $\Psi$.  
\end{theorem}
This theorem follows from Theorem~\ref{thm:chsh-parrep-comb} using an argument relating non-local games with some communication leakage to communication complexity, which was previously used in \cite{JK25,HLM25}. We prove a $t$-out-of-$n$ threshold parallel repetition theorem for the noisy quantum value of the CHSH game. The relation achieving the separation is the same as the $t$-out-of-$n$ CHSH game, whose input and output lengths are both $n$.

Theorem~\ref{thm:comm-sep} thus shows that the communication models with perfect and imperfect shared entanglement are not equivalent, at least if the base state $\Psi$ is sufficiently noisy. Moreover, it implies a lower bound on the interactive communication cost of distilling $n$ EPR pairs, given arbitrarily many noisy EPR pairs. This is because if it were possible to distill $n$ EPR pairs with $o(n)$ communication, then it would be possible to then use the zero-communication noiseless strategy from Theorem~\ref{thm:comm-sep} in order to solve the relation for a total of $o(n)$ communication. We thus have the following corollary of Theorem~\ref{thm:comm-sep}, along with known entanglement distillation protocols.
\begin{corollary}\label{cor:distill-lb}
Let $\Psi$ be as in the statement of Theorem~\ref{thm:comm-sep}. Suppose that $\Psi$ has nonzero distillable entanglement. If Alice and Bob share arbitrarily many copies of $\Psi$, the two-way classical communication cost of distilling $n$ near-perfect EPR pairs is $\Theta(n)$ bits.
\end{corollary}
As far as we know, this is the first time that the known entanglement distillation protocols are shown to be optimal in terms of two-way communication cost, for certain choices of the state $\Psi$. \cite{AY04} considered the problem of distilling one EPR pair out of a state close to $m$ EPR pairs, and showed a communication lower bound for this. But their result does not extend to distilling $n$ EPR pairs, since that would require a direct sum-like property. \cite{DY24} derived an $\Omega(n)$ lower bound on the one-way (quantum) communication complexity for extracting a shared $n$-bit random string out of arbitrarily many depolarized EPR pairs, and their result applies to any maximal correlation $\rho \in [0,1)$. Since sharing $n$ EPR pairs trivially lets Alice and Bob get $n$ shared random bits, this also implies an $\Omega(n)$ communication lower bound on distilling $n$ EPR pairs for the maximal-correlation range $\rho \in [\frac13,1)$. Our result is stronger than \cite{DY24} due to applying to two-way communication and working for more noise models.

\paragraph{Separation between communication with noisy entanglement and perfect shared randomness.} Theorem~\ref{thm:comm-sep} holds for sufficiently high noise, which results in low maximal correlation. For sufficiently low noise, there is also a straightforward separation between quantum communication with noisy shared entanglement, for all four noise models, and randomized communication with perfect shared randomness. We formally prove this in Theorems~\ref{thm:Rpub-Qt}--\ref{thm:Rpub-Quni} in Section~\ref{sec:comm-lb}, using classical parallel repetition for CHSH \cite{DS14,Amb26} and the lower bounds on noisy quantum values of CHSH from Table~\ref{tab:chsh-ub-lb}.

\paragraph{New upper and lower bounds on the noisy entanglement cost of protocols.}
We also revisit the results of \cite{CGMS15} that were mentioned in the introduction. We prove the following strengthened version of their simulation, which builds on an earlier simulation due to \cite{Shi_2008}. 
\begin{theorem}\label{thm:canonnebetter}
Fix an integer $d\geq1$, and let $R^\A,R^\B$ be dependent finite-valued random variables with maximal correlation $\rho>0$. Any two-way entanglement-assisted communication protocol with $n$-bit inputs that communicates $C$ bits and produces one bit can be approximated by a classical SMP protocol with communication $\rho^{-6d}2^{O(C+d)}$ that uses $\rho^{-6d}2^{O(C+d)}n^{1/d}$ independent copies of $(R^\A,R^\B)$. The hidden constants depend only on the distribution of $(R^\A,R^\B)$ and the desired constant approximation error, apart from the displayed dependence on $\rho$.
\end{theorem}
This strengthens Theorem~\ref{thm:canonne} in several directions: the simulation applies to entanglement-assisted protocols, produces an SMP protocol rather than a one-way protocol, and gives an explicit bound on the number of correlated samples. It is also stated directly for arbitrary finite-valued dependent sources and makes the dependence on maximal correlation explicit. Applying this simulation to constant-communication classical protocols for equality gives the following corollary.
\begin{corollary}
Fix dependent finite-valued random variables $R^\A,R^\B$ with maximal correlation $\rho>0$ and an integer $d\geq1$. There exists a family of classical SMP protocols for $\EQ_n$ with communication $\rho^{-6d}2^{O(d)}$ and imperfect-shared-randomness cost $\rho^{-6d}2^{O(d)}n^{1/d}$.
\end{corollary}
This shows that any number of copies that grows polynomially with $n$ suffices for a constant communication protocol, thereby improving on the protocol of \cite{CGMS15}, which turns out to consume $\Theta(n)$ copies. We also show that the polynomial dependence on $n$ is qualitatively optimal, even if the players share noisy entanglement, which is a stronger resource.
\begin{theorem} \label{thm:lowerboundnumberofcopies}
Let $\Psi$ be an entangled state with maximal correlation $\rho < 1$. Any family of two-way bounded-error protocols for $\EQ_n$ with $O(1)$ classical communication that use $N(n)$ independent copies of $\Psi$ and no additional shared randomness or entanglement must use $N(n) = n^{\Omega(1)}$ copies.
\end{theorem}
Specializing to the case in which $\Psi$ encodes imperfectly correlated uniform random bits gives a matching qualitative lower bound for imperfect shared randomness. We note that (for nontrivial parameters), all the noise models we consider in this work result in maximal correlation smaller than $1$ when applied to EPR pairs. Therefore, the above theorem gives a lower bound on the number of EPR pairs under depolarizing, unital, biased reset and erasure noise.

We also note that Theorem~\ref{thm:lowerboundnumberofcopies} answers the open question of \cite{CGMS15} regarding whether $O(\log n)$ copies of imperfectly correlated random bits suffice for any protocol using this resource, in the negative. In fact, polynomially many imperfectly correlated bits, or polynomially many copies of noisy shared entanglement are required, which is exponentially more than in the perfectly correlated case \cite{New91}. If $O(\log n)$ communication is allowed, $O(\log n)$ copies of noisy shared entanglement suffice, since $O(\log n)$ EPR pairs can be distilled out of them with $O(\log n)$ communication. But this no longer achieves the optimal communication for Equality, which is constant.

\subsection{Our techniques}
The main technical results of this work are the noisy parallel repetition theorems for general and unique games, and the lower and upper bounds on noisy entanglement cost in communication complexity. So we describe our techniques for proving these results.

\subsubsection{Noisy parallel repetition theorems}
We focus on the depolarizing-noise case first. Consider a strategy for $G^n$ with $m$ copies of EPR pairs under depolarizing noise, which we denote by $(\Omega^\dep_\mu)^{\otimes m}$. Here $\mu$ is the probability that the depolarizing noise does nothing to the EPR pair, which is also equal to its maximal correlation. The strategy involves measurements $M^\bx_\ba, N^\by_\bb$ for Alice and Bob respectively, where $\bx, \by$ are the $n$-copy inputs and $\ba, \bb$ are $n$-copy output strings. The conditional distribution of outputs given inputs is
\[ \sfP_{\bA\bB|\bx\by}(\ba,\bb) = \Tr\left[(M^\bx_\ba\otimes N^\by_\bb)(\Omega^\dep_\mu)^{\otimes m}\right].\]
The first technical component of our proof is a pointwise upper bound on these probabilities $\sfP_{\bA\bB|\bx\by}(\ba,\bb)$ in terms of the product of Alice and Bob's marginal probabilities. Specifically, using a hypercontractivity theorem for the depolarizing channel (Theorem~\ref{thm:depolarizing-hypercontractivity}, due to \cite{MO08,King14}), we prove that
\[
\sfP_{\bA\bB|\bx\by}(\ba,\bb)
\leq
\left(
\sfP_{\bA|\bx}(\ba)
\sfP_{\bB|\by}(\bb)
\right)^{1/(1+\mu)}.
\]
This result is quite similar to Lemma~15 of \cite{DB14}, which is used there to prove results about local state transformation. However, unlike \cite{DB14}, we obtain our bound in terms of the actual marginal probabilities of the strategy.

\paragraph{General games.} For any quantum strategy, the product of its marginal probabilities can be obtained by a classical strategy. In particular, the probability of winning $G^n$ under this product-of-marginals distribution (which we denote by $\sfQ$) is at most $w_G^n$. To prove our parallel repetition theorem for general games, we observe that the pointwise bound lemma can be used to upper bound the R\'enyi divergence $D_\alpha$ between $\sfP$ and $\sfQ$, for $\alpha=\frac{1+\mu}\mu$. It is well known that if $D_\infty$ between $\sfP$ and $\sfQ$ is bounded, an upper bound on the probability of an event under $\sfQ$ can be used to upper bound the probability of the same event under $\sfP$. As it turns out, a similar result holds for $D_\alpha$ for other $\alpha > 1$, as we show in Lemma~\ref{lem:event-transfer}. Using the upper bound $w_G^n$ on the probability of winning $G^n$ under $\sfQ$ then gives the first part of Theorem~\ref{thm:parrep-comb} for $\sfP$.

\paragraph{Unique games.} The R\'enyi divergence-based argument loses a factor depending exponentially on the output sizes of $G$. For unique games, we avoid most of this loss by exploiting the fact that each answer of one player determines at most one winning answer of the other player. We write $\pi^\B_{\bx\by,n}(\ba)$ to denote the single output string by Bob that is accepted on inputs $\bx,\by$ and Alice-output $\ba$. From the pointwise bound lemma, the winning probability of the unique game can be upper bounded by
\[ \sum_{\substack{\bx,\by,\ba,\bb: \\ V^n(\bx,\by,\ba,\bb)=1}}\left(
\sfP_{\bA|\bx}(\ba)
\sfP_{\bB|\by}(\bb)
\right)^{1/(1+\mu)} = \sum_{\bx,\by,\ba}\left(
\sfP_{\bA|\bx}(\ba)
\sfP_{\bB|\by}(\pi^\B_{\bx\by,n}(\ba))
\right)^{1/(1+\mu)}.\]

To upper bound this quantity, we partition Alice's and Bob's answers into buckets according to their marginal probabilities. In a single bucket for Alice, $\sfP_{\bA|\bx}(\ba)$ is between $2^{-k}$ and $2^{-k+1}$. For each pair of buckets, we bound the number of compatible winning answer pairs in two ways. The first, which we call the hard list bound, uses the classical parallel-repeated value $w_G^n$.

The second way uses the collision value $\col_p(G^n)$ of $G^n$. The $\A\to\B$ collision value of a unique game is the $p\to2$ induced norm of a function $T^{\A^n\to\B^n}_{G^n}$ depending on its input distribution and the function $\pi^\B_{\bx\by,n}$ taking Alice-outputs to Bob-outputs. Roughly speaking, $T^{\A^n\to \B^n}_{G^n}$ acting on an Alice input-output pair $(\bx,\ba)$ produces a mixture of Bob input-output pairs $(\by,\bb)$ such that $(\bx,\by,\ba,\bb)$ are accepted by $V^n$.\footnote{Technically, $T^{\A^n\to\B^n}_{G^n}$ acts on functions $\clX^n\times\clA^n\to\bbR$ to produce functions $\clY^n\times\clB^n\to\bbR$; the effect described can be thought of as its action on indicator functions for a particular $(\bx,\ba) \in \clX^n\times\clA^n$.} We say two input-output pairs $(\bx,\ba), (\bx',\ba')$ collide if $\pi^\B_{\bx\by,n}(\ba)=\pi^\B_{\bx'\by,n}(\ba')$ for some $\by$, and $\|T^{\A^n\to\B^n}_{G^n}z\|_2^2$ measures the probability of this. This is why we call the $p\to2$ norm of $T^{\A^n\to\B^n}_{G^n}$ the collision value. By the definition of $T^{\A^n\to\B^n}_{G^n}$, it is not difficult to see that the number of compatible winning answer pairs inside a bucket can be upper bounded by  $\|T^{\A^n\to\B^n}_{G^n}\|_{p\to2}$ (along with some additional factors depending on answer and bucket size). However, $\|T^{\A^n\to\B^n}_{G^n}\|_{p\to2}$ is just $\|(T^{\A\to\B}_G)^{\otimes n}\|_{p\to2}$. Induced norms tensorize, so the last quantity is just $(\|T^{\A\to\B}_G\|_{p\to2})^n$. $\|T^{\A\to\B}_G\|_{p\to2}$ is a relatively simple quantity depending on only the single-copy game, which can be computed.

Using the hard list bound and the collision list bound, the number of winning answers inside each pair of buckets is upper bounded as the minimum of the two bounds. Interpolating between the two and summing the resulting double geometric series gives the parallel repetition theorem for unique games. We note that dividing into buckets and taking the minimum of the two bounds at each step is important: taking a global minimum of the two bounds gives a worse overall upper bound, and we do not obtain a nontrivial result for CHSH from it.

\paragraph{The CHSH game.} For the CHSH game, the relevant collision operator is only a $4\times4$ matrix. Its $p$-to-$2$ norm reduces to a four-variable scalar optimization. We certify a sufficiently strong upper bound at $\mu=0.36$ and combine it with the classical rate $w_\chsh=(1+\sqrt5)/4$ to obtain $0.8534679$ numerically as the base of the exponent in parallel repetition for CHSH. This is slightly smaller than $\cos^2(\pi/8)\approx 0.8535534$.

\paragraph{Other noise models.} The proofs for unital and erasure noise are very similar. A hypercontractivity theorem analogous to Theorem~\ref{thm:depolarizing-hypercontractivity} was previously not known for general unital channels. We prove this result in Theorem~\ref{thm:unital-hypercontractivity} using a multiplicativity result for unital noise channels in \cite{King14}. For erasure noise, a hypercontractivity theorem (Theorem~\ref{thm:erasure-hypercontractivity}, \cite{BDOY25}) was previously known. We use this hypercontractivity property to get an analogous pointwise bound lemma for these noise models, and the rest of the proof is identical.

\subsubsection{Simulating entanglement with noisy shared randomness in communication}
Theorem~\ref{thm:canonnebetter} follows from an SMP protocol for estimating the real inner product of two vectors using noisy shared randomness. Lemma~\ref{lem:clevebuhrman}, based on \cite{Lal24}, expresses the acceptance probability of a $C$-bit entanglement-assisted protocol as $2^C$ times the inner product of two sub-unit vectors. These vectors may initially lie in a space whose dimension depends on the amount of shared entanglement. The Johnson--Lindenstrauss lemma reduces their dimension to $O(n2^{2C})$, for constant approximation error, independently of the amount of shared entanglement.

To estimate an inner product in dimension $D$, we encode the $D$ coordinates of the input vectors $u,v$ into degree-$d$ monomials in $m$ i.i.d. copies of the correlated random variables $R^\A,R^\B$, where $D=\binom md$. By the definition of maximal correlation, if $\MC(R^\A;R^\B)=\rho$, then there are functions $f\in L^2(R^\A)$ and $g\in L^2(R^\B)$ such that
\[
\bbE[f(R^\A)]=\bbE[g(R^\B)]=0,
\qquad
\bbE[f(R^\A)^2]=\bbE[g(R^\B)^2]=1,
\qquad
\bbE[f(R^\A)g(R^\B)]=\rho.
\]
Given i.i.d. copies $(R^\A_1,R^\B_1),\ldots,(R^\A_m,R^\B_m)$, define
\[
Z^\A=\sum_{S\in\binom{[m]}d}u_S\prod_{j\in S}f(R^\A_j),
\qquad
Z^\B=\sum_{S\in\binom{[m]}d}v_S\prod_{j\in S}g(R^\B_j).
\]
Distinct monomials are orthogonal because the witnesses are centred, while a matching pair of degree-$d$ monomials has correlation $\rho^d$. Consequently,
\[
\bbE[Z^\A Z^\B]=\rho^d\langle u,v\rangle.
\]

The variables $Z^\A,Z^\B$ need not be bounded, so Alice and Bob first clip them to $[-T,T]$. They then use private randomness to round their clipped values to signs $a,b\in\{-1,1\}$ whose conditional expectations are the clipped values divided by $T$. It follows that $T^2ab/\rho^d$ has expectation close to $\langle u,v\rangle$. A fourth-moment bound controls the bias introduced by clipping, and averaging $k$ independent repetitions controls the sampling error by Theorem~\ref{thm:chernoff}. Each repetition uses $m$ copies of $(R^\A,R^\B)$ and sends one bit from each player, giving a total of $km$ source copies and $2k$ bits of communication.

This construction is symmetric: Alice's and Bob's messages are computed independently, and the referee only needs the pairs of signs $(a_i,b_i)_{i=1}^k$. In contrast, the construction of \cite{CGMS15} has Alice choose an input-dependent correlated Gaussian direction and send its index to Bob; Bob's computation depends on that index, which is why their protocol is one-way rather than SMP.

\subsubsection{Lower bound on noisy entanglement cost of Equality}
Theorem~\ref{thm:lowerboundnumberofcopies} combines the standard packing argument for equality with the singular-value decomposition associated with maximal correlation. The standard packing argument involves the following idea: a protocol for $\EQ_n$ implies the existence of a vector $u_x$ for each input $x\in\{0,1\}^n$ to $\EQ_n$ such that $\langle u_x, u_y\rangle$ is small for $y\neq x$. The dimension of these vectors depends on the specific protocol for $\EQ_n$, and there is a standard packing result (stated in Lemma~\ref{lem:packingbound}) which lower bounds the dimension from the pairwise inner product bound. This property has been used to lower bound the communication cost of protocols for $\EQ_n$ previously.

To lower bound the noisy entanglement cost, suppose a constant-communication protocol for $\EQ_n$ uses $m$ copies of a fixed state $\Psi_{E^\A E^\B}$. The transcript decomposition from Lemma~\ref{lem:clevebuhrman}, followed by the maximal-correlation decomposition of $\Psi^{\otimes m}$, expresses its acceptance probability as
\[
\Pr[b=1|x,y]=\langle v_x,\Lambda w_y\rangle,
\]
where $\|v_x\|_2,\|w_y\|_2\leq2^{C/2}$ and $\Lambda$ is diagonal. Its diagonal entries are products of the singular values of the one-copy correlation map corresponding to $\Psi_{E^\A E^\B}$. The largest singular value is $1$, while every other singular value is at most $\rho=\MC(\Psi)<1$.

Correctness gives a constant gap between $\langle v_x,\Lambda w_x\rangle$ and $\langle v_x,\Lambda w_y\rangle$ for $x\neq y$. Coordinates of $\Lambda$ below a fixed threshold contribute too little to this gap and can be discarded. After projecting onto the remaining coordinates, the vectors corresponding to the $2^n$ possible inputs remain separated by a constant distance. Lemma~\ref{lem:packingbound} therefore implies that the retained subspace has dimension $\Omega(n)$.

A tensor-product singular value of $\Lambda$, involving $t$ coordinates corresponding to the non-maximal singular vector, is at most $\rho^t$. Hence only constant-degree tensor-product directions can lie above the threshold. The number of such directions among $m$ copies is polynomial in $m$. Comparing this polynomial upper bound with the $\Omega(n)$ lower bound from packing yields $m=n^{\Omega(1)}$.

\subsection{Open problems}
There are some obvious open questions left from our results.
\begin{itemize}
\item \textbf{Improving the base of the exponent in the parallel repetition results:} We suspect that the base of the exponent for CHSH for example, should be lower than $\cos^2(\pi/8)$ for \emph{any} noise level, and this certainly should be true for noise levels that lead to maximal correlation above $0.36$. Moreover, ideally the base of the exponent should go to the classical parallel repetition rate $w_G$ as soon as the noise level is high enough for the states to become unentangled, and we currently do not have this either.\footnote{For depolarizing noise, the states become unentangled at $\rho=\mu\leq \frac13$. It can be checked that the base of the exponent for the CHSH game $(\col_{1+\rho}(\chsh))^{4\rho}w_\chsh^{1-\rho}$ does not reduce to $w_\chsh$ for $\rho\leq \frac13$.}
\item \textbf{Making the parallel repetition theorem work for biased reset noise:} The chief barrier to making our parallel repetition proof technique work for biased reset noise is the lack of an appropriate hypercontractivity result for biased reset noise. A hypercontractivity result for this type of channel was proved very recently in \cite{DGOYZ26}, but unfortunately it is not quite the type of hypercontractivity result that is useful for us. There may be an extension of their result that is of the correct form to work with our techniques.
\item \textbf{Extending to more noise models:} An important type of noise model that our techniques do not work for at all is amplitude damping noise. This is because all our techniques rely on the maximal correlation of EPR pairs under the relevant noise model being less than $1$, but for amplitude damping noise, this is actually $1$. This is because maximal correlation does not only measure quantum correlation, but also classical correlation; for example, the maximal correlation of Werner states continues to be nonzero even when the states are unentangled. Thus, a characterization of noisy quantum values of games solely in terms of maximal correlation cannot be sufficient. Perhaps a measure such as the \emph{maximal entanglement} introduced by Beigi \cite{Bei14} could be useful here.
\item \textbf{Stronger communication separations:} Generally, separations for Boolean functions are considered the strongest class of separations for any model, and are the hardest to prove separations for. Our communication problem separating quantum communication with noisy vs noiseless entanglement is a relational one, and moreover the relation has a large output alphabet. It would be interesting to see if our result can be extended to relations with smaller alphabets, and to partial Boolean functions.
\end{itemize}

\subsection{Organization of the paper}
In Section~\ref{sec:prelim}, we introduce the mathematical preliminaries used throughout the paper. In Section~\ref{sec:noisyentanglement}, we introduce our noise models and maximal correlation, and compute the maximal correlation of noisy EPR pairs. In Section~\ref{sec:game-comm}, we recall the standard definitions and results concerning non-local games and communication complexity and introduce their noisy-entanglement variants. In Section~\ref{sec:game-ub}, we derive upper bounds on the noisy value of the single-copy CHSH game. In Section~\ref{sec:parrep}, we prove our noisy parallel repetition theorems for general and unique games, with special attention to CHSH. In Section~\ref{sec:comm-lb}, we use these results to separate communication with noisy entanglement from communication with perfect shared randomness, and communication with perfect entanglement from communication with noisy entanglement. Finally, in Section~\ref{sec:ent-lb}, we prove our upper and lower bounds on the amount of noisy shared randomness or noisy entanglement needed by communication protocols.

\section{Preliminaries}\label{sec:prelim}
\subsection{Linear algebra notions}

Given a vector space $V$, a linear map $f: V \to \bbC$ is called a linear functional on $V$. We will use the following properties of linear functionals.
\begin{theorem}[Riesz representation theorem] \label{thm:riesz}
Let $V$ be a finite-dimensional inner product space. Given a linear functional $f$ on $V$, there exists a unique $u \in V$ such that, for all $v \in V$, 
\[f(v) = \langle u, v \rangle\]
\end{theorem}
\begin{theorem}[Singular value decomposition] \label{thm:svd}
Let $V, W$ be finite-dimensional inner product spaces, and let $T: V \to W$ be a linear map. Let $m = \min(\dim V, \dim W)$. There exist orthonormal collections $v_1, ...,v_m \in V$, $w_1, ...,w_m \in W$ such that, for all $i,j \in [m]$ with $i \neq j$, $\langle w_i, T(v_j) \rangle = 0$ and real numbers $\mu_1 \geq \mu_2 \geq ... \geq \mu_m \geq 0$ such that, for all $i,j \in [m]$,
\[\langle w_i, T(v_j) \rangle = \delta_{i,j} \mu_i\]
The $v_i$ and the $w_i$ are called the right and left singular vectors of $T$, respectively. The numbers $\mu_i$ are called the singular values of $T$ and are unique. 
\end{theorem}

\paragraph{Inner products.} For $n\times n$ matrices $X,Y$, $\langle X,Y\rangle=\Tr(X^\dagger Y)$ denotes the usual Hilbert--Schmidt inner product. Given an $n\times n$ density matrix $\sigma$, we write $\langle X,Y\rangle_\sigma=\Tr(\sigma X^\dagger Y)$ for the corresponding weighted inner product. We write $L^2(\sigma)$ for the resulting finite-dimensional inner product space. If $\sigma$ is not full rank, operators of zero weighted norm are identified; equivalently, only their action on the support of $\sigma$ is retained.

We have the following standard consequence of the Johnson--Lindenstrauss lemma for approximating a finite collection of vectors in a lower-dimensional space.
\begin{lemma}[\cite{John84}]
\label{lem:jl-inner-products}
Let $v_1,\dots,v_N \in \bbR^{D_0}$ be sub-unit vectors, and let $0<\eps<1$. There is an integer
\[
D_1=O\!\left(\frac{\log(2N)}{\eps^2}\right)
\]
and sub-unit vectors $w_1,\dots,w_N\in\bbR^{D_1}$ such that
\[
  \bigl|\langle w_i, w_j\rangle - \langle v_i,v_j\rangle\bigr| \;\le\; \eps
  \qquad\text{for all } i,j \in [N].
\]
\end{lemma}
The usual Johnson--Lindenstrauss lemma preserves all pairwise squared distances in a set of $O(N)$ points. Applying it to the vectors together with the origin and using the polarization identity gives the displayed inner-product estimate. The resulting vectors have norm at most $\sqrt{1+O(\eps)}$; a common rescaling makes them sub-unit while changing the approximation error by only a constant factor, which is absorbed into the $O$ notation.

\paragraph{Vector $\ell_p$ norm.} For $p\geq 1$, the $\ell_p$ norm of a vector $v=(v_i)_{i\in I}$ on a finite set $I$ is defined as
\[
 \|v\|_{p}
 =\left(\sum_{i\in I}|v_i|^p\right)^{1/p}.
 \]
If $\sfP$ is a probability distribution on $I$, we also use the weighted version of the $\ell_p$ norm:
\[
 \|v\|_{p(\sfP)}
 =\left(\sum_{i\in I}\sfP(i)|v_i|^p\right)^{1/p}.
 \]

We have the following standard estimate:
\begin{lemma}[{\cite[Lemma~4.2.8]{Ver26}}]\label{lem:packingbound}
    Suppose $v_1, ..., v_k$ are sub-unit vectors in $\bbC^d$ such that, for all $l \neq l'$, 
    \[\|v_l - v_{l'}\|_2 \geq \eps\]
    Then:
    \[k \leq \left(1+\frac{2}{\eps}\right)^{2d}\]
\end{lemma}

\paragraph{Operator Schatten norms.} We use the same notation as the vector $\ell_p$ norm to denote the Schatten $p$-norm of a $d$-dimensional operator $X$:
\[
\|X\|_{p}
=
\left(
\Tr|X|^p
\right)^{1/p}.
\]
A special case is $p=\infty$; $\|X\|_\infty$ is equal to the maximum singular value of $X$.

These norms satisfy H\"older's inequality:
\[
|\Tr(XY)|
\le
\|X\|_{p}\|Y\|_{q},
\qquad
\frac1p+\frac1q=1.
\]
The normalized version of the Schatten norm will be denoted as
\[ \|X\|_{p}^\norm = \left(
\frac1d\Tr|X|^p
\right)^{1/p}.
\]
We also have the normalized H\"older's inequality for H\"older conjugate $p, q$:
\[
\frac1d|\Tr(XY)|
\le
\|X\|_{p}^\norm\|Y\|_{q}^\norm.
\]

\paragraph{Induced operator norms.} The vector $\ell_p$ norms can be used to define induced operator norms. For a linear map $A$ between two vector spaces on which we have an $\ell_p$ norm and an $\ell_q$ norm, the $p\to q$ induced norm is defined as
\[ \|A\|_{p\to q} = \sup_{v\neq 0}\frac{\|Av\|_q}{\|v\|_p}.\]
It can be verified that $\|A\|_{2\to2}=\|A\|_\infty$, i.e., the largest singular value of $A$.

If the two vector spaces have probability distributions $\sfP$ and $\sfQ$, we can define the weighted version of this norm
\[ \|A\|_{p(\sfP)\to q(\sfQ)} = \sup_{v\neq 0}\frac{\|Av\|_{q(\sfQ)}}{\|v\|_{p(\sfP)}}.\]
It is known that induced norms tensorize for $0 < p \leq q < \infty$.
\[ \|A\otimes B\|_{p \to q} = \|A\|_{p\to q}\|B\|_{p\to q}.\]
We provide a proof of the fact that this is also true of the weighted version.
\begin{lemma}\label{lem:tensor}
For $0 < p \leq q < \infty$, and any operator $A: V \to W$, and probability distributions $\sfP$ and $\sfQ$ on $V$ and $W$,
\[ \|A^{\otimes n}\|_{p(\sfP^{\otimes n})\to q(\sfQ^{\otimes n})} = \left(\|A\|_{p(\sfP)\to q(\sfQ)}\right)^n.\]
\end{lemma}
\begin{proof}
The direction
\[ \|A^{\otimes n}\|_{p(\sfP^{\otimes n})\to q(\sfQ^{\otimes n})} \geq \left(\|A\|_{p(\sfP)\to q(\sfQ)}\right)^n\]
follows by using $v^{\otimes n}$, where $v$ is the vector that maximizes $\|A\|_{p(\sfP)\to q(\sfQ)}$, as the argument of $A^{\otimes n}$.

We prove the other direction for $n=2$; the generalization is obvious. For an arbitrary vector $v\in V^{\otimes 2}$, let $u\in V\otimes W$ be $(I\otimes A)v$, and $w\in W^{\otimes 2}$ be $(A\otimes A)v = (A\otimes I)u$. We can write $v, u, w$ with two coordinates corresponding to the two vector spaces in the tensor product; we use $A_{ii'}$ to refer to the entries of $A$. Then
\[ u_{i_1j_2} = \sum_{i_2} A_{j_2i_2}v_{i_1i_2} \qquad w_{j_1j_2} = \sum_{i_1}A_{j_1i_1}u_{i_1j_2}.\]
For a fixed $j_2$, $w^{j_2} = (w_{j_1j_2})_{j_1}$ is a vector in $W$, and $u^{j_2}$ defined similarly is a vector in $V$. Therefore, by the definition of $\|A\|_{p(\sfP)\to q(\sfQ)}$,
\begin{align*}
\left(\sum_{j_1}\sfQ(j_1)|w_{j_1j_2}|^q\right)^{1/q} & = \|w^{j_2}\|_{q(\sfQ)} \\
& \leq \|A\|_{p(\sfP)\to q(\sfQ)}\|u^{j_2}\|_{p(\sfP)} \\
 & = \|A\|_{p(\sfP)\to q(\sfQ)}\left(\sum_{i_1}\sfP(i_1)|u_{i_1j_2}|^p\right)^{1/p}.
\end{align*}
Now,
\begin{align}
\|A^{\otimes 2}v\|_{q(\sfQ^{\otimes 2})} & = \|w\|_{q(\sfQ^{\otimes 2})} \nonumber \\
 & = \left(\sum_{j_2}\sfQ(j_2)\left(\sum_{j_1}\sfQ(j_1)|w_{j_1j_2}|^q\right)\right)^{1/q} \nonumber \\
 & \leq \|A\|_{p(\sfP)\to q(\sfQ)}\left(\sum_{j_2}\sfQ(j_2)\left(\sum_{i_1}\sfP(i_1)|u_{i_1j_2}|^p\right)^{q/p}\right)^{1/q}. \label{eq:exchange1}
\end{align}
We will use a variant of Minkowski's inequality \cite{GS16} which says, for probability distributions $\alpha_i, \beta_j$, $0 < p \leq q < \infty$, and any $z_{ij}$,
\[ \left(\sum_j\beta_j\left(\sum_i\alpha_i|z_{ij}|^p\right)^{q/p}\right)^{1/q} \leq \left(\sum_i\alpha_i\left(\sum_j\beta_j|z_{ij}|^q\right)^{p/q}\right)^{1/p}.\]
Using this on \eqref{eq:exchange1} we get
\begin{equation}\label{eq:exchange2}
\|A^{\otimes 2}v\|_{q(\sfQ^{\otimes 2})} \leq \|A\|_{p(\sfP)\to q(\sfQ)}\left(\sum_{i_1}\sfP(i_1)\left(\sum_{j_2}\sfQ(j_2)|u_{i_1j_2}|^q\right)^{p/q}\right)^{1/p}.
\end{equation}

Define $u^{i_1}$ and $v^{i_1}$ similarly to $w^{j_2}$ and $u^{j_2}$. By the definition of $\|A\|_{p(\sfP)\to q(\sfQ)}$ again we have,
\[ \left(\sum_{j_2}\sfQ(j_2)|u_{i_1j_2}|^q\right)^{1/q} \leq \|A\|_{p(\sfP)\to q(\sfQ)}\left(\sum_{i_2}\sfP(i_2)|v_{i_1i_2}|^p\right)^{1/p}.\]
Using this on \eqref{eq:exchange2} we get
\begin{align*}
\|A^{\otimes 2}v\|_{q(\sfQ^{\otimes 2})} & \leq \|A\|_{p(\sfP)\to q(\sfQ)}\left(\sum_{i_1}\sfP(i_1)\left(\|A\|_{p(\sfP)\to q(\sfQ)}\left(\sum_{i_2}\sfP(i_2)|v_{i_1i_2}|^p\right)^{1/p}\right)^p\right)^{1/p} \\
& = \left(\|A\|_{p(\sfP)\to q(\sfQ)}\right)^2\left(\sum_{i_1}\sfP(i_1)\sum_{i_2}\sfP(i_2)|v_{i_1i_2}|^p\right)^{1/p} \\
& = \left(\|A\|_{p(\sfP)\to q(\sfQ)}\right)^2\|v\|_{p(\sfP^{\otimes 2})}.
\end{align*}
This proves the lemma.
\end{proof}

\subsection{Probability distributions}
For a classical probability distribution $\sfP_{AX}$ on variables $AX$, we will use $\sfP_{A|X=x}$ to denote the conditional distribution of $A$ given $X=x$, which we will sometimes shorten to $\sfP_{A|x}$. We will omit the variable subscripts when the variable names are irrelevant or clear from context. We will use $\sfP(E)$ to denote the probability of an event $E$.

For two classical probability distributions $\sfP$ and $\sfQ$, and $\alpha>1$, the classical order-$\alpha$ R\'enyi divergence is 
\[
D_\alpha(\sfP\|\sfQ)
=
\frac1{\alpha-1}
\log
\sum_\omega
\sfP(\omega)^\alpha \sfQ(\omega)^{1-\alpha}.
\]

The following lemma relates the probability of an event under two probability distributions, using the Renyi divergence. It is elementary, but we have not found a proof elsewhere, so we include it.
\begin{lemma}
\label{lem:event-transfer}
Let $\sfP$ and $\sfQ$ be probability distributions, let $\alpha>1$, and let
$E$ be an event. Then
\[
\sfP(E)
\le
2^{
\frac{\alpha-1}{\alpha}D_\alpha(\sfP\|\sfQ)
}
\cdot\sfQ(E)^{(\alpha-1)/\alpha}.
\]
\end{lemma}

\begin{proof}
Write
\[
\sfP(E)
=
\sum_{\omega\in E}
\left(
\sfP(\omega)\sfQ(\omega)^{(1-\alpha)/\alpha}
\right)
\sfQ(\omega)^{(\alpha-1)/\alpha}.
\]

Applying H\"older's inequality (for vector inner products) with conjugate exponents $\alpha$ and
$\alpha/(\alpha-1)$ gives
\begin{align*}
\sfP(E) & 
\le
\left(
\sum_{\omega\in E}
\left(
\sfP(\omega)\sfQ(\omega)^{(1-\alpha)/\alpha}
\right)^\alpha
\right)^{1/\alpha}
\left(\sum_{\omega\in E}\sfQ(\omega)\right)^{(\alpha-1)/\alpha} \\
& \le \left(
\sum_{\omega\in E}
\sfP(\omega)^\alpha\sfQ(\omega)^{1-\alpha}
\right)^{1/\alpha}
\cdot\sfQ(E)^{(\alpha-1)/\alpha} \\
& \le \left(
\sum_{\omega}
\sfP(\omega)^\alpha\sfQ(\omega)^{1-\alpha}
\right)^{1/\alpha}
\cdot\sfQ(E)^{(\alpha-1)/\alpha} \\
& \le 2^{
\frac{\alpha-1}{\alpha}D_\alpha(\sfP\|\sfQ)
}
\cdot\sfQ(E)^{(\alpha-1)/\alpha}. \qedhere
\end{align*}
\end{proof}

We will need the standard Chernoff--Hoeffding bound and a generalized Chernoff bound for certain dependent Boolean random variables.
\begin{theorem}[Chernoff--Hoeffding bound]\label{thm:chernoff}
Let $X_1,\ldots,X_n$ be independent real random variables with $X_i\in[a_i,b_i]$ almost surely, and set $S=\sum_{i=1}^nX_i$. For every $t\geq0$,
\[
\Pr[S-\bbE[S]\geq t],\ \Pr[S-\bbE[S]\leq-t]
\leq
\exp\!\left(-\frac{2t^2}{\sum_{i=1}^n(b_i-a_i)^2}\right).
\]
Consequently,
\[
\Pr[|S-\bbE[S]|\geq t]
\leq
2\exp\!\left(-\frac{2t^2}{\sum_{i=1}^n(b_i-a_i)^2}\right).
\]
\end{theorem}

\begin{theorem}[Generalized Chernoff bound \cite{PS97}]\label{thm:gen-chernoff}
Let $X_1, \ldots, X_n$ be Boolean random variables such that, for some $\delta \in [0,1]$ and every subset $S \subseteq [n]$,
\[
\Pr\left[\bigwedge_{i\in S}\{X_i=1\}\right] \leq \delta^{|S|}.
\]
Then, for any $\gamma \in [\delta, 1]$, we have
\[\Pr\left[\sum_{i=1}^nX_i \geq \gamma n\right] \leq e^{-2(\gamma-\delta)^2n}.\]
\end{theorem}

We will also use a finite-space fourth-moment bound that is a homogeneous degree-$d$ consequence of O'Donnell's hypercontractivity theorem for general random variables, stated below.
\begin{theorem}[{\cite[Theorem~10.16]{ODo14}}, specialized to $q=4$]\label{thm:odonnell}
Let $X$ be a random variable satisfying $\bbE[X]=0$, $\bbE[X^2]=1$, and $\bbE[X^4]^{1/4}=B_X$, and let $\eta_X=\frac{1}{2\sqrt{3}\,B_X}$. Then
\[ \bbE\left[(a+\eta_XbX)^4\right]^{1/4} \leq \bbE\left[(a+bX)^2\right]^{1/2}
=
(a^2+b^2)^{1/2}\]
for all $a,b\in\bbR$.
\end{theorem}
Using this, we prove the following corollary, which we will actually use.
\begin{corollary}[Polynomial fourth-moment bound]\label{cor:finite-space-hypercontractivity}
Let $X$ be a finite-valued real random variable satisfying $\bbE[X]=0$ and $\bbE[X^2]=1$.
For $m\geq d$, let $X_1,\ldots,X_m$ be independent copies of $X$. Then, for every collection of real coefficients $(c_S)_{S\in\binom{[m]}d}$,
\[
Z=\sum_{S\in\binom{[m]}d}c_S\prod_{i\in S}X_i
\]
satisfies
\[
\bbE[Z^4]
\leq
\left(2\sqrt{3}\right)^{4d}\bbE[X^4]^d\;\bbE[Z^2]^2.
\]
\end{corollary}
\begin{proof}
Since the function $t\mapsto t^2$ is convex on $[0,\infty)$, Jensen's inequality applied to the nonnegative random variable $X^2$ gives
\[
B_X^4
=
\bbE[(X^2)^2]
\geq
\bbE[X^2]^2
=1,
\]
and hence $B_X\geq1$ and $\eta_X<1$.

We prove by induction on $m$ that every homogeneous multilinear polynomial $Z$ of degree $d$ in $X_1,\ldots,X_m$ satisfies
\begin{equation}\label{eq:induction}
\eta_X^d\bbE[Z^4]^{1/4}
\leq
\bbE[Z^2]^{1/2}.
\end{equation}
The case $m=0$ is immediate. The claim is also immediate when $d=0$, since $Z$ is then constant.

Suppose $m\geq1$ and $d\geq1$. Separate the monomials according to whether they contain $X_m$, and write
\[
Z=Z_0+X_mZ_1,
\]
where $Z_0$ is homogeneous of degree $d$ and $Z_1$ is homogeneous of degree $(d-1)$ in $X_1,\ldots,X_{m-1}$. Either polynomial may be zero. Since
\[
\eta_X^dZ
=
\eta_X^dZ_0
+
\eta_X\bigl(\eta_X^{d-1}Z_1\bigr)X_m,
\]
Theorem~\ref{thm:odonnell}, applied with $a=\eta_X^dZ_0$ and $b=\eta_X^{d-1}Z_1$ for each fixed value of $(X_1,\ldots,X_{m-1})$, gives
\[
\bbE_{X_m}\left[(\eta_X^dZ)^4\right]^{1/2}
\leq
\eta_X^{2d}Z_0^2+\eta_X^{2(d-1)}Z_1^2.
\]
Squaring this inequality, averaging over $X_1,\ldots,X_{m-1}$, and then taking a square root yields
\[
\eta_X^{2d}\bbE_{X_1,\ldots, X_{m-1}}[Z^4]^{1/2}
\leq
\bbE_{X_1,\ldots,X_{m-1}}\left[
\left(
\eta_X^{2d}Z_0^2+\eta_X^{2(d-1)}Z_1^2
\right)^2
\right]^{1/2} \leq \eta_X^{2d}\bbE[Z_0^4]^{1/2}
+
\eta_X^{2(d-1)}\bbE[Z_1^4]^{1/2}
\]
by the triangle inequality.
Applying the induction hypothesis to $Z_0$ and $Z_1$, the right-hand side above this is at most
\[
\bbE[Z_0^2]+\bbE[Z_1^2].
\]

Finally, $Z_0$ and $Z_1$ do not depend on $X_m$. Since $\bbE[X_m]=0$ and $\bbE[X_m^2]=1$,
\[
\bbE[Z^2]
=
\bbE[(Z_0+X_mZ_1)^2]
=
\bbE[Z_0^2]+\bbE[Z_1^2].
\]
Consequently,
\[
\eta_X^{2d}\bbE[Z^4]^{1/2}
\leq
\bbE[Z^2].
\]
Taking a square root proves the induction step and therefore \eqref{eq:induction}. Raising \eqref{eq:induction} to the fourth power and substituting the definition of $\eta_X$ gives
\[
\bbE[Z^4]
\leq
\eta_X^{-4d}\bbE[Z^2]^2
=
\left(2\sqrt{3}\right)^{4d}\bbE[X^4]^d\bbE[Z^2]^2,
\]
as required.
\end{proof}

\section{Noisy entanglement} \label{sec:noisyentanglement}
We now turn to introducing concepts about noisy entanglement that we will be using throughout the rest of the paper. The first subsection discusses the maximal correlation, which we will be using as a measure of how noisy a given quantum state is. The next subsections introduce the four noise models we will be studying.

\subsection{Maximal correlation}
In our study of noisy entanglement, it will be useful to have a way to measure how noisy a given entangled state is. We will be using the maximal correlation, which is a real number between zero and one. Classically, the maximal correlation of two random variables is defined as follows.

Throughout this paper, classical random variables are finite-valued and quantum systems are finite-dimensional. Accordingly, the maxima in the following definitions are attained.

\begin{definition}[\cite{Hir35,Geb41,Ren59}] \label{def:maximalcorrelation1}
Given two random variables $R^\A,R^\B$, their maximal correlation is
\[
\MC(R^\A;R^\B)
=
\max_{\substack{f \in L^2(R^\A),g \in L^2(R^\B) \\ \bbE[f(R^\A)] = \bbE[g(R^\B)] = 0 \\ \bbE[f(R^\A)^2] = \bbE[g(R^\B)^2] = 1 }} \bbE[f(R^\A)g(R^\B)].
\]
\end{definition}

This quantity extends to bipartite quantum states in the following way. 
\begin{definition}[\cite{Bei13}] \label{def:quantmaxcor}
Given a bipartite state $\Psi_{E^\A E^\B}$, its maximal correlation is
\[
\MC(\Psi)
=
\max_{\substack{M \in L^2(\Psi_{E^\A}), N \in L^2(\Psi_{E^\B}) \\ \langle I, M \rangle_{\Psi_{E^\A}} = \langle I, N \rangle_{\Psi_{E^\B}} = 0 \\ \langle M, M \rangle_{\Psi_{E^\A}} = \langle N, N \rangle_{\Psi_{E^\B}} = 1 }}
\left|\Tr\!\left(\Psi(M^\dagger\otimes N)\right)\right|.
\]
\end{definition}

If the marginals of $\Psi_{E^\A E^\B}$ are maximally mixed, then the centring conditions simply become $\Tr(M)=\Tr(N)=0$. If $M$ and $N$ satisfy the centring conditions but not necessarily the normalization conditions, the definition gives
\begin{equation}\label{eq:MC}
\left|
\Tr\!\left(\Psi(M^\dagger\otimes N)\right)
\right|
\leq
\MC(\Psi)
\sqrt{\Tr(\Psi_{E^\A} M^\dagger M)}
\sqrt{\Tr(\Psi_{E^\B} N^\dagger N)}.
\end{equation}
For arbitrary $M$ and $N$, subtracting their local means gives the covariance form
\begin{equation}\label{eq:MC-var}
\bigl|
\Tr\bigl(\Psi_{E^{\A}E^{\B}}M^\dagger\otimes N\bigr)-\Tr\bigl(\Psi_{E^{\A}}M^\dagger\bigr)
 \Tr\bigl(\Psi_{E^{\B}}N\bigr)
\bigr|
\leq
\MC(\Psi)
\sqrt{\Var_{\Psi_{E^{\A}}}(M)}
 \sqrt{\Var_{\Psi_{E^{\B}}}(N)},
\end{equation}
where
\[
\Var_\sigma(M)
=
\Tr(\sigma M^\dagger M)
-
\left|\Tr(\sigma M)\right|^2.
\]

We will use the following properties of the quantum maximal correlation, which were shown by Beigi.
\begin{lemma}[\cite{Bei13}] \label{lem:maxcorrproperties}
    The maximal correlation satisfies the following properties:
    \begin{enumerate}
        \item For any state $\Psi_{E^\A E^\B}$, $\MC(\Psi) \in [0,1]$, and $\MC(\Psi) = 0$ if and only if $\Psi$ is a product state.
        \item Given states $\Psi_{E^\A_1 E^\B_1}, \Phi_{E^\A_2 E^\B_2}$, we have
        \[\MC(\Psi \otimes \Phi) = \max(\MC(\Psi), \MC(\Phi))\]
        \item The maximal correlation is non-increasing under local operations.
    \end{enumerate}
\end{lemma}

The quantum maximal correlation can also be described through a singular-value decomposition. Let $\Psi_{E^\A E^\B}$ be a density matrix. By Theorem~\ref{thm:riesz}, there is a unique linear map
\[
\mathcal T_\Psi:L^2(\Psi_{E^\A})\longrightarrow L^2(\Psi_{E^\B})
\]
satisfying
\[
\Tr\!\left(\Psi(M^\dagger\otimes N)\right)
=
\langle \mathcal T_\Psi(M),N\rangle_{\Psi_{E^\B}}
\]
for every $M\in L^2(\Psi_{E^\A})$ and $N\in L^2(\Psi_{E^\B})$. The Cauchy--Schwarz inequality implies that every singular value of $\mathcal T_\Psi$ lies in $[0,1]$, while $\mathcal T_\Psi(I)=I$. Applying Theorem~\ref{thm:svd} to $\mathcal T_\Psi$ gives the following decomposition, which was also proved by Beigi.
\begin{theorem}[Maximal correlation decomposition; {\cite[Theorem~1]{Bei13}}] \label{thm:efronstein}
With the above setup, there are orthonormal families $M_1=I,M_2,\ldots$ in $L^2(\Psi_{E^\A})$ and $N_1=I,N_2,\ldots$ in $L^2(\Psi_{E^\B})$, with at least one family forming a basis of its space, and nonnegative real numbers $1=\rho_1\geq\rho_2\geq\cdots$ such that, for all $i,j$,
\[
\Tr\!\left(\Psi(M_i^\dagger\otimes N_j)\right)
=
\delta_{i,j}\rho_i.
\]
Moreover, $\rho_2=\MC(\Psi)$.
\end{theorem}

\subsection{Depolarizing noise}
Our default model of noisy entanglement is an EPR pair affected by depolarizing noise. The resulting states are the two-qubit Werner states. For $\mu\in[0,1]$, define
\[
\Omega^\dep_\mu
=
\mu\ket{\Phi^+}\!\bra{\Phi^+}
+
(1-\mu)\frac{I_4}{4},
\]
where
\[
\ket{\Phi^+}
=
\frac{\ket{00}+\ket{11}}{\sqrt{2}}.
\]
The parameter $\mu$ is called the Werner visibility. The Werner state is separable for $\mu\leq\frac13$ and entangled for $\mu>\frac13$.

The qubit depolarizing channel with visibility $\mu$ is
\[
\Delta^\dep_\mu(X)
=
\mu X+(1-\mu)\frac{\Tr(X)}{2}I_2.
\]
It gives the Werner state when applied to one half of an EPR pair:
\[
\Omega^\dep_\mu
=
\bigl(\operatorname{id}\otimes\Delta^\dep_\mu\bigr)
\bigl(\ket{\Phi^+}\!\bra{\Phi^+}\bigr).
\]
\begin{theorem}[\cite{Wern89}]
The qubit Werner states are entangled for $\mu>\frac13$ and separable for $\mu\leq\frac13$.
\end{theorem}

More generally,
\[
(\Omega^\dep_\mu)^{\otimes m}
=
\bigl(\operatorname{id}\otimes(\Delta^\dep_\mu)^{\otimes m}\bigr)
\bigl(\ket{\Phi^+_{2^m}}\!\bra{\Phi^+_{2^m}}\bigr),
\qquad
\ket{\Phi^+_d}
=
\frac{1}{\sqrt d}\sum_{j=1}^d\ket{j}\ket{j}.
\]

We will use the elementary maximally entangled state identity
\[
\bra{\Phi^+_d}M\otimes N\ket{\Phi^+_d}
=
\frac1d\Tr(MN^{\mathsf T})
=
\frac1d\Tr(M^{\mathsf T}N).
\]
The depolarizing channel is self-adjoint with respect to the Hilbert--Schmidt inner product and commutes with transposition. Consequently, for arbitrary operators $M$ and $N$ on the two local $m$-qubit spaces,
\begin{equation}
\label{eq:choi-identity}
\Tr\!\left[
(M\otimes N)(\Omega^\dep_\mu)^{\otimes m}
\right]
=
\frac1{2^m}
\Tr\!\left[
M\,(\Delta^\dep_\mu)^{\otimes m}(N^{\mathsf T})
\right].
\end{equation}

We will use the following tensorized hypercontractive inequality.
\begin{theorem}[Depolarizing hypercontractivity \cite{MO08,King14}]
\label{thm:depolarizing-hypercontractivity}
Let $1<p\leq q<\infty$ and $0\leq\mu\leq1$. Then for every $m\geq1$,
\[
\left\|(\Delta^\dep_\mu)^{\otimes m}(X)\right\|^\norm_q
\leq
\|X\|^\norm_p
\]
for every operator $X$ if and only if
\[
\mu
\leq
\sqrt{\frac{p-1}{q-1}}.
\]
\end{theorem}

The maximal correlation of Werner states was computed in \cite{QY21}.
\begin{lemma}[{\cite[Lemma~3.9]{QY21}}]
\label{lem:MC-werner}
The maximal correlation of the Werner state $\Omega^\dep_\mu$ is
\[
\MC(\Omega^\dep_\mu)
=
\mu.
\]
\end{lemma}

\subsection{General unital noise}

A qubit channel is unital if it maps the identity to itself. We denote a unital qubit channel by $\Delta^\uni_K$, where $K\in\bbR^{3\times3}$ is its Bloch matrix, defined by
\[
\Delta^\uni_K(\sigma_j)
=
\sum_{i=1}^3
K_{ij}\sigma_i,
\qquad
j\in\{1,2,3\},
\]
where $\sigma_1,\sigma_2,\sigma_3$ are the Pauli operators. Since $\Delta^\uni_K$ is unital and trace preserving,
\[
\Delta^\uni_K(I_2)=I_2
\]
and its action on an arbitrary operator $X=aI_2+\sum_{i=1}^3v_i\sigma_i$ is
\[
\Delta^\uni_K(X)
=
aI_2+
\sum_{i=1}^3(Kv)_i\sigma_i.
\]
Thus, the Bloch matrix $K$ completely specifies the channel. Applying $\Delta^\uni_K$ to one half of an EPR pair produces the state
\[
\Omega^\uni_K
=
\bigl(\operatorname{id}\otimes\Delta^\uni_K\bigr)
\bigl(\ket{\Phi^+}\!\bra{\Phi^+}\bigr),
\]
which has maximally mixed marginals. Depolarizing noise is the special case $K=\mu I_3$, for which $\Delta^\uni_{\mu I_3}=\Delta^\dep_\mu$.

Although the whole matrix $K$ is generally required to characterize unital noise, in some cases only its singular values are relevant. We use $s_1,s_2,s_3$ to denote the singular values of $K$ in descending order; thus, $s_1=\|K\|_\infty$. Equivalently, $\|K\|_\infty=\|K\|_{2\to2}$; we use the latter notation when invoking the induced-norm characterization
\[
\|K\|_{2\to2} = \max_{v\neq 0}\frac{\|Kv\|_2}{\|v\|_2}.
\]

We first state the multiplicativity result for unital qubit channels proved by King \cite{King14}. This result is stated using the induced norms obtained from the unnormalized Schatten norms.

\begin{theorem}[Unital multiplicativity theorem {\cite[Theorem~3]{King14}}]
\label{thm:unital-multiplicativity}
Let $\Delta^\uni_K$ be a unital qubit channel, and let $\Lambda$ be any completely positive map on operators of arbitrary finite dimension. If
\[
1\leq p\leq2\leq q,
\]
then
\[
\left\|\Lambda\otimes\Delta^\uni_K\right\|_{p\to q}
=
\|\Lambda\|_{p\to q}
\left\|\Delta^\uni_K\right\|_{p\to q}.
\]
\end{theorem}

King's theorem is a multiplicativity theorem rather than, by itself, a hypercontractive inequality for every unital qubit channel. We will obtain the hypercontractive statement that we need by comparing each unital channel to a depolarizing channel and then applying Theorem~\ref{thm:unital-multiplicativity}.

\begin{theorem}[Hypercontractivity for general unital qubit noise]
\label{thm:unital-hypercontractivity}
Let $\Delta^\uni_{K_1},\ldots,\Delta^\uni_{K_m}$ be unital qubit channels, with Bloch matrices $K_1,\ldots,K_m$. Suppose that
\[
\max_{j\in[m]}
\left\|K_j\right\|_{2\to2}
\leq
\lambda
\]
for some $\lambda\in[0,1]$. If
\[
1<p\leq2\leq q<\infty
\qquad\text{and}\qquad
\lambda
\leq
\sqrt{\frac{p-1}{q-1}},
\]
then, for every operator $X$ on $m$ qubits,
\[
\left\|
\left(
\Delta^\uni_{K_1}\otimes\cdots\otimes\Delta^\uni_{K_m}
\right)(X)
\right\|^\norm_q
\leq
\|X\|^\norm_p.
\]
\end{theorem}

We first prove a lemma dealing with the special case $m=1$.
\begin{lemma}\label{lem:uni-qubit}
Let $X$ be a $2\times2$ matrix, and let $\Delta^\uni_K$ be a unital qubit channel satisfying $\|K\|_{2\to2} \leq \lambda$. If $\lambda,p,q$ satisfy the same conditions as in Theorem~\ref{thm:unital-hypercontractivity}, then
\[ \|\Delta^\uni_K(X)\|_q^\norm \leq \|X\|_p^\norm.\]
\end{lemma}
\begin{proof}
We first prove the lemma for positive semidefinite $X$. Write the $2\times2$ positive semidefinite matrix $X$ as
\[
X
=
aI_2+\sum_{i=1}^3v_i\sigma_i,
\]
where $\sigma_i$ are the Paulis. Since $X$ is Hermitian, $a$ and $v_1,v_2,v_3$ are real. It can be checked that the eigenvalues of $X$ are $a+\|v\|_2$ and $a-\|v\|_2$. In particular, positivity of $X$ implies $a\geq\|v\|_2$.

Because $\Delta^\uni_K$ is unital and trace preserving, by the definition of $K$,
\[
\Delta^\uni_K(X)
=
aI_2+
\sum_{i=1}^3
(K v)_i\sigma_i.
\]
Its eigenvalues are therefore $a+\|K v\|_2$ and $a-\|K v\|_2$. The assumption on the Bloch matrix gives
\begin{equation}\label{eq:K-contract}
\|K v\|_2
\leq
\|K\|_{2\to2}\|v\|_2
\leq
\lambda\|v\|_2.
\end{equation}
Since $\lambda\leq1$, both $\|K v\|_2$ and $\lambda\|v\|_2$ lie in $[0,a]$, by the bound on $\|v\|_2$.

For fixed $a\geq0$, the normalized Schatten $q$-norm of a positive operator with eigenvalues $a+s$ and $a-s$ is non-decreasing in $s\in[0,a]$. Indeed, the $q$-th power of the Schatten norm is
\[
F_a(s) = \frac{(a+s)^q+(a-s)^q}{2},
\]
whose derivative with respect to $s$ is
\[
\frac q2
\left(
(a+s)^{q-1}-(a-s)^{q-1}
\right)
\geq0.
\]
Therefore,
\begin{align*}
\left(\left\|\Delta^\uni_K(X)\right\|^\norm_q\right)^q & = F_a\left(\|K v\|_2\right) \\
 & \leq F_a\left(\lambda\|v\|_2\right) \\
 & = F_a\left(\|\lambda v\|_2\right) \\
 & = \left(\left\|
aI_2+
\sum_{i=1}^3\lambda v_i\sigma_i
\right\|^\norm_q\right)^q.
\end{align*}
The operator in the final expression above is precisely $\Delta^\dep_\lambda(X)$. Taking $q$th roots therefore gives
\[
\left\|\Delta^\uni_K(X)\right\|^\norm_q
\leq
\left\|\Delta^\dep_\lambda(X)\right\|^\norm_q.
\]
By Theorem~\ref{thm:depolarizing-hypercontractivity} and the assumption
\[
\lambda
\leq
\sqrt{\frac{p-1}{q-1}},
\]
we conclude that
\[
\left\|\Delta^\uni_K(X)\right\|^\norm_q
\leq
\left\|\Delta^\dep_\lambda(X)\right\|^\norm_q
\leq
\|X\|^\norm_p
\]
for all PSD $X$.

Since $\Delta^\uni_K$ is completely positive, its
$p\to q$ induced norm can be evaluated by restricting the supremum to positive
semidefinite inputs; see~\cite[Theorem~1]{Watrous05}. The same statement
holds if we define induced norms using normalized Schatten norms, since normalization only multiplies
the induced norm by a constant depending on the input and output
dimensions. Since the lemma statement is essentially upper bounding $\|\Delta^\uni_K\|_{p\to q}^\norm$, the inequality proved above for $X\succeq0$
already implies
\[
\left\|\Delta^\uni_K(X)\right\|^\norm_q
\leq
\|X\|^\norm_p
\]
for every operator $X$.
\end{proof}

We are now ready to prove the full hypercontractivity theorem.
\begin{proof}[Proof of Theorem~\ref{thm:unital-hypercontractivity}]
King's theorem is stated using unnormalized Schatten norms, so we now keep track of the normalization. For a map on qubit operators,
\[
\sup_{X\neq0}
\frac{\|\Delta^\uni_{K_j}(X)\|^\norm_q}
{\|X\|^\norm_p}
=
2^{1/p-1/q}
\left\|\Delta^\uni_{K_j}\right\|_{p\to q}.
\]
Lemma~\ref{lem:uni-qubit} therefore implies that
\[
\left\|\Delta^\uni_{K_j}\right\|_{p\to q}
\leq
2^{1/q-1/p}.
\]

The tensor product of completely positive maps is completely positive. We may therefore apply Theorem~\ref{thm:unital-multiplicativity} repeatedly to obtain
\[
\left\|
\Delta^\uni_{K_1}\otimes\cdots\otimes\Delta^\uni_{K_m}
\right\|_{p\to q}
=
\prod_{j=1}^m
\left\|\Delta^\uni_{K_j}\right\|_{p\to q}
\leq
2^{m(1/q-1/p)}.
\]
Finally, the input and output spaces both have dimension $2^m$. Hence, for every operator $X$,
\begin{align*}
\left\|
\left(
\Delta^\uni_{K_1}\otimes\cdots\otimes\Delta^\uni_{K_m}
\right)(X)
\right\|^\norm_q
&=
2^{-m/q}
\left\|
\left(
\Delta^\uni_{K_1}\otimes\cdots\otimes\Delta^\uni_{K_m}
\right)(X)
\right\|_q\\
&\leq
2^{-m/q}
2^{m(1/q-1/p)}
\|X\|_p\\
&=
2^{-m/p}\|X\|_p\\
&=
\|X\|^\norm_p.
\end{align*}
This proves the theorem.
\end{proof}

We now calculate the maximal correlation of EPR pairs under general unital noise.
\begin{lemma}
\label{lem:MC-unital}
Let $\Delta^\uni_K$ be a unital qubit channel. Then
\[
\MC\left(\Omega^\uni_K\right)
=
\left\|K\right\|_\infty.
\]
\end{lemma}

\begin{proof}
Both marginals of $\Omega^\uni_K$ are maximally mixed. Hence the operators in the definition of maximal correlation satisfy
\[
\Tr(M)
=
\Tr(N)
=
0,
\qquad
\frac12\Tr(M^\dagger M)
=
\frac12\Tr(N^\dagger N)
=
1.
\]
The three Pauli operators form an orthonormal basis for the traceless qubit operators with respect to the normalized Hilbert--Schmidt inner product. We may therefore write
\[
M
=
\sum_{j=1}^3m_j\sigma_j,
\qquad
N
=
\sum_{i=1}^3n_i\sigma_i,
\]
where the coefficient vectors $m,n\in\bbC^3$ satisfy $\|m\|_2=\|n\|_2=1$.

By the definition of $\Omega^\uni_K$ and the maximally entangled-state identity,
\[
\Tr\left(\Omega^\uni_KM^\dagger\otimes N\right)
=
\frac12\Tr\left((M^\dagger)^{\mathsf T}(\Delta^\uni_K)^*(N)\right)
\]
where $(\Delta^\uni_K)^*$ is the adjoint of $\Delta^\uni_K$. Indeed, the first equality below moves the channel from the state to the observable, and the second applies the maximally entangled state identity:
\[
\begin{split}
\Tr\left[
(M^\dagger\otimes N)
\bigl(\operatorname{id}\otimes\Delta^\uni_K\bigr)
\bigl(\ket{\Phi^+}\!\bra{\Phi^+}\bigr)
\right]
&=
\bra{\Phi^+}
M^\dagger\otimes(\Delta^\uni_K)^*(N)
\ket{\Phi^+}\\
&=
\frac12\Tr\left((M^\dagger)^{\mathsf T}(\Delta^\uni_K)^*(N)\right).
\end{split}
\]

Transposition fixes $\sigma_1$ and $\sigma_3$ and changes the sign of $\sigma_2$. Denote the corresponding orthogonal sign change on Pauli coefficient vectors by
\[
R_{\mathrm{trans}}
=
\operatorname{diag}(1,-1,1).
\]
Thus $(M^\dagger)^{\mathsf T}$ has coefficient vector $R_{\mathrm{trans}}\overline m$. By the definition of the Bloch matrix, $(\Delta^\uni_K)^*(N)$ has coefficient vector $K^{\mathsf T}n$. Orthonormality of the Pauli basis now gives
\[
\Tr\left(\Omega^\uni_KM^\dagger\otimes N\right)
=
m^\dagger R_{\mathrm{trans}}K^{\mathsf T}n.
\]
Here $R_{\mathrm{trans}}$ is orthogonal, so $\|R_{\mathrm{trans}}m\|_2=\|m\|_2=1$. It follows from the definition of the largest singular value that
\[
\left|
\Tr\left(\Omega^\uni_KM^\dagger\otimes N\right)
\right|
\leq
\left\|K\right\|_\infty.
\]

To see that equality is attained, let $u,v\in\bbR^3$ be unit right and left singular vectors for the largest singular value, chosen so that
\[
v^{\mathsf T}K u
=
\left\|K\right\|_\infty.
\]
Choose $m=R_{\mathrm{trans}}u$ and $n=v$. The corresponding $M$ and $N$ are traceless Hermitian operators satisfying the required normalization conditions, and
\[
\Tr\left(\Omega^\uni_KM^\dagger\otimes N\right)
=
v^{\mathsf T}K u
=
\left\|K\right\|_\infty.
\]
This proves the claimed equality.
\end{proof}

\subsection{Biased reset noise}
Fix a reset bias $b\in[-1,1]$, and let
\[
\tau_b
=
\frac{I_2+b\sigma_Z}{2}.
\]
Thus $b=0$ gives the maximally mixed state, while $b=1$ and $b=-1$ give the pure states $\ket{0}$ and $\ket{1}$, respectively. For a retention probability $\lambda\in[0,1]$, define the biased reset channel by
\[
\Delta^\res_{\lambda,b}(X)
=
\lambda X
+
(1-\lambda)\Tr(X)\tau_b.
\]
The channel leaves the input unchanged with probability $\lambda$ and otherwise replaces it by $\tau_b$. Applying it to one half of an EPR pair gives
\[
\Omega^\res_{\lambda,b}
=
\bigl(\operatorname{id}\otimes\Delta^\res_{\lambda,b}\bigr)
\bigl(\ket{\Phi^+}\!\bra{\Phi^+}\bigr)
=
\lambda\ket{\Phi^+}\!\bra{\Phi^+}
+
(1-\lambda)\frac{I_2}{2}\otimes\tau_b.
\]
Its marginals are
\[
\bigl(\Omega^\res_{\lambda,b}\bigr)_{E^\A}
=
\frac{I_2}{2},
\qquad
\bigl(\Omega^\res_{\lambda,b}\bigr)_{E^\B}
=
\lambda\frac{I_2}{2}
+
(1-\lambda)\tau_b.
\]
For $b=0$, this model reduces to depolarizing noise. For $b\neq0$, the channel is not unital and Bob's marginal is not maximally mixed.


A hypercontractivity theorem for the adjoint of the biased reset noise channel, analogous to Theorems~\ref{thm:depolarizing-hypercontractivity} and \ref{thm:unital-hypercontractivity}, was proved very recently in \cite{DGOYZ26}. This result uses a norm weighted by the reset state $\tau_b$ instead of the normalized Schatten norm $\|.\|^\norm_p$ (which is the norm weighted by the maximally mixed state $\frac{I_2}{2}$).

Our application requires a hypercontractivity theorem in which the norm is weighted by $\Delta^{\text{noise}}(I_2/2)$, the marginal obtained by applying the noise operator to one half of an EPR pair. This is true of Theorems~\ref{thm:depolarizing-hypercontractivity} and \ref{thm:unital-hypercontractivity}, and also of Theorem~\ref{thm:erasure-hypercontractivity} in the next subsection. For biased reset noise, a result useful for our purposes would require a norm weighted by the state $\lambda\frac{I_2}2+(1-\lambda)\tau_b$, and this is not quite available from \cite{DGOYZ26}.

We calculate the maximal correlation of EPR pairs under biased reset noise below.
\begin{lemma}
\label{lem:MC-reset}
The maximal correlation of the biased-reset EPR state is
\[
\MC\left(\Omega^\res_{\lambda,b}\right)
=
\begin{cases}
\displaystyle
\frac{\lambda}{\sqrt{1-(1-\lambda)|b|}},
& \lambda>0,\\[1ex]
0,
& \lambda=0.
\end{cases}
\]
In particular, unbiased reset gives $\MC(\Omega^\res_{\lambda,0})=\lambda$, while reset to a pure state gives $\MC(\Omega^\res_{\lambda,\pm1})=\sqrt{\lambda}$.
\end{lemma}

\begin{proof}
The eigenvalues of Bob's marginal $(\Omega^\res_{\lambda,b})_{E^\B}$ are
\[
p_0
=
\frac{1+(1-\lambda)b}{2},
\qquad
p_1
=
\frac{1-(1-\lambda)b}{2}.
\]
We will do the proof for the case $b\geq 0$, so $p_0\geq p_1$; the other case is symmetric.

Let $M$ and $N$ satisfy the conditions
\begin{align*}
\Tr\left(M(\Omega^\res_{\lambda,b})_{E^\A}\right) = 0 & \qquad \Tr\left(M^\dagger M(\Omega^\res_{\lambda,b})_{E^\A} \right) = 1 \\
\Tr\left(N(\Omega^\res_{\lambda,b})_{E^\B}\right) = 0 & \qquad \Tr\left(N^\dagger N(\Omega^\res_{\lambda,b})_{E^\B} \right) = 1
\end{align*}
in the definition of maximal correlation. Alice's marginal is maximally mixed, so $\Tr(M)=0$ and $\frac12\Tr(M^\dagger M)=1$. Therefore,
\[
\Tr\left(\Omega^\res_{\lambda,b}M^\dagger\otimes N\right)
= \lambda\Tr(\state{\Phi^+}M^\dagger\otimes N) + (1-\lambda)\Tr\left[\left(\frac{I}2\otimes\tau_b\right)(M^\dagger\otimes N)\right]
 = \frac{\lambda}{2}\Tr\left((M^\dagger)^{\mathsf T}N\right).
\]
Define the traceless part of $N$
\[
N_0
=
N-\frac{\Tr(N)}{2}I_2.
\]
It can be checked that $\Tr((M^\dagger)^{\mathsf T}N) = \Tr((M^\dagger)^{\mathsf T}N_0)$. Therefore, Cauchy--Schwarz gives
\begin{equation}\label{eq:MC-CS}
\left|
\Tr\left(\Omega^\res_{\lambda,b}M^\dagger\otimes N\right)
\right|
\leq
\frac\lambda2\sqrt{\Tr(M^\dagger M)\Tr(N_0^\dagger N_0)} = 
\lambda\sqrt{\frac12\Tr(N_0^\dagger N_0)}.
\end{equation}

Write
\[
N
=
\begin{pmatrix}
a&u\\
v&d
\end{pmatrix}.
\]
The centring condition with respect to $(\Omega^\res_{\lambda,b})_{E^\B}$ for Bob implies $p_0a+p_1d=0$, and his normalization condition is
\[
\Tr\left(N^\dagger N(\Omega^\res_{\lambda,b})_{E^\B} \right) = p_0\bigl(|a|^2+|v|^2\bigr)
+
p_1\bigl(|u|^2+|d|^2\bigr)
=
1.
\]
Using $d=-p_0a/p_1$ from the centring condition in the normalization factor,
\begin{align}
\Tr\left(N^\dagger N(\Omega^\res_{\lambda,b})_{E^\B} \right) & = p_0|v|^2 + p_1|u|^2 + \left(p_0+\frac{p_0^2}{p_1}\right)|a|^2 \nonumber \\
  & = p_0|v|^2 + p_1|u|^2 + \frac{p_0}{p_1}|a|^2 \label{eq:N-norm}
\end{align}
where we have used $p_0+p_1=1$.

Moreover,
\[ N_0 = \begin{pmatrix} 
\frac{a-d}{2} & u \\
v & \frac{d-a}{2}
\end{pmatrix}.
\]
So,
\[
\frac12\Tr(N_0^\dagger N_0)
=
\frac{|a-d|^2}{4}
+
\frac{|u|^2+|v|^2}{2}.
\]
Using $a-d = a+\frac{p_0}{p_1}a = \frac{a}{p_1}$,
\begin{equation}\label{eq:N_0-norm}
\frac12\Tr(N_0^\dagger N_0) = \frac12|v|^2+\frac12|u|^2+\frac1{4p_1^2}|a|^2.
\end{equation}

We now take the ratios of the coefficients of $|v|^2, |u|^2$ and $|a|^2$ in \eqref{eq:N_0-norm} and \eqref{eq:N-norm}. The ratios are, respectively, $\frac{1}{2p_0}, \frac1{2p_1}$ and $\frac1{4p_0p_1}$. Since $p_1\leq p_0$, all three ratios are at most $\frac1{2p_1}$. Hence
\[
\frac12\Tr(N_0^\dagger N_0)
\leq
\frac1{2p_1}\Tr\left(N^\dagger N(\Omega^\res_{\lambda,b})_{E^\B} \right) \leq \frac1{2p_1}
\]
by the normalization condition on $N$. Due to \eqref{eq:MC-CS}, this proves
\[
\MC\left(\Omega^\res_{\lambda,b}\right)
\leq \frac{\lambda}{\sqrt{2p_1}} =
\frac{\lambda}{\sqrt{1-(1-\lambda)b}}.
\]

For $\lambda>0$, equality is attained by
\[
M
=
\sqrt2\ket{1}\!\bra{0},
\qquad
N
=
\frac{\ket{0}\!\bra{1}}{\sqrt{p_1}}.
\]
These operators are centred and normalized for the two marginals, and it can be checked that their correlation is $\lambda/\sqrt{2p_1}$. When $\lambda=0$, the state is a product state and its maximal correlation is zero.
\end{proof}

\subsection{Erasure noise}
Let $\ket{e}$ be an erasure flag orthogonal to $\{\ket{0},\ket{1}\}$. The qubit erasure channel with erasure probability $\eps\in[0,1]$ is
\[
\Delta^\era_\eps(X)
=
(1-\eps)X+\eps\Tr(X)\ket{e}\!\bra{e}.
\]
Thus the input is transmitted perfectly with probability $1-\eps$, while with probability $\eps$ it is replaced by a locally visible erasure flag. Applying this channel to one half of an EPR pair gives the erased EPR state
\[
\Omega^\era_\eps
=
\bigl(\operatorname{id}\otimes\Delta^\era_\eps\bigr)
\bigl(\ket{\Phi^+}\!\bra{\Phi^+}\bigr)
=
(1-\eps)\ket{\Phi^+}\!\bra{\Phi^+}
+
\eps\frac{I_2}{2}\otimes\ket{e}\!\bra{e}.
\]
Its marginals are
\[
\bigl(\Omega^\era_\eps\bigr)_{E^\A}
=
\frac{I_2}{2},
\qquad
\bigl(\Omega^\era_\eps\bigr)_{E^\B}
=
(1-\eps)\frac{I_2}{2}
+
\eps\ket{e}\!\bra{e}.
\]
In particular, unlike depolarizing and unital noise, erasure noise does not produce a state with two maximally mixed marginals.

For an operator $X$ on $m$ qubits and $S\subseteq[m]$, the normalized partial trace over the qubits in $S$ is $2^{-|S|}\Tr_S(X)$. The following form of erasure channel hypercontractivity is especially convenient because it is expressed entirely in terms of operators on the input qubits.
\begin{theorem}[Erasure channel hypercontractivity \cite{BDOY25}]
\label{thm:erasure-hypercontractivity}
Let $1<p\leq q<\infty$ and $\eps\in[0,1]$. For every operator $X$ on $m$ qubits,
\[
\left(
\sum_{S\subseteq[m]}
(1-\eps)^{m-|S|}\eps^{|S|}
\left(
\left\|
\frac{\Tr_S(X)}{2^{|S|}}
\right\|^\norm_q
\right)^q
\right)^{1/q}
\leq
\|X\|^\norm_p
\]
whenever either of the following conditions holds:
\begin{enumerate}[label=(\roman*)]
\item $1\leq p\leq2\leq q$ and $1-\eps\leq\frac{p-1}{q-1}$;
\item $1\leq p\leq q\leq2$ and $1-\eps\leq\left(\frac{p-1}{q-1}\right)^2$.
\end{enumerate}
\end{theorem}

Finally, we compute the maximal correlation of EPR pairs under erasure noise.
\begin{lemma}
\label{lem:MC-erasure}
The maximal correlation of the erased EPR state is
\[
\MC(\Omega^\era_\eps)
=
\sqrt{1-\eps}.
\]
\end{lemma}

\begin{proof}
We first assume $\eps<1$. Let $M$ and $N$ satisfy the conditions in the definition of maximal correlation for $\Omega^\era_\eps$. Since Alice's marginal is maximally mixed, $\Tr(M)=0$ and $\frac12\Tr(M^\dagger M)=1$. Write the restriction of $N$ to Bob's unerased qubit subspace as $N_{\mathrm{qubit}}$, and let
\[
N_{\mathrm{qubit},0}
=
N_{\mathrm{qubit}}
-
\frac{\Tr(N_{\mathrm{qubit}})}{2}I_2
\]
be its traceless part. The block of $N$ connecting the qubit subspace to the erasure subspace does not contribute to the correlation (since $M$ has no component in this subspace). Further, since $M$ is traceless,
\[
\Tr\left(\Omega^\era_\eps M^\dagger\otimes N\right)
=
\frac{1-\eps}{2}
\Tr\left((M^\dagger)^{\mathsf T}N_{\mathrm{qubit},0}\right).
\]
By Cauchy--Schwarz,
\begin{equation}\label{eq:era-CS}
\left|
\Tr\left(\Omega^\era_\eps M^\dagger\otimes N\right)
\right|
\leq
(1-\eps)
\sqrt{\frac12\Tr\left(
N_{\mathrm{qubit},0}^\dagger
N_{\mathrm{qubit},0}
\right)}.
\end{equation}
On the other hand, the normalization of $N$ and positivity of all the omitted block contributions imply
\begin{equation}\label{eq:era-norm}
1
=
\Tr\left(\bigl(\Omega^\era_\eps\bigr)_{E^\B}N^\dagger N\right)
\geq
\frac{1-\eps}{2}
\Tr\left(N_{\mathrm{qubit}}^\dagger N_{\mathrm{qubit}}\right)
\geq
\frac{1-\eps}{2}
\Tr\left(
N_{\mathrm{qubit},0}^\dagger
N_{\mathrm{qubit},0}
\right).
\end{equation}
Combining \eqref{eq:era-CS} and \eqref{eq:era-norm} gives the upper bound $\sqrt{1-\eps}$:
\[ \left|\Tr(\Omega^\era_\eps M^\dagger\otimes N)\right| \leq \sqrt{1-\eps},\]
which gives the required upper bound on $\MC(\Omega^\era_\eps)$.

For the matching lower bound, take
\[
M=Z\oplus0,
\qquad
N
=
\frac{Z}{\sqrt{1-\eps}}\oplus0,
\]
where the second direct-sum block acts on the erasure flag. These operators are centred and normalized for the two marginals, and their correlation is $\sqrt{1-\eps}$. When $\eps=1$, the state is a product state and its maximal correlation is zero, agreeing with the same formula.
\end{proof}

\begin{remark}
For depolarizing noise and biased reset noise, the maximal correlation of the corresponding noisy EPR pair has a linear dependence on the probability that the noise channel does nothing ($\mu$ for depolarizing noise, $\lambda$ for biased reset). For erasure noise, the maximal correlation is the square root of the do-nothing probability $1-\eps$. This is because erasure noise tells Bob when the bad event (corresponding to not being an EPR pair) happens, so Bob can amplify his observable on the qubit block by $\sqrt{1-\eps}$. This can be seen by examining the operators $M,N$ that achieve $\MC(\Omega^\era_\eps)$: they have weight only on the qubit block and $N=M^{\mathsf T}/\sqrt{1-\eps}$. In all the other noise models, in the component in which Alice and Bob are uncorrelated, Bob's state is not in an orthogonal subspace, so this strategy does not work.
\end{remark}

\section{Non-local games and communication complexity}\label{sec:game-comm}
\subsection{Non-local games}
A two-player non-local game $G$ is described as $(\sfP, \clX\times\clY,\clA\times\clB, V)$ where $\sfP$ is a probability distribution over the input set $\clX\times\clY$, $\clA\times\clB$ is the output set, and $V:\clX\times\clY\times\clA\times\clB \to \{0,1\}$ is a predicate. It is played as follows: a referee selects inputs $(x,y)$ according to $\sfP$, sends $x$ to Alice and $y$ to Bob.
\begin{itemize}
\item In a classical strategy for $G$, Alice and Bob send outputs $a$ and $b$ back to the referee, which are just functions of $x$ and $y$. In principle, Alice and Bob can also use shared or private randomness to produce their outputs from their inputs, but it can be shown that this does not help.
\item In a quantum strategy, Alice and Bob are allowed to share any entangled state of any size. They perform measurements on their respective halves of the entangled state depending on their inputs, and send their outputs $(a,b)$ back to the referee.
\item In a depolarized quantum strategy with noise parameter $\mu$, Alice and Bob are only allowed to share (an unbounded number of) copies of the qubit Werner state $\Omega^\dep_\mu$. They can do any measurements on this state, and produce outcomes, just like a standard quantum strategy. We can additionally allow shared randomness in a noisy quantum strategy\footnote{A standard quantum strategy gets shared randomness for free, since it can be obtained by measuring EPR pairs in the standard basis. But if Alice and Bob only share noisy EPR pairs, they can no longer get perfect shared randomness for free, and we have to give it to them as an additional resource.}, but just like in the classical case, it can be shown that this makes no difference.
\item Noisy strategies with other kinds of noise are defined analogously to depolarized quantum strategies, with copies of the Werner state replaced by copies of an EPR pair under the corresponding noise channel.
\end{itemize}
The referee accepts and Alice and Bob win the game iff $V(x,y,a,b) = 1$.
\begin{definition}
The classical value of a game $G=(\sfP,\clX\times\clY,\clA\times\clB,V)$, denoted by $\omega(G)$, is the maximum winning probability of Alice and Bob, averaged over the distribution $\sfP$, over all classical strategies for $G$.

The quantum value of a game $G$, denoted by $\omega^*(G)$, is the maximum winning probability of Alice and Bob, averaged over the distribution $\sfP$ as well as inherent randomness in the strategy, over all quantum strategies for $G$.

The depolarized quantum value of $G$, denoted by $\omega^\dep_\mu(G)$, is the maximum winning probability of Alice and Bob, averaged over the distribution $\sfP$ as well as inherent randomness in the strategy, over all $\mu$-noisy quantum strategies for $G$.

We define the noisy quantum value of $G$ under other kinds of noise similarly, and denote them by $\omega^\uni_K(G), \omega^\res_{\lambda,b}(G)$, and $\omega^\era_\eps(G)$.
\end{definition}

For any game $G$, we can define its parallel-repeated version, and its threshold parallel-repeated version, as follows.
\begin{itemize}
\item In the standard parallel-repeated version of $G$, denoted by $G^n$, Alice and Bob get inputs $\bx = (x_1, \ldots, x_n)$ and $\by = (y_1,\ldots,y_n)$ from $\clX^n\times\clY^n$ according to the product input distribution $\sfP^{\otimes n}$, and must produce outputs $\ba = (a_1,\ldots,a_n)$ and $\bb = (b_1,\ldots,b_n)$ in $\clA^n\times\clB^n$ such that $V(x_i,y_i,a_i,b_i)=1$ for every $i\in[n]$. We use $V^n$ to denote this larger predicate on $\bx,\by,\ba,\bb$.
\item In the $t$-out-of-$n$ threshold version of $G$, denoted by $G^{t/n}$, Alice and Bob get the same inputs as in $G^n$ and must produce outputs from $\clA^n\times\clB^n$ such that their inputs and outputs satisfy
\[ \left|\left\{i\in [n]: V(x_i, y_i, a_i, b_i)=1\right\}\right| \geq t.\]
We use $V^{t/n}$ to denote this larger predicate on $\bx,\by,\ba,\bb$.
\end{itemize}

\paragraph{Unique games.} Of particular interest to us will be a class of games called unique games. In a unique game, for a fixed input pair $(x,y)$, and a fixed Alice-output $a$, there is a single Bob-output $b$ that is accepted by $V$, and vice versa. We use $\pi^\B_{xy}(a)$ to denote the Bob-output that is accepted on inputs $(x,y)$ and Alice-output $a$, and $\pi^\A_{xy}(b)$ to denote the Alice-output that is accepted on $(x,y)$ and Bob-output $b$. For a fixed $(x,y)$, some $a$ may not have any valid Bob-outputs, or vice versa. In these cases, we can assign $\pi^\B_{xy}(a)=\bot$ or $\pi^\A_{xy}(b)=\bot$.

\subsubsection{The CHSH game}
A very well-studied unique game is the CHSH game. In a single copy of the CHSH game, the inputs and outputs are all bits, and they must satisfy $a\oplus b = x\cdot y$. The input distribution is uniform.

The classical value of a single copy of CHSH is $\frac{3}{4}$, while the quantum value is $\cos^2(\pi/8) \approx 0.85$. We state some useful properties of the optimal quantum strategy for CHSH.
\begin{theorem}\label{thm:chsh-prop}
There exists a quantum strategy for CHSH achieving the quantum value $\cos^2(\pi/8)$, which has the following properties:
\begin{enumerate}[label=(\roman*)]
\item It uses a single EPR pair.
\item For every $x,y,a,b\in\{0,1\}$, Alice and Bob's measurement operators are projectors satisfying $\Tr(M^x_a)=\Tr(N^y_b)=1$.
\item For every input pair $(x,y)$, the probability (over the internal randomness of the strategy) that the outputs $(a,b)$ satisfy the winning condition is $\cos^2(\pi/8)$.
\end{enumerate}
\end{theorem}

Theorem~\ref{thm:chsh-prop} has the following consequence for the noisy quantum value of the CHSH game under depolarizing, biased reset and erasure noise.
\begin{corollary}\label{cor:mu-chsh}
The noisy quantum values of the CHSH game under depolarizing, biased reset and erasure noise have the lower bound
\[ \omega^{\mathrm{noise}}_p(\chsh) \geq p\cos^2\left(\frac\pi8\right) + \frac{1-p}{2},\]
where $p$ is the probability with which the noise channel does nothing, i.e., $p=\mu$ for depolarizing noise, $p=\lambda$ for biased reset noise, and $p=1-\eps$ for erasure noise.

Moreover, the strategy achieving this value satisfies the winning condition with the same probability on every input pair.
\end{corollary}
\begin{proof}
By Item~(i) of Theorem~\ref{thm:chsh-prop}, the optimal strategy for CHSH uses a single EPR pair. We consider the depolarized strategy in which the same measurements as in the noiseless strategy are performed on a single depolarized EPR pair. Let $\sfP_{XYAB}$ be the input-output distribution of the optimal quantum strategy for CHSH, and let $\sfU_{XYAB}$ be the uniform distribution. Then, by Item~(ii) of Theorem~\ref{thm:chsh-prop}, the input-output distribution of this depolarized quantum strategy is
\[ \mu\sfP_{XYAB} + (1-\mu)\sfU_{XYAB}.\]
It can be checked that $\sfU_{XYAB}$ satisfies the winning condition of CHSH with probability $\frac12$. Therefore, from this strategy we get
\[ \omega^\dep_\mu(\chsh) \geq \mu\cos^2\left(\frac\pi8\right) + \frac{1-\mu}{2}.\]

The $\sfP$ part of the depolarized quantum strategy satisfies the winning condition with the same probability on every input --- this is inherited from the noiseless strategy by Item~(iii) of Theorem~\ref{thm:chsh-prop}. It can be checked that the $\sfU$ part of the strategy also satisfies the winning condition with the same probability on every input pair. Therefore, the noisy strategy overall has this property.

The proof for biased reset and erasure noise is similar; the output distribution corresponding to the actual noise part of the biased reset or erasure noise are not necessarily uniform --- Alice's marginal distribution is uniform, but Bob's is not. Regardless, the XOR condition is satisfied with probability $\frac{1}{2}$ as long as Alice's output is uniform.
\end{proof}

We can get an analogous result for unital noise via the Horodecki criterion \cite{Hor95}, which gives necessary and sufficient conditions for arbitrary two-qubit states to violate the CHSH inequality. The following lemma is an equivalent form of the Horodecki criterion, and the strategy that achieves the maximum CHSH violation for a two-qubit state that satisfies the criterion.
\begin{lemma}[\cite{Hor95}]\label{lem:chsh-uni-lb}
Suppose the Bloch matrix $K$ has singular values $s_1 \geq s_2 \geq s_3 \geq 0$. Then
\[ \omega^\uni_K(\chsh) \geq \frac12+\frac14\sqrt{s_1^2+s_2^2}.\]
Moreover, the strategy achieving this value uses only a single copy of $\Omega^\uni_K$, and its success probabilities on the different input pairs are
\begin{align*}
\frac12\left(1+\frac{s_1^2}{\sqrt{s_1^2+s_2^2}}\right) &  \qquad \text{on } (0,0), (0,1) \\
\frac12\left(1+\frac{s_2^2}{\sqrt{s_1^2+s_2^2}}\right) & \qquad  \text{on } (1,0), (1,1).
\end{align*}
\end{lemma}
Note that even though the Horodecki criterion gives necessary and sufficient conditions and the optimal CHSH violation, the above result for $\omega^\uni_K(\chsh)$ is a lower bound and not an equality, since the optimal strategy for $\chsh$ with unital noisy EPR pairs may involve using more than one copy.

It will be useful for us to give a symmetrized version of the above strategy, which we do in the next lemma.
\begin{lemma}\label{lem:wchsh-uni-lb}
There is a unital quantum strategy for CHSH that achieves the same value as in Lemma~\ref{lem:chsh-uni-lb}, and has the same success probability on all input pairs $(x,y)$. Moreover, this strategy uses a single copy of $\Omega^\uni_K$ and two classical bits of randomness.
\end{lemma}
\begin{proof}
We symmetrize the strategy in Lemma~\ref{lem:chsh-uni-lb} using shared random bits. Let Alice and Bob share the uniformly random bits $q, r$. On inputs $x,y$, they feed
\[
x'=x\oplus q,
\qquad
y'=y\oplus r
\]
to the strategy from Lemma~\ref{lem:chsh-uni-lb} and obtain outputs $a',b'$. They return
\[
a=a'\oplus rx\oplus qr,
\qquad
b=b'\oplus qy.
\]
Since
\[
x'y'
=
xy\oplus rx\oplus qy\oplus qr,
\]
we have
\[
a\oplus b=xy
\quad\Longleftrightarrow\quad
a'\oplus b'=x'y'.
\]
For every fixed pair $(x,y)$, the pair $(x',y')$ is uniform over the four CHSH inputs. The symmetrized strategy therefore wins on every input with probability equal to the average winning probability of the original strategy.
\end{proof}

The optimal strategy for a parallel-repeated game is not always to do the optimal strategy for the single-copy game $n$ times. The exact values of parallel-repeated games are therefore often not known. However, for the CHSH game, the parallel-repeated classical and quantum values are in fact well-characterized.
\begin{theorem}[\cite{CSUU08}]
The quantum value of $n$ copies of the CHSH game is
\[ \omega^*(\chsh^n) = \cos^{2n}\left(\frac{\pi}{8}\right).\]
\end{theorem}

The following upper bound on the classical value of $\chsh^n$ is due to Ambainis, following \cite{DS14}.
\begin{theorem}[\cite{DS14,Amb26}]\label{thm:chsh-n}
The classical value of $n$ copies of the CHSH game satisfies
\[ \omega(\chsh^n) \leq \left(\frac{1+\sqrt{5}}{4}\right)^n \approx (0.81)^n.\]
Moreover, for large enough $n$, the inequality is an equality.
\end{theorem}

\subsection{Communication complexity}
The communication model is similar to the non-local game model in that Alice and Bob get inputs from $\clX\times\clY$, share some resource, and want to produce outputs in $\clA\times\clB$ that satisfy a predicate $V$. In this setting, however, Alice and Bob are allowed to communicate interactively before producing their outputs. In this work, we consider only models in which this communication is classical, although Alice and Bob may share entanglement. The literature also studies a model in which Alice and Bob do not share entanglement and communicate using quantum states.
\begin{itemize}
\item In the classical deterministic model, Alice and Bob do not share any resource or have access to private randomness, and must produce deterministic outputs $a, b$ such that for every $x, y$, $V(x,y,a,b)=1$.
\item In the classical randomized model, Alice and Bob have perfectly correlated public randomness and private randomness, collectively denoted by $r$, and they produce random outputs $A,B$ that are functions of their inputs, communicated messages, and randomness. Their outputs must satisfy $\Pr_r[V(x,y,A,B)=1] \geq 2/3$ for all $x,y$.
\item In the quantum model with perfect entanglement, Alice and Bob pre-share an arbitrary entangled state, perform measurements on their respective parts of the state depending on their inputs and the transcript so far, and send the measurement outcomes as messages. As in the randomized model, their random outputs $A,B$ must satisfy $\Pr[V(x,y,A,B)=1] \geq 2/3$ for all $x,y$, where the probability is over the internal randomness arising from the quantum measurements.

Note that this model is equivalent (up to a factor of $2$ in the communication) to the model where Alice and Bob share entanglement, and are also able to send quantum messages, due to quantum teleportation.
\item The quantum model with noisy entanglement is identical to the quantum model with perfect entanglement, except that Alice and Bob may share only unboundedly many copies of $\Omega^\dep_\mu$ or one of the other noisy EPR states. Alice and Bob may perform arbitrary measurements on their parts of the entangled state and may attach noiseless ancillas. As usual, the correctness condition is $\Pr[V(x,y,A,B)=1]\geq 2/3$ for all $x,y$.
\item The quantum model with noisy entanglement and shared randomness is identical to the quantum model with noisy entanglement, except for the fact that Alice and Bob are allowed to share perfectly correlated randomness in addition to the noisy entanglement. Alice and Bob's measurements and messages can depend on this shared randomness, and the success probability is considered on average over it.

Note that if we want the noisy quantum model to be at least as powerful as the classical randomized model, it is important to give Alice and Bob this additional shared randomness. This is in contrast to the non-local games case, where giving additional shared randomness made no difference, because giving shared randomness also made no difference in the classical model.
\end{itemize}

There is also a quantum communication model in which Alice and Bob do not share entanglement, but send quantum states to each other to communicate. Since we are specifically interested in noisy entanglement as a resource, we will not be considering this model in this work.

With some abuse of notation, we will use $V\subseteq(\clX\times\clY)\times(\clA\times\clB)$ to denote the relation $\{(x,y,a,b) : V(x,y,a,b)=1\}$. We now define the communication complexities of $V$ in the different models.
\begin{definition}[Communication complexities, $\D, \R, \Q^*, \Q^\dep_\mu, \Q^{\dep,\pub}_\mu$]
The deterministic communication complexity $\D(V)$ of $V$ is the minimum communication of a deterministic protocol that computes $V$.

The randomized communication complexity $\R^\pub(V)$ of $V$ is the minimum communication of a randomized protocol that computes $V$.

The quantum communication complexity $\Q^*(V)$ of $V$ is the minimum communication of a quantum protocol with perfect entanglement that computes $V$.

The depolarized-entanglement quantum communication complexity $\Q^\dep_\mu(V)$ of $V$ is the minimum communication of a quantum protocol that computes $V$ when the shared entanglement consists of copies of $\Omega^\dep_\mu$.

The depolarized-entanglement and shared-randomness quantum communication complexity $\Q^{\dep,\pub}_\mu(V)$ of $V$ is the minimum communication of a quantum protocol computing $V$ when the shared entanglement consists of copies of $\Omega^\dep_\mu$ and the players additionally share perfectly correlated randomness.

$\Q^\uni_K(V), \Q^\res_{\lambda,b}(V), \Q^\era_\eps(V)$ and $\Q^{\uni,\pub}_K(V), \Q^{\res,\pub}_{\lambda,b}(V), \Q^{\era,\pub}_\eps(V)$ are defined similarly to $\Q^\dep_\mu(V)$ and $\Q^{\dep,\pub}_\mu(V)$, respectively, using the noisy EPR pairs from the corresponding noise model.
\end{definition}

It is clear that
\[ \D(V) \geq \R^\pub(V) \geq \Q^{\dep,\pub}_\mu(V) \geq \Q^*(V), \]
and moreover,
\[ \D(V) \geq \Q^\dep_\mu(V) \geq \Q^{\dep,\pub}_\mu(V) \geq \Q^*(V).\]
Similar statements hold for all the other noise models. Strictly speaking, $\R^\pub(V)$ and $\Q^\dep_\mu(V)$ are incomparable. However, since in $\Q^\dep_\mu$, Alice and Bob are allowed to use ancillas, they can use these ancillas to generate private randomness. Let $\R^{\mathrm{priv}}(V)$ denote the minimum communication cost of a protocol for $V$ in a model where Alice and Bob are only allowed to use private randomness (and the correctness is measured over this private randomness). It is a well-known result of Newman \cite{New91} that\footnote{Newman's theorem is usually stated for Boolean functions, but it is not difficult to generalize it to relations.}
\[ \R^{\mathrm{priv}}(V) \leq \R^\pub(V)+O(\log\log(|\clX|\cdot|\clY|)).\]
Using a similar argument, we can also show for any $\mu$,
\[ \Q^\dep_\mu(V) \leq \R^\pub(V)+O(\log\log(|\clX|\cdot|\clY|)),\]
and a similar statement holds for any other noise model.

\begin{definition}[Distributional communication complexities]
For any communication model, correctness can be measured with respect to a distribution $\sfP$ over the inputs instead of in the worst case. That is, we require $\Pr_{\sfP}[V(x,y,A,B)=1] \geq 2/3$, where the probability is over the input distribution and the internal randomness of the protocol. We use $\D(V,\sfP)$, $\R^\pub(V,\sfP)$, $\Q^*(V,\sfP)$, $\Q^\dep_\mu(V,\sfP)$, and $\Q^{\dep,\pub}_\mu(V,\sfP)$ to denote the minimum communication of the corresponding types of protocols that satisfy average-case correctness over $\sfP$. Similar notation is used for all the other noise models.
\end{definition}

\begin{theorem}[Yao's lemma]\label{thm:yao}
Models with shared perfect randomness (explicit or implicit like in the case of $\Q^*$), such as $\R^\pub, \Q^*, \Q^{\dep,\pub}_\mu$ satisfy
\begin{align*}
\R^\pub(V) & = \sup_\sfP\D(V,\sfP)\footnotemark \\
\Q^*(V) & = \sup_\sfP\Q^*(V,\sfP) \\
\Q^{\dep,\pub}_\mu(V) & = \sup_\sfP\Q^{\dep,\pub}_\mu(V,\sfP).
\end{align*}
\footnotetext{For reasons similar to non-local games, when we consider correctness with respect to a distribution, extra shared randomness does not really help in communication. Therefore, $\R^\pub(V,\sfP)$ is actually just equal to $\D(V,\sfP)$. Similarly, $\Q^{\dep,\pub}_\mu(V,\sfP)$ is equal to $\Q^\dep_\mu(V,\sfP)$.}
Similar statements hold for all the other noise models.
\end{theorem}

We have the following analogue of the Yao--Kremer decomposition for the entanglement-assisted model with classical communication.
\begin{lemma}[{\cite[Section~3.3]{Lal24}}] \label{lem:clevebuhrman}
Fix a two-way entanglement-assisted classical-communication protocol with input sets $\clX,\clY$ and output bit $b$. Suppose the protocol communicates $C$ bits in the worst case and uses the initial entangled state $\Psi_{E^\A E^\B}$. There exist collections of positive semidefinite operators
\[
\{M^{x,z}\}_{x\in\clX,\,z\in\{0,1\}^C}\subseteq L^2(\Psi_{E^\A}),
\qquad
\{N^{y,z}\}_{y\in\clY,\,z\in\{0,1\}^C}\subseteq L^2(\Psi_{E^\B}),
\]
all with spectral norm at most one, such that, for all $x\in\clX$ and $y\in\clY$,
\[
\Pr[b=1|x,y]
=
\sum_{z\in\{0,1\}^C}
\Tr\!\left(\Psi_{E^\A E^\B}(M^{x,z}\otimes N^{y,z})\right).
\]
\end{lemma}

\subsubsection{Simultaneous message passing model}
In the simultaneous message passing model of communication complexity, Alice and Bob cannot directly communicate with each other. Instead, both of them, after getting their inputs, send a single message to a referee, who does not have access to any input. The referee must then produce an output based on these two messages alone, and we require that the output is correct with probability at least $\frac{2}{3}$ as before. The communication cost of an SMP protocol is the total number of bits communicated by Alice and Bob to the referee.

The shared resource (such as randomness and entanglement) can be of several different types in the SMP model:
\begin{itemize}
\item Alice and Bob share a resource between themselves but share nothing with the referee;
\item Alice and Bob separately share resources with the referee, but nothing with each other;
\item Alice, Bob, and the referee share a joint resource.
\end{itemize}
In this work, we will only be dealing with the first type of shared resource. Communication complexities in the SMP model are denoted by $\D^{\parallel}(V), \R^\parallel(V), \Q^\parallel(V)$.

We will only define the noisy randomized version of SMP communication complexity.
\begin{definition}[Noisy randomized SMP communication complexity]
In the noisy randomized SMP model, Alice and Bob share noisy randomness, i.e., they respectively have random strings $R^\A$ and $R^\B$ that have $\MC(R^\A;R^\B)=\rho$ for some $\rho\in(0,1)$. The referee will not have access to $R^\A$ or $R^\B$. Alice and Bob use $R^\A$ and $R^\B$ respectively, alongside their inputs, in order to produce their messages to the referee. The noisy randomized SMP communication cost of $V$, denoted by $\R^{\parallel,\mathrm{noise}}_\rho(V)$, is the minimum communication cost of a noisy randomized SMP protocol with parameter $\rho$, for $V$.
\end{definition}

\section{Upper bounds on noisy quantum values of non-local games}\label{sec:game-ub}
In this section, we prove that the unital quantum value of a non-local game can be upper bounded by an NPA local level 1 SDP with an additional constraint. We will also upper bound the noisy quantum value of the CHSH game specifically under other kinds of noise.

Before getting into the results of this section, we will first introduce the NPA local level 1 formulation.

\subsection{The NPA hierarchy}
Let $G=(\sfP,\clX\times\clY,\clA\times\clB,V)$ be a two-player non-local game.

For every $x,a,y,b$, introduce formal Hermitian symbols
$M^x_a$ and $N^y_b$. Define the local level 1 word set
\[
\mathcal W_{\mathrm{loc},1}
=
\{I\}
\cup
\{M^x_a:x\in\mathcal X,\ a\in\mathcal A\}
\cup
\{N^y_b:y\in\mathcal Y,\ b\in\mathcal B\}
\cup
\{M^x_aN^y_b:x\in\mathcal X,\ y\in\mathcal Y,\
  a\in\mathcal A,\ b\in\mathcal B\}.
\]
Thus, the NPA local level 1 moment matrix $\Gamma$ is a square matrix of size
\[
1+|\mathcal X||\mathcal A|
+|\mathcal Y||\mathcal B|
+|\mathcal X||\mathcal Y||\mathcal A||\mathcal B|
\]
that satisfies certain consistency constraints. For an actual strategy with shared state $\Psi_{E^\A E^\B}$, its entries would be
\[
\Gamma_{u,v}=\Tr\!\left(\Psi_{E^\A E^\B}u^\dagger v\right).
\]
The consistency constraints encode the operator identities satisfied by the POVMs $\{M^x_a\}$ and $\{N^y_b\}$.

A real symmetric matrix
$\Gamma=(\Gamma_{u,v})_{u,v\in\mathcal W_{\mathrm{loc},1}}$ is
\emph{POVM-consistent} at local level 1 if every real linear
identity among the products $u^\dagger v$ that follows from
\[
(M^x_a)^\dagger=M^x_a,
\qquad
(N^y_b)^\dagger=N^y_b,
\qquad
[M^x_a,N^y_b]=0,
\]
and
\[
\sum_{a\in\mathcal A}M^x_a=I,
\qquad
\sum_{b\in\mathcal B}N^y_b=I
\]
is imposed on the corresponding entries of $\Gamma$. More explicitly,
whenever real coefficients $\alpha_{u,v}$ satisfy
\[
\sum_{u,v\in\mathcal W_{\mathrm{loc},1}}
\alpha_{u,v}u^\dagger v=0
\]
as a formal consequence of the displayed relations, we impose
\[
\sum_{u,v\in\mathcal W_{\mathrm{loc},1}}
\alpha_{u,v}\Gamma_{u,v}=0.
\]
This is only a finite system of linear equations: it is obtained by
reducing the finitely many products $u^\dagger v$ using the displayed
relations.

For projective measurements, one adds the standard projector
relations
\[
M^x_aM^x_{a'}=\delta_{a,a'}M^x_a,
\qquad
N^y_bN^y_{b'}=\delta_{b,b'}N^y_b.
\]
We then say that $\Gamma$ is \emph{projector-consistent} at local level 1. For example, these relations imply
\[
\Gamma_{M^x_a,M^x_a}=\Gamma_{I,M^x_a},
\qquad
\Gamma_{M^x_aN^y_b,M^x_aN^y_b}
=
\Gamma_{I,M^x_aN^y_b},
\]
as well as the corresponding orthogonality and completeness
identities.

\begin{theorem}[NPA local level 1 value \cite{NPA08}]
\label{thm:NPA1}
For any two-player non-local game $G=(\sfP,\clX\times\clY,\clA\times\clB,V)$, its quantum value $\omega^*(G)$ is upper bounded by the following SDP:
\[
\begin{aligned}
\text{\rm maximize}\quad
&
\sum_{x,y}\sfP(x,y)
\sum_{a,b}V(x,y,a,b)
\Gamma_{I,M^x_aN^y_b}
\\
\text{\rm subject to}\quad
&
\Gamma\succeq0,\qquad
\Gamma\text{ real symmetric},\qquad
\Gamma_{I,I}=1,
\\
&
\Gamma\text{ is POVM-consistent},
\\
&
\Gamma_{M^x_a,M^x_a}\leq\Gamma_{I,M^x_a},
\qquad
\Gamma_{N^y_b,N^y_b}\leq\Gamma_{I,N^y_b},
\\
&
\Gamma_{M^x_aN^y_b,M^x_aN^y_b}
\leq
\Gamma_{I,M^x_aN^y_b},
\end{aligned}
\]
for every $x,x'\in\mathcal X$, $y,y'\in\mathcal Y$,
$a,a'\in\mathcal A$, and $b,b'\in\mathcal B$.
If the players are restricted to projective measurements, one may
replace POVM consistency by projector consistency.
\end{theorem}

\subsection{NPA local level 1 bound on noisy value of a general game}
In this section, we show that the value of a game with noisy entangled states that have maximally mixed marginals and a known maximal correlation can be upper bounded by the NPA local level 1 SDP with some extra constraints. We first prove a lemma needed for the main result of the section.

\begin{lemma}
\label{lem:commutator}
Suppose that the state $\Psi_{E^\A E^\B}$ has maximally mixed marginals and
$\MC(\Psi)=\rho$. For any Hermitian operators
$M,M'$ on $E^\A$ and $N,N'$ on
$E^\B$,
\[
\left|
\Tr\!\left(
\Psi_{E^\A E^\B}
[M,M']\otimes[N,N']
\right)
\right| \leq
\rho\,
\|[M,M']\|_\infty
\|[N,N']\|_\infty.
\]
In particular:
\begin{enumerate}
\item if $0\preceq M,M',N,N'\preceq I$, then
\[
\left|
\Tr\!\left(
\Psi_{E^\A E^\B}
[M,M']\otimes[N,N']
\right)
\right|
\leq \frac{\rho}{4};
\]
\item if $M,M',N,N'$ are observables with eigenvalues in
$\{-1,+1\}$, then
\[
\left|
\Tr\!\left(
\Psi_{E^\A E^\B}
[M,M']\otimes[N,N']
\right)
\right|
\leq 4\rho.
\]
\end{enumerate}
\end{lemma}

\begin{proof}
The two commutators in the statement are anti-Hermitian, so
$-i[M,M']$ and $-i[N,N']$ are Hermitian. They satisfy
\[
\Tr\!\left(\Psi_{E^\A}[M,M']\right)
=
\frac{1}{d_{E^\A}}\Tr([M,M'])
=0,
\]
and similarly on Bob's side. Equation~\eqref{eq:MC}
therefore gives
\[
\begin{aligned}
&\left|
\Tr\!\left(
\Psi_{E^\A E^\B}
[M,M']\otimes[N,N']
\right)
\right|\\
&\quad\leq
\rho
\sqrt{
\Tr\!\left(
\Psi_{E^\A}[M,M']^\dagger[M,M']
\right)}
\sqrt{
\Tr\!\left(
\Psi_{E^\B}[N,N']^\dagger[N,N']
\right)}\\
&\quad\leq
\rho\,\|[M,M']\|_\infty\|[N,N']\|_\infty.
\end{aligned}
\]

For Item~1, if $0\preceq M,M'\preceq I$, then
$2M-I$ and $2M'-I$ are Hermitian contractions. Therefore, $\|2M-I\|_\infty, \|2M'-I\|_\infty \leq 1$. Moreover, $[M,M']=\frac14[2M-I,2M'-I]$. Thus
\begin{align*}
\|[M,M']\|_\infty
& =
\frac14
\|[2M-I,2M'-I]\|_\infty \\
& \leq \frac14\left(\|(2M-I)(2M'-I)\|_\infty + \|(2M'-I)(2M-I)\|_\infty\right) \\
& \leq \frac12\|2M-I\|_\infty\|2M'-I\|_\infty \leq\frac12.
\end{align*}
The same estimate holds for $N,N'$, proving Item~1.

If the four operators are observables, the elementary estimate
\[
\|[M,M']\|_\infty
\leq
2\|M\|_\infty\|M'\|_\infty
=2
\]
and its Bob-side analogue prove Item~2.
\end{proof}

We will now state our NPA local level 1 upper bound on the value of a non-local game with unital noise (which produces maximally mixed marginals when acting on EPR pairs).
\begin{theorem}\label{thm:noisy-NPA}
For any two-player non-local game $G=(\sfP,\clX\times\clY,\clA\times\clB,V)$, its noisy quantum value with unital noise that produces maximal correlation $\rho$ on EPR pairs is upper bounded by the following SDP:
\[
\begin{aligned}
\text{\rm maximize}\quad
&
\sum_{x,y}\sfP(x,y)
\sum_{a,b}V(x,y,a,b)
\Gamma_{I,M^x_aN^y_b}
\\
\text{\rm subject to}\quad
&
\Gamma\succeq0,\qquad
\Gamma\text{ real symmetric},\qquad
\Gamma_{I,I}=1,
\\
&
\Gamma\text{ is POVM-consistent},
\\
&
\Gamma_{M^x_a,M^x_a}\leq\Gamma_{I,M^x_a},
\qquad
\Gamma_{N^y_b,N^y_b}\leq\Gamma_{I,N^y_b},
\\
&
\Gamma_{M^x_aN^y_b,M^x_aN^y_b}
\leq
\Gamma_{I,M^x_aN^y_b}, \\
&
\left|
\Gamma_{M^x_aN^y_b,M^{x'}_{a'}N^{y'}_{b'}}
-\Gamma_{M^x_aN^{y'}_{b'},M^{x'}_{a'}N^y_b}
-\Gamma_{M^{x'}_{a'}N^y_b,M^x_aN^{y'}_{b'}}
+\Gamma_{M^{x'}_{a'}N^{y'}_{b'},M^x_aN^y_b}
\right|
\leq\frac{\rho}{4}
\end{aligned}
\]
for every $x,x'\in\mathcal X$, $y,y'\in\mathcal Y$,
$a,a'\in\mathcal A$, and $b,b'\in\mathcal B$.
\end{theorem}

\begin{proof}
From the standard NPA local level 1 SDP, we know that any state and measurements $\{M^x_a\}, \{N^y_b\}$ for $G$ must satisfy all constraints except the last one. It remains to show that the last constraint is satisfied when the shared state is $\Psi_{E^\A E^\B}^{\otimes m}$ for some $\Psi_{E^\A E^\B}$ with maximally mixed marginals and maximal correlation $\rho$.

For fixed $x,x',y,y',a,a',b,b'$, the expression inside the absolute
value in the last constraint equals
\[
\operatorname{Re}\Tr\!\Bigl(
\Psi^{\otimes m}
\bigl(
M^x_aM^{x'}_{a'}\otimes N^y_bN^{y'}_{b'}
-M^x_aM^{x'}_{a'}\otimes N^{y'}_{b'}N^y_b
-M^{x'}_{a'}M^x_a\otimes N^y_bN^{y'}_{b'}
+M^{x'}_{a'}M^x_a\otimes N^{y'}_{b'}N^y_b
\bigr)
\Bigr).
\]
The operator in parentheses is
\[
[M^x_a,M^{x'}_{a'}]\otimes[N^y_b,N^{y'}_{b'}].
\]
It is Hermitian, so its expectation is already real. Item~2 of Lemma~\ref{lem:maxcorrproperties} and Item~1 of Lemma~\ref{lem:commutator} prove the required
$\rho/4$ bound on
\[ \left|
\Gamma_{M^x_aN^y_b,M^{x'}_{a'}N^{y'}_{b'}}
-\Gamma_{M^x_aN^{y'}_{b'},M^{x'}_{a'}N^y_b}
-\Gamma_{M^{x'}_{a'}N^y_b,M^x_aN^{y'}_{b'}}
+\Gamma_{M^{x'}_{a'}N^{y'}_{b'},M^x_aN^y_b}
\right|
. \qedhere \]
\end{proof}

We can get a slightly improved SDP for games with binary outputs, since an optimal binary strategy can be taken to use projective measurements.\footnote{Naimark dilation alone does not justify this reduction here: appending pure ancillas would destroy the maximally mixed marginal assumption in Theorem~\ref{thm:noisy-NPA}. Instead, the proof below uses the fact that a linear functional over the interval $[0,I]$ attains its maximum at a projection, without changing the shared state.}

For binary games, we can write the NPA SDP in terms of observables. For projective measurements $\{M^x_a\}, \{N^y_b\}$, we define the observables
\[
M^x=M^x_0-M^x_1,
\qquad
N^y=N^y_0-N^y_1.
\]
These satisfy
$(M^x)^2=(N^y)^2=I$. We will use the smaller word set
\[
\mathcal W_{\mathrm{bin}}
=
\{I\}
\cup\{M^x:x\in\mathcal X\}
\cup\{N^y:y\in\mathcal Y\}
\cup\{M^xN^y:x\in\mathcal X,\ y\in\mathcal Y\}.
\]
The moment matrix consequently has size
\[
1+|\mathcal X|+|\mathcal Y|+|\mathcal X||\mathcal Y|.
\]

\begin{corollary}\label{cor:noisy-NPA-bin}
Under the hypotheses of Theorem~\ref{thm:noisy-NPA}, the value of
any binary-output game with unital noise is upper bounded by
\[
\begin{aligned}
\text{\rm maximize}\quad
&
\sum_{x,y}\frac{\sfP(x,y)}{4}
\sum_{a,b\in\{0,1\}}V(x,y,a,b)
\Bigl(
1+(-1)^a\Gamma_{I,M^x}
\\[-2pt]
&\hspace{52mm}
+(-1)^b\Gamma_{I,N^y}
+(-1)^{a+b}\Gamma_{I,M^xN^y}
\Bigr)
\\
\text{\rm subject to}\quad
&
\Gamma\succeq0,\qquad
\Gamma\text{ real symmetric},
\\
&
\Gamma_{u,u}=1
\quad\text{for every }u\in\mathcal W_{\mathrm{bin}},
\\
&
\Gamma\text{ obeys all reductions following from}
\\[-2pt]
&
(M^x)^2=(N^y)^2=I
\quad\text{and}\quad
[M^x,N^y]=0,
\\
&
\left|
\Gamma_{M^xN^y,M^{x'}N^{y'}}
-\Gamma_{M^xN^{y'},M^{x'}N^y}
-\Gamma_{M^{x'}N^y,M^xN^{y'}}
+\Gamma_{M^{x'}N^{y'},M^xN^y}
\right|
\leq4\rho
\end{aligned}
\]
for every $x,x'\in\mathcal X$ and $y,y'\in\mathcal Y$.
\end{corollary}

\begin{proof}
The fact that the objective function has the form in the corollary statement can be checked by just writing out $\Tr(M^x_aN^y_b\Psi^{\otimes m})$ in terms of $\Tr(M^x\Psi^{\otimes m}), \Tr(N^y\Psi^{\otimes m})$ and $\Tr(M^xN^y\Psi^{\otimes m})$, using the relations:
\[ M^x_a=\frac{I+(-1)^aM^x}{2},
\qquad
N^y_b=\frac{I+(-1)^bN^y}{2}.
\]
Moreover, for the last constraint, it can be checked that the quantity inside the absolute value is
\[ \Tr\left([M^x,M^{x'}]\otimes[N^y,N^{y'}]\Psi^{\otimes m}\right).\]
If $M^x_a,N^y_b$ are projectors, then the corresponding observables have eigenvalues in $\{-1,+1\}$ and therefore satisfy the conditions of Item~2 of Lemma~\ref{lem:commutator}. This gives the $4\rho$ upper bound.

It remains to argue that $\{M^x_a\},\{N^y_b\}$ can be taken to be projective, even when $\Psi^{\otimes m}$ is fixed. To do this, note that after all
other measurements are fixed, the payoff is affine in a binary POVM
element $0\leq M^x_0\leq I$. A linear functional on the interval
$[0,I]$ attains its maximum at a projection: one takes the spectral
projection onto the positive eigenspace of the Hermitian coefficient
of $M^x_0$. Replacing the effects one at a time on Alice's and Bob's
sides never decreases the payoff and produces a projective binary
strategy. Thus the reflection SDP upper bounds arbitrary binary-POVM
strategies, not only strategies initially presented as projective.
\end{proof}

Finally, the NPA local level 1 formulation can be used to get an upper bound on the noisy quantum value of the CHSH game, when the noise channel is unital. We will get a different upper bound on the value of the CHSH game under more general noise in the next subsection.
\begin{corollary}\label{cor:noisy-NPA-chsh}
The unital quantum value of the CHSH game under $\Delta^\uni_K$, whose maximal correlation is $\|K\|_\infty=\rho$, is upper bounded by
\[ \omega^\uni_K(\chsh) \leq \frac12+\frac{\sqrt{1+\rho}}4.\]
\end{corollary}
\begin{proof}
We will calculate the value of the SDP from Corollary~\ref{cor:noisy-NPA-bin} for CHSH. It can be checked that the commutator constraints are just a single constraint for CHSH. The objective function for CHSH is
\[ \frac{1}{2} + \frac{1}{8}(\Gamma_{I,M^0N^0} + \Gamma_{I,M^0N^1} + \Gamma_{I,M^1N^0} - \Gamma_{I,M^1N^1} ).\]

Because $\Gamma\succeq0$, there are vectors
$\{\ket{u}:u\in\mathcal W_{\mathrm{bin}}\}$ such that
\[
\Gamma_{u,v}=\langle u| v\rangle.
\]
With these vectors, define
\[
\ket{C}
=
\ket{M^0N^0}
+\ket{M^0N^1}
+\ket{M^1N^0}
-\ket{M^1N^1}
\]
Then the objective function is
\begin{equation}\label{eq:gamma-CS}
\frac12+\frac18\langle I| C\rangle \leq \frac12+\frac18\sqrt{\langle I|I\rangle\langle C|C\rangle} = \frac12+\frac18\sqrt{\langle C|C\rangle}
\end{equation}
where we have used the Cauchy--Schwarz inequality, and the fact that $\Gamma_{I,I}=1$.

We will now calculate this $\langle C|C\rangle$ using the consistency relations. Because all four observables $M^0, M^1, N^0, N^1$ square to \(I\),
\[
(N^0+N^1)^2+(N^0-N^1)^2=4I.
\]
We can use the consistency relations corresponding to this in the SDP to simplify the expression for $\langle C|C\rangle$ to
\[
\langle C|C\rangle
=
4-\langle M^0N^0|M^1N^1\rangle + \langle M^0N^1|M^1N^0\rangle +\langle M^1N^0|M^0N^1\rangle -\langle M^1N^1|M^0N^0\rangle = 4 - \Delta_\Gamma.
\]
We now notice that $\Delta_\Gamma$ above is
\[
\Delta_\Gamma =
\Gamma_{M^0N^0,M^1N^1}
-\Gamma_{M^0N^1,M^1N^0}-
\Gamma_{M^1N^0,M^0N^1}
+\Gamma_{M^1N^1,M^0N^0}.
\]
which is precisely the expression that appears inside the absolute value in the last constraint of the SDP. By this constraint, $|\Delta_\Gamma|\leq 4\rho$. Using \eqref{eq:gamma-CS}, the objective value of the SDP is at most
\[ \frac12+\frac18\sqrt{4-\Delta_\Gamma} \leq \frac12+\frac18\sqrt{4+4\rho} = \frac12+\frac14\sqrt{1+\rho}.\]

With a bit more effort, it is also possible to prove that the upper bound is actually obtained by the SDP. In any case
\[ \omega^\uni_K(\chsh) \leq \frac12+\frac{\sqrt{1+\rho}}4. \qedhere\]
\end{proof}

\begin{remark}
Unlike the NPA level 1 SDP for the CHSH game, which shows perfect parallel repetition, the SDP from Corollary~\ref{cor:noisy-NPA-bin} for CHSH does not show parallel repetition. In fact, if we write out the SDP for two copies of CHSH, its value for any $\rho\in[0,1]$ is simply $\cos^4(\pi/8)$, i.e., just the quantum value. This is why we will show a parallel repetition theorem for the noisy quantum value of CHSH in Section~\ref{sec:parrep} with entirely different techniques.
\end{remark}

\subsection{Improved upper bound on the noisy quantum value of CHSH}
In this section, we prove an upper bound on the value of the CHSH game with noisy entanglement for any noise model with bounded maximal correlation.
\begin{theorem}\label{thm:chsh-c}
The noisy quantum value of the CHSH game with noisy entangled states that have maximal correlation $\rho\in[0,1)$ is upper bounded by $\frac12+\frac1{4\sqrt{1-\rho^2}}$.
\end{theorem}
Besides applying to more general noise, this result performs better than Corollary~\ref{cor:noisy-NPA-chsh} for unital noise in some regimes. When the unital noise is depolarizing, $\rho=\mu$, the visibility parameter of the Werner state. A comparison of the two upper bounds is given in Figure~\ref{fig:chsh-bounds}.

Combining Corollary~\ref{cor:noisy-NPA-chsh} and Theorem~\ref{thm:chsh-c}, and the maximal correlation calculations for various noise models we've considered, we get the following theorem.
\begin{corollary}\label{cor:chsh-comb}
The noisy quantum value of the CHSH game with noisy entangled states that have maximal correlation $\rho$ is upper bounded by
\[ \min\left\{\frac12+\frac1{4\sqrt{1-\rho^2}}, \frac12 + \frac1{2\sqrt{2}}\right\}.\]
Further, if the noise is unital, then the game value is upper bounded by
\[ \min\left\{\frac12+\frac1{4\sqrt{1-\rho^2}}, \frac12 + \frac{\sqrt{1+\rho}}{4}\right\}.\]
\end{corollary}
Using Lemmas~\ref{lem:MC-werner}, \ref{lem:MC-unital}, \ref{lem:MC-reset}, and \ref{lem:MC-erasure}, the upper bounds on the noisy quantum value of the CHSH game under the respective noise models from Corollary~\ref{cor:chsh-comb} are those given in Table~\ref{tab:chsh-ub-lb}.

\begin{figure}[!ht]
\centering
\includegraphics[width=0.9\linewidth]{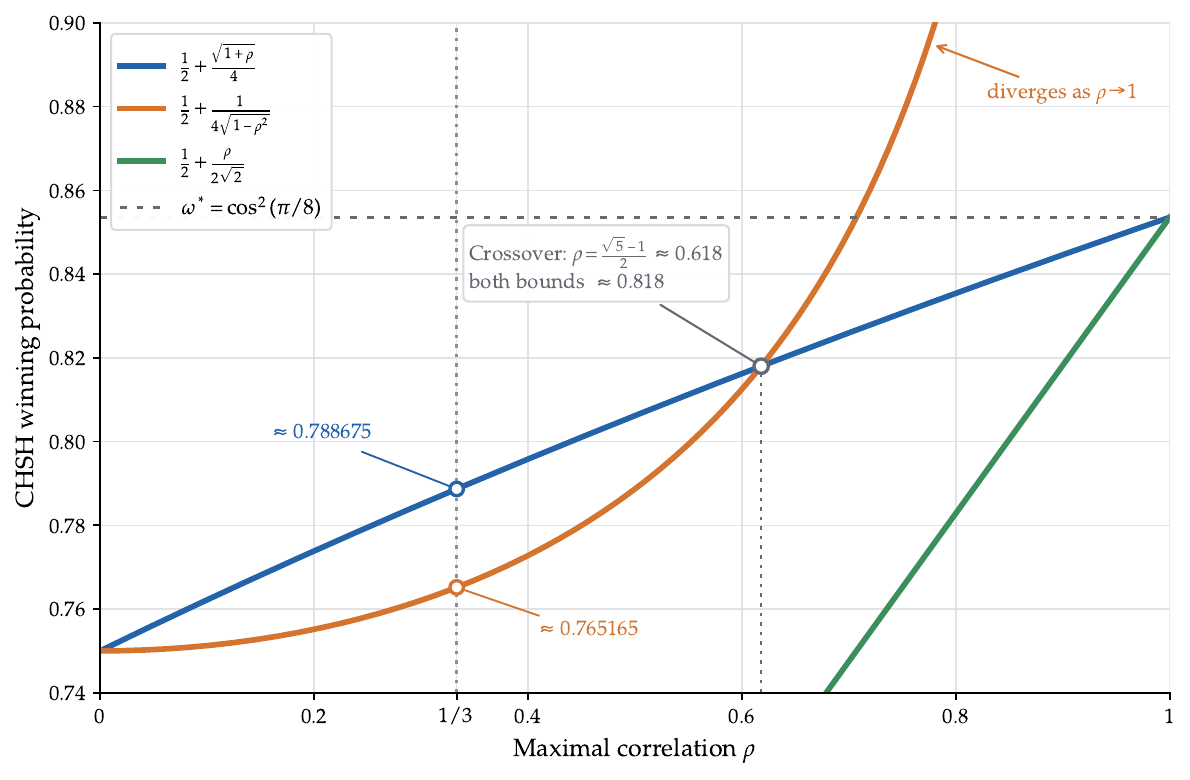}
\caption{The two upper bounds for $\omega^\uni_K(\chsh)$, and the lower bound. The bound from Theorem~\ref{thm:chsh-c} performs better up to $\|K\|_\infty = \rho=\frac{\sqrt{5}-1}2$, but diverges as $\rho$ approaches $1$. For depolarizing noise, $\rho=\mu$, the Werner visibility. At the entanglement threshold $\rho=\frac13$ for depolarizing noise, the bound from Theorem~\ref{thm:chsh-c} is better, although both bounds are higher than the correct classical value of $0.75$. There is a gap between both upper bounds and the lower bound in all regimes.}
\label{fig:chsh-bounds}
\end{figure}
 
\begin{proof}[Proof of Theorem~\ref{thm:chsh-c}]
Let $\Psi_{E^\A E^\B}$ be the noisy shared state. This state may itself consist of copies of a smaller state; by Item~2 of Lemma~\ref{lem:maxcorrproperties}, both have maximal correlation $\rho$. Let $M^0,M^1,N^0,N^1$ be the observables used in the CHSH strategy. These observables are formed from arbitrary, possibly non-projective POVMs, so they do not necessarily satisfy $(M^x)^2=(N^y)^2=I$. However, since $M^x=M^x_0-M^x_1=2M^x_0-I$ for some $0 \preceq M^x_0 \preceq I$, we have
\[
-I\preceq M^x\preceq I,
\qquad
-I\preceq N^y\preceq I.
\]
Hence the observables squared are at most $I$.

Introduce the local traces
\[
\alpha_x
=
\Tr\bigl(\Psi_{E^{\A}}M^x\bigr),
\qquad
\beta_y
=
\Tr\bigl(\Psi_{E^{\B}}N^y\bigr),
\]
and their complementary magnitudes
\[
\alpha_x^\perp=\sqrt{1-\alpha_x^2},
\qquad
\beta_y^\perp=\sqrt{1-\beta_y^2}.
\]
Because the observables are Hermitian,
\[
\begin{aligned}
\Var_{\Psi_{E^{\A}}}(M^x)
&=
\Tr\bigl(\Psi_{E^{\A}}(M^x)^2\bigr)-\alpha_x^2
\leq 1-\alpha_x^2,\\
\Var_{\Psi_{E^{\B}}}(N^y)
&=
\Tr\bigl(\Psi_{E^{\B}}(N^y)^2\bigr)-\beta_y^2
\leq 1-\beta_y^2.
\end{aligned}
\]
Equation~\eqref{eq:MC-var} therefore gives, for every
\(x,y\in\{0,1\}\),
\begin{equation}\label{eq:chsh-var}
\left|
\Tr\bigl(\Psi_{E^{\A}E^{\B}}M^x\otimes N^y\bigr)
-\alpha_x\beta_y
\right|
\leq
\rho\,\alpha_x^\perp\beta_y^\perp.
\end{equation}

Define the CHSH bias of this strategy as
\[ \bias(M^0,M^1;N^0,N^1) = \sum_{x,y}(-1)^{xy}\Tr(\Psi_{E^\A E^\B}M^x\otimes N^y).\]
The value of the game with this strategy is then $\frac12+\frac18\bias(M^0,M^1;N^0,N^1)$. Using the bound in \eqref{eq:chsh-var} on each term in the bias we get
\begin{equation}\label{eq:bias-bound}
\bias(M^0,M^1;N^0,N^1)
\leq
\alpha_0(\beta_0+\beta_1)
+\alpha_1(\beta_0-\beta_1) +
\rho(\alpha_0^\perp+\alpha_1^\perp)
 (\beta_0^\perp+\beta_1^\perp).
\end{equation}

Fix Alice's two local traces and write $\rho^{\A} = \rho(\alpha_0^\perp+\alpha_1^\perp)$.
The two terms in the RHS of \eqref{eq:bias-bound} depending on \(\beta_0\) and \(\beta_1\) can now be
maximized independently. For every \(h\in\bbR\) and \(k\geq0\),
\[
\max_{-1\leq z\leq1}
\left\{
hz+k\sqrt{1-z^2}
\right\}
=
\sqrt{h^2+k^2}.
\]
Indeed, \((z,\sqrt{1-z^2})\) ranges over the upper unit semicircle, and
Cauchy--Schwarz gives the upper bound. Choosing
\[
z=\frac{h}{\sqrt{h^2+k^2}}
\]
attains the upper bound when \((h,k)\neq(0,0)\). Applying this to the RHS of \eqref{eq:bias-bound}, separately with $z=\beta_0$, $h=\alpha_0+\alpha_1$, $k=\rho^\A$, and with $z=\beta_1$, $h=\alpha_0-\alpha_1$, $k=\rho^\A$, gives
\[
\bias(M^0,M^1;N^0,N^1)
\leq
\sqrt{(\alpha_0+\alpha_1)^2+(\rho^{\A})^2}+
\sqrt{(\alpha_0-\alpha_1)^2+(\rho^{\A})^2}.
\]

Choose \(\theta_0,\theta_1\in[0,\pi]\) such that
\[
\alpha_0=\cos\theta_0, \quad \alpha_0^\perp = \sin\theta_0
\qquad
\alpha_1=\cos\theta_1, \quad \alpha_1^\perp = \sin\theta_1.
\]
Define the midpoint and half-difference angles
\[
\theta_{\mathrm{mid}}
=
\frac{\theta_0+\theta_1}{2},
\qquad
\theta_{\mathrm{diff}}
=
\frac{\theta_0-\theta_1}{2}.
\]
Since $\theta_{\mathrm{mid}}\in[0,\pi]$ and $\theta_{\mathrm{diff}}\in[-\pi/2,\pi/2]$, both \(\sin\theta_{\mathrm{mid}}\) and
\(\cos\theta_{\mathrm{diff}}\) are nonnegative. The elementary
sum-to-product identities for $\sin$ and $\cos$ give
\[
\begin{aligned}
\alpha_0+\alpha_1
&=
2\cos\theta_{\mathrm{mid}}\cos\theta_{\mathrm{diff}},\\
\alpha_0-\alpha_1
&=
-2\sin\theta_{\mathrm{mid}}\sin\theta_{\mathrm{diff}},\\
\alpha_0^\perp+\alpha_1^\perp
&=
2\sin\theta_{\mathrm{mid}}\cos\theta_{\mathrm{diff}}.
\end{aligned}
\]
Substituting these identities into the preceding bound for the bias gives
\begin{equation}\label{eq:bias-bound-2}
\frac12\bias(M^0,M^1;N^0,N^1)
\leq \cos\theta_{\mathrm{diff}}
\sqrt{
\cos^2\theta_{\mathrm{mid}}
+\rho^2\sin^2\theta_{\mathrm{mid}}
}
+
\sin\theta_{\mathrm{mid}}
\sqrt{
\sin^2\theta_{\mathrm{diff}}
+\rho^2\cos^2\theta_{\mathrm{diff}}
}.
\end{equation}

Setting $\rho^\perp=\sqrt{1-\rho^2}$, the RHS of \eqref{eq:bias-bound-2} can be rewritten as
\[
\cos\theta_{\mathrm{diff}}
\sqrt{1-(\rho^\perp)^2\sin^2\theta_{\mathrm{mid}}} +
\sin\theta_{\mathrm{mid}}
\sqrt{1-(\rho^\perp)^2\cos^2\theta_{\mathrm{diff}}}.
\]
Choose \(\phi,\psi\in[0,\pi/2]\) satisfying
\[
\sin\phi
=
\rho^\perp\sin\theta_{\mathrm{mid}},
\qquad
\sin\psi
=
\rho^\perp\cos\theta_{\mathrm{diff}}.
\]
Multiplying both sides of \eqref{eq:bias-bound-2} by \(\rho^\perp\) therefore yields
\[
\frac{\rho^\perp}{2}
\bias(M^0,M^1;N^0,N^1)
\leq
\sin\psi\cos\phi+\sin\phi\cos\psi =
\sin(\phi+\psi) \leq 1.
\]
Since \(\rho^\perp=\sqrt{1-\rho^2}\), we conclude that the value of the game with this noisy entangled state has the upper bound
\[
\frac12+\frac18\bias(M^0,M^1;N^0,N^1)
\leq
\frac12+\frac{1}{4\sqrt{1-\rho^2}}. \qedhere
\]
\end{proof}

\section{Noisy parallel repetition theorems for non-local games}\label{sec:parrep}
In this section, we prove a parallel repetition theorem for the noisy quantum value of general games, followed by improved parallel repetition theorems for unique games and the CHSH game. We first state and prove the results for depolarizing noise, since it is the simplest case. We later generalize them to other noise models.
\subsection{General games}\label{sec:gen-parrep}
\begin{theorem}\label{thm:parrep}
Suppose for a non-local game $G=(\sfP,\clX\times\clY, \clA\times\clB, V)$, the classical value of $G^n$ is upper bounded by $\omega(G^n) \leq w_G^n$ for some $w_G\in(0,1)$. Then for any visibility level $0 \leq \mu \leq 1$, the (depolarized) noisy quantum value of $G^n$ is upper bounded by
\begin{align*}
\omega^\dep_\mu(G^n) & \leq \left(2^{\mu\log(|\clA|\cdot|\clB|)}\cdot w_G\right)^{\frac{n}{1+\mu}}.
\end{align*}
\end{theorem}

The above theorem is nontrivial for any $\mu>\frac{1}{3}$, as long as 
\[ 2^{\mu\log(|\clA|\cdot|\clB|)}\cdot w_G < 1\]
but it may be fairly weak. For the CHSH game in particular, we know exact parallel repetition theorems in the classical and noiseless quantum case. The theorem gives the following upper bound on $\omega^\dep_\mu(\chsh^n)$:
\[ \omega^\dep_\mu(\chsh^n) \leq \left(2^{2\mu}\cdot\left(\frac{1+\sqrt{5}}{4}\right)\right)^{\frac{n}{1+\mu}}.\]
Since the quantum value of $\chsh^n$ is $(\cos^2(\pi/8))^n =\left(\frac12+\frac1{2\sqrt{2}}\right)^n$, the upper bound on $\omega^\dep_\mu(\chsh^n)$ in the above theorem is only nontrivial for noise levels $\mu$ satisfying
\[ 2^{\frac{2\mu}{1+\mu}}\cdot\left(\frac{1+\sqrt{5}}{4}\right)^{\frac1{1+\mu}} \leq \frac{1+\sqrt{2}}{2\sqrt{2}} \qquad \Rightarrow \qquad \mu < 0.03.\]
However, the Werner state $\Omega^\dep_\mu$ is already unentangled for $\mu\leq \frac{1}{3}$, so this theorem does not give a nontrivial result for the CHSH game.

\paragraph{Setup.}
In order to prove Theorem~\ref{thm:parrep}, we first define some notation.
We use $\sfP_{XY}$ to denote the input distribution of the single-copy
game, and $\sfP_{\bX\bY}=\sfP_{XY}^{\otimes n}$ to denote the input
distribution of the $n$-fold $G$ game. The shared state between Alice and
Bob, when the noise parameter is $\mu$, is $(\Omega^\dep_\mu)^{\otimes m}$ for
some finite $m$.

Fix a noisy quantum strategy for this game in which, for every input
$\bx\in\clX^n$, Alice uses a POVM
$\{M^\bx_\ba\}_{\ba\in\clA^n}$, and for every
$\by\in\clY^n$, Bob uses a POVM
$\{N^\by_\bb\}_{\bb\in\clB^n}$. The resulting conditional output
distribution is
\[
\sfP_{\bA\bB|\bx\by}(\ba,\bb)
=
\Tr\!\left[
(M^\bx_\ba\otimes N^\by_\bb)(\Omega^\dep_\mu)^{\otimes m}
\right].
\]
Note that we are not considering additional shared or private randomness
in this strategy, since, as discussed, it cannot increase the winning
probability.

Because the local reduced states of $(\Omega^\dep_\mu)^{\otimes m}$ are
maximally mixed, Alice and Bob's marginal output distributions satisfy
\[
\sfP_{\bA|\bx}(\ba)
=
\frac{\Tr(M^\bx_\ba)}{2^m},
\qquad
\sfP_{\bB|\by}(\bb)
=
\frac{\Tr(N^\by_\bb)}{2^m}.
\]
We now define the additional distribution
\[
\sfQ_{\bX\bY\bA\bB}(\bx,\by,\ba,\bb)
=
\sfP_{\bX\bY}(\bx,\by)
\cdot
\sfP_{\bA|\bx}(\ba)
\cdot
\sfP_{\bB|\by}(\bb).
\]
Note that $\sfQ_{\bX\bY\bA\bB}$ is an input-output distribution that can
be obtained by a classical strategy for $G$: on inputs $\bx$ and $\by$,
Alice and Bob produce outputs from the marginal distributions
$\sfP_{\bA|\bx}$ and $\sfP_{\bB|\by}$ using private randomness.

We will first upper bound the conditional probabilities
$\sfP_{\bA\bB|\bx\by}(\ba,\bb)$ in terms of the product of the
marginals, using hypercontractivity of the depolarizing noise channel. The following lemma is similar to Lemma~15 of \cite{DB14}.
\begin{lemma}
\label{lem:pointwise}
For every $\bx,\by,\ba,\bb$ and every $\mu\in(0,1]$, the conditional
output distribution $\sfP_{\bA\bB|\bx\by}$ satisfies
\[
\sfP_{\bA\bB|\bx\by}(\ba,\bb)
\le
\left(
\sfP_{\bA|\bx}(\ba)
\sfP_{\bB|\by}(\bb)
\right)^{1/(1+\mu)}.
\]
\end{lemma}

\begin{proof}
Using \eqref{eq:choi-identity},
\[
\sfP_{\bA\bB|\bx\by}(\ba,\bb)
=
\frac{1}{2^m}
\Tr\!\left[
M^\bx_\ba
(\Delta^\dep_\mu)^{\otimes m}
\bigl((N^\by_\bb)^{\mathsf T}\bigr)
\right].
\]
Using normalized H\"older's inequality on the right-hand side above, we
have
\begin{equation}
\label{eq:prob-holder}
\sfP_{\bA\bB|\bx\by}(\ba,\bb)
\le
\|M^\bx_\ba\|^\norm_r
\left\|
(\Delta^\dep_\mu)^{\otimes m}
\bigl((N^\by_\bb)^{\mathsf T}\bigr)
\right\|^\norm_s
\end{equation}
for H\"older-conjugate $r$ and $s$. We take
\[
r=1+\mu,
\qquad
s=\frac{1+\mu}{\mu}.
\]
Aside from being H\"older conjugate, these satisfy
\[
\sqrt{\frac{r-1}{s-1}}
=
\sqrt{\frac{\mu}{1/\mu}}
=
\mu.
\]
Therefore, Theorem~\ref{thm:depolarizing-hypercontractivity} gives
\begin{equation}
\label{eq:dep-Bob}
\left\|
(\Delta^\dep_\mu)^{\otimes m}
\bigl((N^\by_\bb)^{\mathsf T}\bigr)
\right\|^\norm_s
\le
\left\|(N^\by_\bb)^{\mathsf T}\right\|^\norm_r
=
\left\|N^\by_\bb\right\|^\norm_r.
\end{equation}

Because $M^\bx_\ba$ is a POVM element, $0\le M^\bx_\ba\le I$, and all
eigenvalues of $M^\bx_\ba$ are in $[0,1]$. The scalar inequality
$t^r\le t$ for $t\in[0,1]$ and $r\geq1$ implies
\[
\Tr\!\left((M^\bx_\ba)^r\right)
\le
\Tr(M^\bx_\ba).
\]
Therefore,
\[
\begin{aligned}
\left(
\frac{\Tr((M^\bx_\ba)^r)}{2^m}
\right)^{1/r}
&\le
\left(
\frac{\Tr(M^\bx_\ba)}{2^m}
\right)^{1/r}\\
&=
\sfP_{\bA|\bx}(\ba)^{1/r}.
\end{aligned}
\]
Similarly,
\[
\left\|N^\by_\bb\right\|^\norm_r
\le
\sfP_{\bB|\by}(\bb)^{1/r}.
\]

Combining these estimates with \eqref{eq:prob-holder} and
\eqref{eq:dep-Bob}, we obtain
\[
\sfP_{\bA\bB|\bx\by}(\ba,\bb)
\le
\sfP_{\bA|\bx}(\ba)^{1/(1+\mu)}
\sfP_{\bB|\by}(\bb)^{1/(1+\mu)}.
\qedhere
\]
\end{proof}

Using the above lemma, we now upper bound the $D_\alpha$ divergence
between $\sfP$ and $\sfQ$.

\begin{lemma}
\label{lem:Renyi}
For $\mu\in(0,1]$ and
\[
\alpha=\frac{1+\mu}{\mu},
\]
the distributions $\sfP$ and $\sfQ$ defined above satisfy
\[
D_\alpha\!\left(
\sfP_{\bX\bY\bA\bB}
\middle\|
\sfQ_{\bX\bY\bA\bB}
\right)
\le
\mu\log(|\clA|\cdot|\clB|)n.
\]
\end{lemma}

\begin{proof}
By Lemma~\ref{lem:pointwise}, we have for every
$\bx,\by,\ba,\bb$,
\[
\sfP_{\bA\bB|\bx\by}(\ba,\bb)
\le
\sfQ_{\bA\bB|\bx\by}(\ba,\bb)^{1/(1+\mu)}.
\]

If $\sfQ_{\bA\bB|\bx\by}(\ba,\bb)=0$, then either $\Tr(M^\bx_\ba)=0$ or $\Tr(N^\by_\bb)=0$. A positive semidefinite operator with zero trace is zero, so in this case $\sfP_{\bA\bB|\bx\by}(\ba,\bb)=0$.

For $\sfQ_{\bA\bB|\bx\by}(\ba,\bb)>0$, raising both sides of the
inequality in Lemma~\ref{lem:pointwise} to the power $\alpha$ and
multiplying by
$\sfQ_{\bA\bB|\bx\by}(\ba,\bb)^{1-\alpha}$ gives
\[
\sfP_{\bA\bB|\bx\by}(\ba,\bb)^\alpha
\sfQ_{\bA\bB|\bx\by}(\ba,\bb)^{1-\alpha}\le
\sfQ_{\bA\bB|\bx\by}(\ba,\bb)^{\alpha/(1+\mu)+1-\alpha}
=
1.
\]
Here we have used
\[
\alpha=\frac{1+\mu}{\mu},
\qquad
\frac{\alpha}{1+\mu}+1-\alpha=0.
\]

For every fixed input pair $(\bx,\by)$, there are at most
\[
|\clA^n\times\clB^n|
=
(|\clA|\cdot|\clB|)^n
\]
output pairs $(\ba,\bb)$ on which
$\sfQ_{\bA\bB|\bx\by}(\ba,\bb)$ is nonzero. Therefore,
\[
\sum_{\ba,\bb}
\sfP_{\bA\bB|\bx\by}(\ba,\bb)^\alpha
\sfQ_{\bA\bB|\bx\by}(\ba,\bb)^{1-\alpha}
\le
(|\clA|\cdot|\clB|)^n.
\]

Furthermore, because $\sfP_{\bX\bY\bA\bB}$ and
$\sfQ_{\bX\bY\bA\bB}$ have the same input distribution,
\[
\begin{aligned}
&\sum_{\bx,\by,\ba,\bb}
\sfP_{\bX\bY\bA\bB}(\bx,\by,\ba,\bb)^\alpha
\sfQ_{\bX\bY\bA\bB}(\bx,\by,\ba,\bb)^{1-\alpha}\\
&=
\sum_{\bx,\by}
\sfP_{\bX\bY}(\bx,\by)
\sum_{\ba,\bb}
\sfP_{\bA\bB|\bx\by}(\ba,\bb)^\alpha
\sfQ_{\bA\bB|\bx\by}(\ba,\bb)^{1-\alpha}\\
&\le
(|\clA|\cdot|\clB|)^n
\sum_{\bx,\by}\sfP_{\bX\bY}(\bx,\by)\\
&=
(|\clA|\cdot|\clB|)^n.
\end{aligned}
\]

Consequently,
\[
D_\alpha\!\left(
\sfP_{\bX\bY\bA\bB}
\middle\|
\sfQ_{\bX\bY\bA\bB}
\right)
 \le
\frac{1}{\alpha-1}
\log\left((|\clA|\cdot|\clB|)^n\right) =
\mu\log(|\clA|\cdot|\clB|)n.
\qedhere \]
\end{proof}

With this, we can complete the proof of
Theorem~\ref{thm:parrep}.

\begin{proof}[Proof of Theorem~\ref{thm:parrep}]
For a strategy for the $n$-fold game $G$, let $E_n$ be the event that
all $n$ games are won. As noted above, $\sfQ$ is a classical
strategy for the $n$-fold game $G$, so, by definition,
\[
\sfQ(E_n)\leq w_G^n.
\]
Therefore, by Lemmas~\ref{lem:Renyi} and \ref{lem:event-transfer}, for
\[
\alpha=\frac{1+\mu}{\mu},
\]
we have
\[
\sfP(E_n)
\le
2^{\frac{\alpha-1}{\alpha}D_\alpha(\sfP\|\sfQ)}
\sfQ(E_n)^{\frac{\alpha-1}{\alpha}}\le
2^{\frac{\mu}{1+\mu}
\log(|\clA|\cdot|\clB|)n}
w_G^{\frac{n}{1+\mu}}.
\]
Taking $\sfP$ to correspond to the optimal strategy for $G^n$ completes
the proof.
\end{proof}

\subsection{Unique games}\label{sec:u-parrep}
In this section, we prove a parallel repetition theorem for unique games with depolarizing noise. This theorem is tighter than Theorem~\ref{thm:parrep} and gives a nontrivial result for CHSH.

\paragraph{Game collision operators and collision values.} Before stating the theorem, we need to set up some notation and define some terms. Let $\pi^\B_{xy}:\clA\to\clB\cup\{\bot\}$ and $\pi^\A_{xy}:\clB\to\clA\cup\{\bot\}$ be the functions determining the unique outputs of Bob and Alice that are accepted for a given output of Alice and Bob, respectively, on inputs $(x,y)$. We use $\sfP_{XY}$ to denote the input distribution of the single-copy unique game $G$. We also define auxiliary input-output distributions for a single party, unrelated to the input-output distribution of any particular strategy, by
\[ \hat{\sfQ}_{XA}(x,a) = \sfP_X(x)\cdot\frac{1}{|\clA|} \qquad \hat{\sfQ}_{YB}(y,b) = \sfP_Y(y)\cdot\frac{1}{|\clB|}.\]
Let $V^\A$ be the vector space consisting of functions from Alice's input-output sets $\clX\times\clA$ to $\bbR$. Similarly, let $V^\B$ be the vector space of functions from $\clY\times\clB$ to $\bbR$. Note that $\hat{\sfQ}_{XA}$ and $\hat{\sfQ}_{YB}$ are valid weight schemes for defining weighted $\ell_p$ norms on $V^\A$ and $V^\B$. That is, for $f \in V^\A, h \in V^\B$,
\[ \|f\|_{p(\hat{\sfQ}_{XA})} = \left(\sum_{xa}\hat{\sfQ}_{XA}(x,a)|f(x,a)|^p\right)^{1/p} \qquad \|h\|_{p(\hat{\sfQ}_{YB})} = \left(\sum_{yb}\hat{\sfQ}_{YB}(y,b)|h(y,b)|^p\right)^{1/p}.\]

Using $\{\pi^\B_{xy}\}_{xy}, \{\pi^\A_{xy}\}_{xy}$, and the distributions $\hat{\sfQ}_{XA}$ and $\hat{\sfQ}_{YB}$, we will now define two collision operators for $G$, acting between $V^\A$ and $V^\B$. Consider a function $f: \clX\times\clA\to\bbR$, which is in $V^\A$. The $\A\to\B$ collision operator $T^{\A\to\B}_G: V^\A \to V^\B$ acts on it to produce a function $T^{\A\to \B}_Gf$ in $V^\B$ as follows:
\[ (T^{\A\to \B}_Gf)(y,b) = \sum_x\sfP_{X|y}(x)\mathbbm{1}[\pi^\A_{xy}(b)\in\clA]f(x,\pi_{xy}^\A(b)).\]
Similarly, for a function $h:\clY\times\clB\to\bbR$ in $V^\B$, the $\B\to \A$ collision operator $T^{\B\to \A}_G:V^\B\to V^\A$ acts on it to produce a function $T^{\B\to \A}_Gh$ in $V^\A$ as follows:
\[ (T^{\B\to \A}_Gh)(x,a) = \sum_y\sfP_{Y|x}(y)\mathbbm{1}[\pi^\B_{xy}(a)\in\clB]h(y,\pi_{xy}^\B(a)).\]
Note that the distributions $\hat{\sfQ}_{XA}$ and $\hat{\sfQ}_{YB}$ can be used to define weighted induced norms on the collision operators.

We now define a quantity that will be used in the statement of our refined noisy parallel repetition theorem. Roughly speaking, this quantity measures the collision of compatible output masses of Alice and Bob for the game $G$; it will appear in a collision list lemma we will later prove.
\begin{definition}[$p$-collision value]\label{def:col-value}
For $p>0$, the $\A\to \B$ and $\B\to \A$ $p$-collision values of a unique game $G=(\sfP,\clX\times\clY, \clA\times\clB,V)$ are defined as follows:
\begin{align*}
\col^{\A\to \B}_p(G) & = |\clA|^{1/2}|\clB|^{-1/p}\left\|T^{\A\to \B}_G\right\|_{p(\hat{\sfQ}_{XA})\to2(\hat{\sfQ}_{YB})} \\
\col^{\B\to \A}_p(G) & = |\clA|^{-1/p}|\clB|^{1/2}\left\|T^{\B\to \A}_G\right\|_{p(\hat{\sfQ}_{YB})\to2(\hat{\sfQ}_{XA})}.
\end{align*}
The $p$-collision value of $G$ is defined as
\[ \col_p(G) = \sqrt{\col^{\A\to \B}_p(G)\cdot \col^{\B\to \A}_p(G)}.\]
\end{definition}
The reasons for the $|\clA|^{1/2}|\clB|^{-1/p}$ and $|\clA|^{-1/p}|\clB|^{1/2}$ normalization values in the above definition will become clear later.

We are now ready to state our refined noisy parallel repetition theorem for unique games.
\begin{theorem}\label{thm:u-parrep}
Suppose $G$ is a unique game whose parallel-repeated classical value is upper bounded as $\omega(G^n) \leq w_G^n$. Then for any visibility level $0 \leq \mu < 1$, the depolarized quantum value of $G^n$ is upper bounded by
\[ \omega^\dep_\mu(G^n) \leq O(n^2)\cdot\left((\col_{1+\mu}(G))^{4\mu}w_G^{1-\mu}\right)^{\frac{n}{1+3\mu}}.\]
\end{theorem}

We note that our techniques give an upper bound on $\omega^\dep_\mu(G^n)$ in terms of $\col_{1+\mu}(G)$ and $w_G$, but the exponents $\frac{4\mu}{1+3\mu}$ and $\frac{1-\mu}{1+3\mu}$ do not necessarily give the best possible upper bound for every game $G$. As will be clear from the proof of Theorem~\ref{thm:u-parrep}, these exponents can be optimized for an individual game. We have stated the theorem using the exponent choices that are best for the CHSH game, since $w_\chsh < \col_{1+\mu}(\chsh)$ for the relevant range of $\mu$. For other games, different functions $f,h$ satisfying certain conditions may give a tighter bound of the form $O(n^2)\left((\col_{1+\mu}(G))^{f(\mu)}w_G^{h(\mu)}\right)^n$.

\paragraph{Setup.} Throughout the proof of Theorem~\ref{thm:u-parrep}, as in Section~\ref{sec:gen-parrep}, we use $\sfP_{\bX\bY}=\sfP_{XY}^{\otimes n}$ to denote the input distribution of the $n$-fold game, and $\sfP_{\bX\bY\bA\bB}$ to denote the input-output distribution of a fixed strategy. We use $\hat{\sfQ}_{\bX\bA}$ to denote $\hat{\sfQ}_{XA}^{\otimes n}$ and $\hat{\sfQ}_{\bY\bB}$ to denote $\hat{\sfQ}_{YB}^{\otimes n}$.

For fixed inputs $\bx, \by$, using the marginal distributions $\sfP_{\bA|\bx}$ and $\sfP_{\bB|\by}$, we define probability buckets for $k, l \geq 0$,
\[ \clL^\A_{\bx,k} = \{\ba \in \clA^n: 2^{-(k+1)} < \sfP_{\bA|\bx}(\ba) \leq 2^{-k}\} \qquad \clL^\B_{\by,l} = \{\bb \in \clB^n: 2^{-(l+1)} < \sfP_{\bB|\by}(\bb) \leq 2^{-l}\}.\]
For the family of buckets $(\clL^\A_k, \clL^\B_l) = ((\clL^\A_{\bx,k})_{\bx},(\clL^\B_{\by,l})_{\by})$, define the winning probability score as
\begin{equation}\label{eq:score-1}
S_G\left(\clL^\A_{k},\clL^\B_{l}\right) = \bbE_{\sfP_{\bX\bY}}\left[\sum_{\ba\in\clA^n}\mathbbm{1}[\ba\in \clL^\A_{\bx,k}]\mathbbm{1}[\pi^\B_{\bx\by,n}(\ba)\in \clL^\B_{\by,l}]\right]
\end{equation}
where
\[ \pi^\B_{\bx\by,n}(\ba) = (\pi^\B_{x_1y_1}(a_1),\ldots,\pi^\B_{x_ny_n}(a_n)).\]
Note that the score $S_G\left(\clL^\A_k,\clL^\B_l\right)$ represents the average number of winning answer pairs in the family $(\clL^\A_{\bx,k}\times \clL^\B_{\by,l})_{\bx,\by}$. The score can equivalently be defined as
\begin{equation}\label{eq:score-2}
S_G\left(\clL^\A_{k},\clL^\B_{l}\right) = \bbE_{\sfP_{\bX\bY}}\left[\sum_{\bb\in\clB^n}\mathbbm{1}[\pi^\A_{\bx\by,n}(\bb)\in \clL^\A_{\bx,k}]\mathbbm{1}[\bb\in \clL^\B_{\by,l}]\right]
\end{equation}
with $\pi^\A_{\bx\by,n}(\bb)$ defined similarly.

Finally, define the quantity
\begin{equation}\label{eq:lambda-n}
\Lambda_{n,\mu}(G) = \sum_{\bx,\by}\sfP_{\bX\bY}(\bx,\by)\sum_{\ba,\bb: \,V^n(\bx,\by,\ba,\bb)=1}\left(\sfP_{\bA|\bx}(\ba)\sfP_{\bB|\by}(\bb)\right)^{1/(1+\mu)}
\end{equation}
which is an upper bound on the probability of winning $G^n$ by this strategy, by Lemma~\ref{lem:pointwise}.

We now prove two separate upper bounds on the score of a family $(\clL^\A_k, \clL^\B_l)$.
\begin{lemma}[Hard list bound]\label{lem:hard-list}
The score of a family of buckets $(\clL^\A_k, \clL^\B_l)$ satisfies
\[ S_G\left(\clL^\A_{k},\clL^\B_{l}\right) \leq 2^{k+l+2}w_G^n.\]
\end{lemma}
\begin{proof}
We first note that since every probability in $\clL^\A_{\bx,k}$ is at least $2^{-(k+1)}$, $\clL^\A_{\bx,k}$ can have at most $2^{k+1}$ elements; similarly, $\clL^\B_{\by,l}$ can have at most $2^{l+1}$ elements. Consider a classical strategy for $G^n$ where on input $\bx$, Alice outputs a single value $\ba$ in $\clL^\A_{\bx,k}$, and Bob outputs a single value $\bb\in \clL^\B_{\by,l}$. $S_G(\{\ba\}_\bx, \{\bb\}_\by)$ is equal to the winning probability of this classical strategy, and therefore must be upper bounded by $w_G^n$. Therefore, to upper bound the score of the whole family, we can write
\[ S_G\left(\clL^\A_{k},\clL^\B_{l}\right) \leq \sum_{\ba\in \clL^\A_{\bx,k}}\sum_{\bb\in \clL^\B_{\by,l}}S_G(\{\ba\}_\bx, \{\bb\}_\by)\leq 2^{k+l+2}w_G^n. \qedhere\]
\end{proof}

\begin{lemma}[Collision list bound]\label{lem:col-list}
The score of a family of buckets $(\clL^\A_k, \clL^\B_l)$ satisfies, for any $0 < p \leq 2$,
\[ S_G\left(\clL^\A_{k},\clL^\B_{l}\right) \leq \min\left\{2^{(k+1)/p+(l+1)/2}(\col^{\A\to\B}_p(G))^n, 2^{(k+1)/2+(l+1)/p}(\col^{\B\to\A}_p(G))^n\right\}.\]
\end{lemma}
\begin{proof}
We first note that by Lemma~\ref{lem:tensor}, we have for any $0< p \leq 2$,
\begin{align*}
\left\|\left(T^{\A\to \B}_G\right)^{\otimes n}\right\|_{p(\hat{\sfQ}_{XA}^{\otimes n})\to2(\hat{\sfQ}_{YB}^{\otimes n})} & = \left(\left\|T^{\A\to \B}_G\right\|_{p(\hat{\sfQ}_{XA})\to2(\hat{\sfQ}_{YB})}\right)^n, \\
\left\|\left(T^{\B\to \A}_G\right)^{\otimes n}\right\|_{p(\hat{\sfQ}_{YB}^{\otimes n})\to2(\hat{\sfQ}_{XA}^{\otimes n})} & = \left(\left\|T^{\B\to \A}_G\right\|_{p(\hat{\sfQ}_{YB})\to2(\hat{\sfQ}_{XA})}\right)^n.
\end{align*}
Now define the functions
\[ F(\bx,\ba) = \frac{\mathbbm{1}[\ba \in \clL^\A_{\bx,k}]}{2^{k+1}}, \qquad H(\by,\bb) = \frac{\mathbbm{1}[\bb \in \clL^\B_{\by,l}]}{2^{l+1}}.\]
By the definition of $T^{\B\to\A}_G$, we have
\begin{align*}
\left((T^{\B\to\A}_G)^{\otimes n}H\right)(\bx,\ba) & = \sum_\by\sfP_{\bY|\bx}(\by)\mathbbm{1}[\pi^\B_{\bx\by,n}(\ba)\in\clB^n]\frac{\mathbbm{1}[\pi^\B_{\bx\by,n}(\ba)\in \clL^\B_{\by,l}]}{2^{l+1}} \\
 & = \frac{1}{2^{l+1}}\sum_\by\sfP_{\bY|\bx}(\by)\mathbbm{1}[\pi^\B_{\bx\by,n}(\ba)\in \clL^\B_{\by,l}]
\end{align*}
where the last equality is due to the fact that $\clL^\B_{\by,l} \subseteq \clB^n$. 

Therefore, using this, the definition of $S_G(\clL^\A_k,\clL^\B_l)$ from \eqref{eq:score-1}, and the definition of $\hat{\sfQ}_{\bX\bA}$ we have,
\begin{align}
& S_G\left(\clL^\A_{k},\clL^\B_{l}\right) \nonumber \\
& = \sum_{\bx\by}\sfP_\bX(\bx)\sfP_{\bY|\bx}(\by)\sum_{\ba\in\clA^n}\mathbbm{1}[\ba\in \clL^\A_{\bx,k}]\mathbbm{1}[\pi^\B_{\bx\by,n}(\ba)\in \clL^\B_{\by,l}] \nonumber \\
& = 2^{k+l+2}\sum_{\bx,\ba}\sfP_{\bX}(\bx)F(\bx,\ba)\frac{1}{2^{l+1}}\sum_\by\sfP_{\bY|\bx}(\by)\mathbbm{1}[\pi^\B_{\bx\by,n}(\ba)\in \clL^\B_{\by,l}] \nonumber \\
& = 2^{k+l+2}\cdot|\clA|^n\sum_{\bx,\ba}\hat{\sfQ}_{\bX\bA}(\bx,\ba)F(\bx,\ba)\frac{1}{2^{l+1}}\sum_\by\sfP_{\bY|\bx}(\by)\mathbbm{1}[\pi^\B_{\bx\by,n}(\ba)\in \clL^\B_{\by,l}] \nonumber \\
& = 2^{k+l+2}\cdot|\clA|^n\sum_{\bx,\ba}\hat{\sfQ}_{\bX\bA}(\bx,\ba)F(\bx,\ba)\left((T^{\B\to\A}_G)^{\otimes n}H\right)(\bx,\ba) \nonumber \\
& \overset{\textcolor{red}{a}}{\leq} 2^{k+l+2}\cdot|\clA|^n\left(\sum_{\bx,\ba}\hat{\sfQ}_{\bX\bA}(\bx,\ba)F(\bx,\ba)^2\right)^{1/2}\left(\sum_{\bx,\ba}\hat{\sfQ}_{\bX\bA}(\bx,\ba)\left((T^{\B\to\A}_G)^{\otimes n}H\right)(\bx,\ba)^2\right)^{1/2} \nonumber \\
& = 2^{k+l+2}\cdot|\clA|^n\|F\|_{2(\hat{\sfQ}_{\bX\bA})}\left\|(T^{\B\to\A}_G)^{\otimes n}H\right\|_{2(\hat{\sfQ}_{\bX\bA})} \nonumber \\
& \leq 2^{k+l+2}\cdot|\clA|^n\|F\|_{2(\hat{\sfQ}_{\bX\bA})}\left\|(T^{\B\to\A}_G)^{\otimes n}\right\|_{p(\hat{\sfQ}_{\bY\bB})\to2(\hat{\sfQ}_{\bX\bA})}\|H\|_{p(\hat{\sfQ}_{\bY\bB})} \qquad (\text{for arbitrary } 0 < p \leq 2) \nonumber \\
& = 2^{k+l+2}\cdot|\clA|^n\|F\|_{2(\hat{\sfQ}_{\bX\bA})}\left(\left\|T^{\B\to\A}_G\right\|_{p(\hat{\sfQ}_{YB})\to2(\hat{\sfQ}_{XA})}\right)^n\|H\|_{p(\hat{\sfQ}_{\bY\bB})}. \label{eq:CS-norm}
\end{align}
In the above calculation, inequality \textcolor{red}{$a$} is due to the weighted version of the Cauchy--Schwarz inequality.

We will now calculate the norms $\|F\|_{2(\hat{\sfQ}_{\bX\bA})}$ and $\|H\|_{p(\hat{\sfQ}_{\bY\bB})}$. It is easy to see from the definition of $F$ and $H$ that
\begin{align*}
\|F\|_{2(\hat{\sfQ}_{\bX\bA})} & = \left(\sum_\bx\frac{\sfP_\bX(\bx)}{|\clA|^n}\cdot\frac{|\clL^\A_{\bx,k}|}{2^{2(k+1)}}\right)^{1/2} \\
 & \overset{\textcolor{red}{b}}{\leq} \left(\sum_\bx\frac{\sfP_\bX(\bx)}{|\clA|^n}\cdot\frac{1}{2^{k+1}}\right)^{1/2} \\
 & = 2^{-(k+1)/2}|\clA|^{-n/2}
\end{align*}
where the inequality \textcolor{red}{$b$} uses $|\clL^\A_{\bx,k}|\leq2^{k+1}$. Similarly,
\[ \|H\|_{p(\hat{\sfQ}_{\bY\bB})} \leq 2^{-(l+1)(1-1/p)}|\clB|^{-n/p}.\]
Putting these into \eqref{eq:CS-norm} we get
\begin{align*}
S_G\left(\clL^\A_{k},\clL^\B_{l}\right) & \leq 2^{(k+1)/2+(l+1)/p}\left(|\clA|^{1/2}|\clB|^{-1/p}\left\|T^{\B\to\A}_G\right\|_{p(\hat{\sfQ}_{YB})\to2(\hat{\sfQ}_{XA})}\right)^n \\
 & = 2^{(k+1)/2+(l+1)/p}\left(\col^{\B\to\A}_p(G)\right)^n
\end{align*}
by the definition of $\col^{\B\to\A}_p(G)$.

We can get a similar upper bound of $2^{(k+1)/p+(l+1)/2}(\col^{\A\to\B}_p(G))^n$ using the definition of $S_G(\clL^\A_k,\clL^\B_l)$ from \eqref{eq:score-2}. The lemma thus follows.
\end{proof}

We are now ready to prove the main theorem.
\begin{proof}[Proof of Theorem~\ref{thm:u-parrep}]
Let $E_n$ be the event that $G^n$ is won with our fixed strategy. By Lemma~\ref{lem:pointwise}, we can upper bound the probability of $E_n$ as
\[
\Pr[E_n] = \sum_{\bx,\by}\sfP_{\bX\bY}(\bx,\by)\sum_{\ba,\bb: \,V^n(\bx,\by,\ba,\bb)=1}\sfP_{\bA\bB|\bx\by}(\ba,\bb) \leq \Lambda_{n,\mu}(G).
\]
Now that we have a product of the marginal conditional probabilities, we can make use of the probability buckets (which are defined on the marginals). Writing $p=1+\mu$,
\begin{align}
& \Lambda_{n,p-1}(G) \nonumber \\
& = \sum_{\bx,\by}\sfP_{\bX\bY}(\bx,\by)\sum_{\ba,\bb: \,V^n(\bx,\by,\ba,\bb)=1}\left(\sfP_{\bA|\bx}(\ba)\sfP_{\bB|\by}(\bb)\right)^{1/p} \nonumber \\
& = \sum_{\bx,\by}\sfP_{\bX\bY}(\bx,\by)\sum_{k,l\geq 0}\sum_{\substack{\ba \in \clL^\A_{\bx,k}, \bb \in \clL^\B_{\by,l}: \\V^n(\bx,\by,\ba,\bb)=1}}\left(\sfP_{\bA|\bx}(\ba)\sfP_{\bB|\by}(\bb)\right)^{1/p} \nonumber \\
& \leq \sum_{k,l\geq 0}2^{-(k+l)/p}\sum_{\bx,\by}\sfP_{\bX\bY}(\bx,\by)\sum_{\substack{\ba \in \clL^\A_{\bx,k}, \bb \in \clL^\B_{\by,l}: \\ \pi^\B_{\bx\by,n}(\ba)=\bb}}1 \nonumber \\
& = \sum_{k,l\geq 0}2^{-(k+l)/p}S_G(\clL^\A_{k},\clL^\B_l) \nonumber \\
& \leq \sum_{k,l\geq 0}2^{-(k+l)/p}\min\left\{2^{k+l+2}w_G^n, 2^{(k+1)/p+(l+1)/2}(\col^{\A\to\B}_p(G))^n, 2^{(k+1)/2+(l+1)/p}(\col^{\B\to\A}_p(G))^n\right\} \label{eq:En-opt}
\end{align}
where in the last line, we have used Lemmas~\ref{lem:hard-list} and \ref{lem:col-list}. Note that the choice $p=1+\mu$ satisfies the requirement $0 < p \leq 2$ for Lemma~\ref{lem:col-list}. Note that we could have upper bounded $S_G(\clL^\A_k,\clL^\B_l)$ using $\col^{\A\to\B}_q$ and $\col^{\B\to\A}_q$ for any $q$ satisfying $0 < q \leq 2$; $q=p$ just happens to be a convenient choice that produces some cancellations.

Let $1/p=r$, $\col_p^{\A\to\B}(G)=c_1$, and $\col^{\B\to\A}_p(G)=c_2$. Then we can write a single term in the summand in \eqref{eq:En-opt} as
\[ \min\left\{4w_G^n2^{(k+l)(1-r)}, \sqrt{2}c_1^n2^{r+(r-1/2)k}, \sqrt{2}c_2^n2^{r+(r-1/2)l}\right\}. \]
Now, we have for any positive $\alpha, \beta, \gamma$ satisfying $\alpha+\beta+\gamma=1$,
\[ \min\{u, v, w\} \leq u^\alpha v^\beta w^\gamma.\]
Using this, the summand is at most
\[ 4^\alpha2^{(r+1/2)(\beta+\gamma)}(w_G^\alpha c_1^\beta c_2^\gamma)^n2^{(\alpha t- \beta s)k}2^{(\alpha t- \gamma s)l},\]
where $s=r-\frac{1}{2}$, and $t=1-r$. The sum is then a double geometric series in $k$ and $l$. The series converges provided $\beta s > \alpha t$ (for $k$), and $\gamma s > \alpha t$ (for $l$). Under these conditions, \eqref{eq:En-opt} is upper bounded by
\begin{align*}
\Lambda_{n,p-1}(G) & \leq 4^\alpha2^{(r+1/2)(\beta+\gamma)}(w_G^\alpha c_1^\beta c_2^\gamma)^n\sum_{k\geq 0}2^{(\alpha t- \beta s)k}\sum_{l\geq0}2^{(\alpha t- \gamma s)l} \\
& \leq \frac{4^\alpha2^{(r+1)(1-\alpha)}}{(1-2^{\alpha t- \beta s})(1-2^{\alpha t - \gamma s})}(w_G^\alpha c_1^\beta c_2^\gamma)^n. \\
\end{align*}
Picking the best strategy for $G^n$ and optimizing over $\alpha, \beta, \gamma$ gives
\begin{equation}\label{eq:En-opt2}
\omega^\dep_\mu(G^n) \leq \Lambda_{n, \mu}(G) \leq \inf_{\substack{\alpha, \beta, \gamma \geq 0: \\ \alpha+\beta+\gamma=1, \\ \beta s > \alpha t, \\ \gamma s > \alpha t}}\frac{4^\alpha2^{(r+1)(1-\alpha)}}{(1-2^{\alpha t- \beta s})(1-2^{\alpha t - \gamma s})}(w_G^\alpha c_1^\beta c_2^\gamma)^n.
\end{equation}

In order to get the best parallel repetition rate for a specific game, the optimization above can be carried out with the specific numbers $w_G, c_1, c_2$. We state a general result with one possible choice (which will later turn out to be the optimal choice for CHSH --- see Lemma~\ref{lem:chsh-alpha}). Pick $\alpha=\frac{r-1/2}{3/2-r} - \frac1n = \frac{1-\mu}{1+3\mu}-\frac1n$, $\beta=\gamma=\frac{1-\alpha}{2}$. Then $\beta s - \alpha t= \gamma s - \alpha t = \frac{3/2-r}{2}\cdot\frac1n = \frac{1+3\mu}{4(1+\mu)}\cdot\frac1n$. Consequently,
\begin{align*}
\omega^\dep_\mu(G^n) & \leq \frac{4^{\frac{1-\mu}{1+3\mu}-\frac{1}{n}}2^{\frac{(3+\mu)}{2(1+\mu)}\left(\frac{4\mu}{1+3\mu}+\frac1n\right)}}{\left(1-2^{-\frac{1+3\mu}{4(1+\mu)n}}\right)^2}\left(w_G^{\frac{1-\mu}{1+3\mu}}(c_1c_2)^{\frac{2\mu}{1+3\mu}}\right)^n.
\end{align*}
Since $1-2^{-\frac{1+3\mu}{4(1+\mu)n}} = \Theta_\mu(1/n)$, substituting the values of $c_1$ and $c_2$ gives
\[ \omega^\dep_\mu(G^n) \leq O(n^2)\cdot\left((\col_{1+\mu}(G))^{4\mu}w_G^{1-\mu}\right)^{\frac{n}{1+3\mu}}. \qedhere \]
\end{proof}

\begin{remark}
In principle, it is possible to do a proof similar to that of Theorem~\ref{thm:u-parrep} for non-unique games as well. If a game has a bound $q^\A$ on the number of Alice-answers that are acceptable for a fixed Bob-answer and $x,y$, and a bound $q^\B$ on the number of Bob-answers that are acceptable for a fixed Alice-answer, then it should be possible to define the collision values $\col^{\A\to\B}_p(G)$ and $\col^{\B\to\A}_p(G)$ with appropriate pre-factors of $q^\A$ and $q^\B$ in order to get a very similar result.
\end{remark}

\subsection{The CHSH game}
In this section, we compute the values of $\col^{\A\to\B}_{1+\mu}(\chsh), \col^{\B\to\A}_{1+\mu}(\chsh)$ for the CHSH game, and justify that the choice of $\alpha, \beta, \gamma$ for the optimization in \eqref{eq:En-opt2} that are picked in the proof of Theorem~\ref{thm:u-parrep} are optimal for CHSH.

With these ingredients, we can finally prove our optimized noisy parallel repetition theorem for CHSH. It gives a nontrivial bound---smaller than the noiseless quantum parallel-repetition value---for $\mu\leq 0.36$. When $\mu\leq\frac{1}{3}$, the noisy EPR pairs are separable, so Theorem~\ref{thm:chsh-n} gives a better upper bound. The new content of our theorem therefore lies in the range $\frac{1}{3}<\mu\leq 0.36$.

\begin{lemma}\label{lem:chsh-col}
The $p$-collision value of the CHSH game is
\[ (\col_p(\chsh))^2
 =
 2^{2/p-1}
 \max_{\substack{z_i\geq0\\ \sum_i z_i^p=1}}
 \frac12
 \left[
  \sum_i z_i^2
  +(z_0+z_1)(z_2+z_3)
 \right].
 \]
\end{lemma}

\begin{proof}
The CHSH game is symmetric with respect to Alice and Bob, so $T^{\A\to\B}_\chsh$ and $T^{\B\to\A}_\chsh$ are the same; we denote this operator by $T_\chsh$. Both are operators on the space of functions from $\{0,1\}^2\to\bbR$. We use the basis of normalized indicator functions on this space. The functions $2\delta_{x,a}$ form an orthonormal basis for $(x,a)\in\{0,1\}^2$. In this basis, we have
\[ T_\chsh = \frac12\begin{pmatrix}
 1&0&1&0\\
 0&1&0&1\\
 1&0&0&1\\
 0&1&1&0
 \end{pmatrix}.
 \]
 The distributions $\hat{\sfQ}_{XA}$ and $\hat{\sfQ}_{YB}$ are both uniform over $\{0,1\}^2$. Their normalization factors cancel between the numerator and denominator of $\|T_\chsh\|_{p(\hat{\sfQ}_{XA})\to2(\hat{\sfQ}_{YB})}$, so we may consider the unweighted norms. For a vector $z=(z_0,z_1,z_2,z_3)$, its $\ell_p$ norm is $\left(\sum_{i=0}^3z_i^p\right)^{1/p}$. Moreover,
 \[ T_\chsh z = \frac12\begin{pmatrix}
 1&0&1&0\\
 0&1&0&1\\
 1&0&0&1\\
 0&1&1&0
 \end{pmatrix}
 \begin{pmatrix}
 z_0 \\
 z_1 \\
 z_2 \\
 z_3
 \end{pmatrix}
 =
 \frac12\begin{pmatrix}
z_0+z_2 \\
z_1+z_3 \\
z_0+z_3 \\
z_1+z_2
 \end{pmatrix}.
 \]
 So it can be checked that $\|T_\chsh z\|_2^2 = \frac12\left(\sum_{i=0}^3z_i^2 + (z_0+z_1)(z_2+z_3)\right)$. Therefore, $\|T_\chsh\|_{p\to 2}$ can be written as
 \begin{align}
\|T_\chsh\|^2_{p\to 2} & = \max_z \frac{\frac12\left(\sum_{i=0}^3z_i^2 + (z_0+z_1)(z_2+z_3)\right)}{\left(\sum_{i=0}^3z_i^p\right)^{2/p}} \nonumber \\
 & = \max_{z: \, \sum_{i=0}^3z_i^p = 1} \frac12\left(\sum_{i=0}^3z_i^2 + (z_0+z_1)(z_2+z_3)\right). \label{eq:Q}
 \end{align}
 Since $|\clA|=|\clB|=2$, putting in the normalization factor in the definition of $\col_p(\chsh)$, we get the lemma.
\end{proof}

In order to do the optimization over $z$ in Lemma~\ref{lem:chsh-col}, we can set up a Lagrangian
\[ L(z_0,z_1,z_2,z_3,\lambda) = f(z_0,z_1,z_2,z_3) - \lambda g(z_0,z_1,z_2,z_3)\]
where $f(z_0,z_1,z_2,z_3) = \frac12\sum_{i=0}^3z_i^2+(z_0+z_1)(z_2+z_3)$ and $g(z_0,z_1,z_2,z_3)=\sum_{i=0}^3z_i^p-1$. We can obtain its stationary points by setting all partial derivatives of $L$ to be zero:
\[ \frac{\partial f}{\partial z_i}=\lambda\frac{\partial g}{\partial z_i} \quad \forall i=0,1,2,3, \qquad g(z_0,z_1,z_2,z_3)=0\]
where the last condition arises from the derivative with respect to $\lambda$. The solutions to these equations are the candidate optimizers of $f$. We can check the value of $f$ at each of these points (and potentially other boundary cases) to find the maximum value of $f$.

We have done this optimization numerically for some values of $\mu\geq \frac13$, and they are listed in Table~\ref{tab:c_mu}. See Appendix~\ref{app:chsh-cert} for how this calculation is done for a specific value $\mu=0.36$.
\begin{table}[!ht]
\centering
\begin{tabular}{|c|c|}
\hline
$\mu$ & $\col_{1+\mu}(\chsh)$ \\
\hline
$\frac{1}{3}$ & $0.873935$ \\
$0.36$ & $0.873606$ \\
$0.4$ & $0.875137$ \\
$0.5$ & $0.890899$ \\
$0.6$ & $0.917004$ \\
$0.75$ & $0.951695$ \\
\hline
\end{tabular}
\caption{Values of $\col_{1+\mu}(\chsh)$ for different values of the Werner visibility $\mu$}
\label{tab:c_mu}
\end{table}
The function $\col_{1+\mu}(\chsh)$ does not appear to be strictly increasing in $\mu$: numerically, within the relevant range $\mu\geq\frac13$, the minimum appears to be attained around $\mu=0.3545553$.

\begin{lemma}\label{lem:chsh-alpha}
The choice of $\alpha=\frac{r-1/2}{3/2-r}-\eps, \beta=\gamma=\frac{1-\alpha}{2}$, for $\eps\to0$, is optimal for the CHSH game in the optimization \eqref{eq:En-opt2}.
\end{lemma}

\begin{proof}
For the CHSH game, in the notation of \eqref{eq:En-opt2}, $c_1=c_2=\col_{1+\mu}(\chsh)$. So it is clear $\beta=\gamma$ should be an optimal choice. Moreover, from Table~\ref{tab:c_mu} we can see that for all $\mu$ in the relevant range $\mu\geq \frac13$, $w_\chsh = \frac{1+\sqrt{5}}4 < \col_{1+\mu}(\chsh)$. Since $c_1, c_2$ are the same, the three-variable optimization in \eqref{eq:En-opt2} actually simplifies to
\begin{align*}
\Lambda_{n,\mu}(\chsh) & \leq \inf_{\substack{0 \leq \alpha \leq 1: \\ \frac{1-\alpha}{2}s > \alpha t}}\frac{4^\alpha2^{(r+1/2)(1-\alpha)}}{\left(1-2^{\alpha t - (1-\alpha)s/2}\right)^2}\left(w_\chsh^\alpha(\col_{1+\mu}(\chsh))^{1-\alpha}\right)^n \\
 & = \inf_{\substack{0 \leq \alpha \leq 1: \\ \frac{1-\alpha}{2}s > \alpha t}} C_{\alpha,\mu}\left(w_\chsh^\alpha(\col_{1+\mu}(\chsh))^{1-\alpha}\right)^{n}
\end{align*}
where $C_{\alpha,\mu}$ absorbs the $n$-independent pre-factor.

Since $C_{\alpha,\mu}$ does not affect the asymptotic (in $n$) rate, we essentially only need to make the base of the exponent $n$ as small as possible in the feasible region of the optimization. The boundary of the feasible region satisfies $\alpha t = \frac{1-\alpha}{2}s$, implying $\alpha = \frac{s}{s+2t} = \frac{r-1/2}{3/2-r}$. Making $\alpha$ as large as possible gives the best exponent base, since $w_\chsh < \col_{1+\mu}(\chsh)$. So clearly $\alpha = \frac{r-1/2}{3/2-r} - \eps$ for $\eps\to 0$ is the best choice.
\end{proof}

We now prove the optimized parallel repetition theorem for the CHSH game.
\begin{theorem}\label{thm:chsh-parrep}
For large enough $n$, the depolarized quantum value of $\chsh^n$ is upper bounded by
\[ \omega^\dep_\mu(\chsh^n) \leq \alpha_\mu^n,\]
where $\alpha_\mu < \cos^2(\frac{\pi}{8})$ for $0 \leq \mu \leq 0.36$.
\end{theorem}

\begin{proof}
From Theorems~\ref{thm:u-parrep} and \ref{thm:chsh-n},
\[ \omega^\dep_\mu(\chsh^n) \leq \Lambda_{n,\mu}(\chsh) \leq O_\mu(n^2)\left(\left(\frac{1+\sqrt{5}}{4}\right)^{1-\mu}\cdot\left(\col_{1+\mu}(\chsh)\right)^{4\mu}\right)^{\frac{n}{1+3\mu}}.\]
Using Lemma~\ref{lem:chsh-col}, we certify that $\col_{1+\mu}(\chsh) < 0.873606$ at $\mu=0.36$ in Appendix~\ref{app:chsh-cert}. Consequently, at this $\mu$ value, the base of the $n$-exponent in the upper bound is
\[ \left(\frac{1+\sqrt{5}}{4}\right)^{4/13}\cdot(0.8736041)^{9/13} \leq 0.8534679 < \cos^2\left(\frac{\pi}{8}\right).\]
Taking $\alpha_{0.36} = 0.85347$ (we are increasing the base of the exponent slightly so that we can absorb the $O(n^2)$ factor) we can thus say at $\mu=0.36$, for large enough $n$,
\[ \Lambda_{n,\mu}(\chsh) \leq \alpha_\mu^n.\]
Now from \eqref{eq:lambda-n} we can see that $\Lambda_{n,\mu}$ is a non-decreasing function of $\mu$, since
\[ (\sfP_{\bA|\bx}(\ba)\sfP_{\bB|\by}(\bb))^{1/(1+\mu_1)} \leq (\sfP_{\bA|\bx}(\ba)\sfP_{\bB|\by}(\bb))^{1/(1+\mu_2)}\]
for $\mu_1 \leq \mu_2$. Therefore, we must have for all $0\leq\mu\leq 0.36$,
\[ \omega^\dep_\mu(\chsh^n) \leq \Lambda_{n,\mu}(\chsh) \leq \alpha_\mu^n,\]
for some $\alpha_\mu < \cos^2(\pi/8)$. Note that this holds despite the fact that $(w_\chsh)^{\frac{1-\mu}{1+3\mu}}\col_{1+\mu}(\chsh)^{\frac{4\mu}{1+3\mu}}$ itself might not be non-decreasing in $\mu$.
\end{proof}

\begin{remark}
In the proof of Theorem~\ref{thm:u-parrep}, in deriving \eqref{eq:En-opt2}, we could have used $\col^{\A\to\B}_q(G)$ and $\col^{\B\to\A}_q(G)$ for any $q$, but we went with the simple choice of $q=p$. We can instead keep $q$ as a free variable and also optimize the base of the exponent $n$ over it. This potentially pushes the nontrivial $\mu$ threshold for CHSH up to $0.367$. But numerically certifying this more complex optimization is tedious, and we have not bothered to do so for this small gain.
\end{remark}

\subsection{Other noise models}
In this section, we prove the analogues of Theorems~\ref{thm:parrep}, \ref{thm:u-parrep} and \ref{thm:chsh-parrep} for general unital noise and erasure noise. We are not able to prove a parallel repetition theorem for biased reset noise, since we do not have an appropriate hypercontractivity result for the noise channel.

It suffices to prove the analogue of Lemma~\ref{lem:pointwise} for unital and erasure noise, since the noise model does not appear in the proofs after that lemma. The analogous lemma follows relatively straightforwardly from the hypercontractivity properties of unital and erasure noise.
\begin{lemma}
Let $\sfP^\uni_{\bX\bY\bA\bB}$ and $\sfP^\era_{\bX\bY\bA\bB}$ be the input-output distributions of a strategy for an $n$-fold game $G$ under unital noise $\Delta^\uni_K$ and erasure noise $\Delta^\era_\eps$, respectively. Then
\begin{align*}
\sfP^\uni_{\bA\bB|\bx\by}(\ba,\bb)
& \le
\left(
\sfP^\uni_{\bA|\bx}(\ba)
\sfP^\uni_{\bB|\by}(\bb)
\right)^{1/(1+\|K\|_\infty)}, \\
\sfP^\era_{\bA\bB|\bx\by}(\ba,\bb)
& \le
\left(
\sfP^\era_{\bA|\bx}(\ba)
\sfP^\era_{\bB|\by}(\bb)
\right)^{1/(1+\sqrt{1-\eps})}.
\end{align*}
\end{lemma}

\begin{proof}
The proof for unital noise is exactly the same as for depolarizing noise, since the marginals are also uniform and the hypercontractivity result in Theorem~\ref{thm:unital-hypercontractivity} has exactly the same form as Theorem~\ref{thm:depolarizing-hypercontractivity}.

 For the erasure noise case, the claim is immediate when $\eps=1$, because the shared state is a product state. Assume $\eps<1$. We have
\[ \sfP^\era_{\bA\bB|\bx\by}(\ba,\bb) = \Tr\left[(M^\bx_\ba\otimes N^\by_\bb)(\Omega^\era_\eps)^{\otimes m}\right] \]
for some integer $m$. Let Alice and Bob's registers in the $i$-th copy of $\Omega^\era_\eps$ be denoted by $E^\A_i, E^\B_i$. We also use $E^\A_S, E^\B_S$ to denote all the registers $E^\A_i, E^\B_i$ respectively for all $i\in S$. Then
\[ (\Omega^\era_\eps)^{\otimes m} = \sum_{S \subseteq [m]}(1-\eps)^{m-|S|}\eps^{|S|}\Omega_S,\]
where
\[ \Omega_S = \frac{I_{E^\A_S}}{2^{|S|}}\otimes\state{e^{|S|}}_{E^\B_S}\otimes\state{\Phi^+_{2^{m-|S|}}}_{E^\A_{S^c}E^\B_{S^c}}\]
and $\ket{e^{|S|}}$ denotes $|S|$ copies of $\ket{e}$. We can assume the $M^\bx_\ba$ have no component in the erasure subspace, since Alice's qubit is not erased in $\Omega^\era_\eps$. $N^\by_\bb$ may have components in the erasure subspace, and we use
 \[ (N^\by_\bb)_{S^c} = \bra{e^{|S|}}_{E^\B_S}N^\by_\bb\ket{e^{|S|}}_{E^\B_S}.\]
 Then using the expansion of $(\Omega^\era_\eps)^{\otimes m}$,
 \begin{align*}
 \Tr\left[(M^\bx_\ba\otimes N^\by_\bb)(\Omega^\era_\eps)^{\otimes m}\right] & = \sum_{S\subseteq [m]}(1-\eps)^{m-|S|}\eps^{|S|}\Tr\left[\frac{\Tr_{E^\A_S}(M^\bx_\ba)}{2^{|S|}}\otimes (N^\by_\bb)_{S^c}\state{\Phi^+_{2^{m-|S|}}}_{S^c}\right] \\
 & = \sum_{S\subseteq [m]}(1-\eps)^{m-|S|}\eps^{|S|}\frac{1}{2^{m-|S|}}\Tr\left[\frac{\Tr_{E^\A_S}(M^\bx_\ba)}{2^{|S|}}((N^\by_\bb)_{S^c})^{\mathsf T}\right].
 \end{align*}
 Choose
 \[ r = 1+\sqrt{1-\eps}, \qquad s = \frac{1+\sqrt{1-\eps}}{\sqrt{1-\eps}},\]
 which are H\"older conjugates and satisfy condition (i) in Theorem~\ref{thm:erasure-hypercontractivity}. Applying normalized H\"older's inequality to each summand and then scalar H\"older's inequality gives
 \begin{align}
 & \Tr\left[(M^\bx_\ba\otimes N^\by_\bb)(\Omega^\era_\eps)^{\otimes m}\right] \nonumber \\
 & \leq \sum_{S \subseteq [m]}(1-\eps)^{m-|S|}\eps^{|S|}\left\|\frac{\Tr_{E^\A_S}(M^\bx_\ba)}{2^{|S|}}\right\|_s^\norm\left\|((N^\by_\bb)_{S^c})^{\mathsf T}\right\|_r^\norm \nonumber \\
 & \leq \left(\sum_{S\subseteq [m]}(1-\eps)^{m-|S|}\eps^{|S|}\left(\left\|\frac{\Tr_{E^\A_S}(M^\bx_\ba)}{2^{|S|}}\right\|_s^\norm\right)^s\right)^{1/s}\left(\sum_{S\subseteq [m]}(1-\eps)^{m-|S|}\eps^{|S|}\left(\left\|((N^\by_\bb)_{S^c})^{\mathsf T}\right\|_r^\norm\right)^r\right)^{1/r} \label{eq:era-hyp-prod}
 \end{align}
 where the last step is obtained by applying H\"older a second time on the scalar sequences (which can be thought of as diagonal matrices)
 \[\left\{\left((1-\eps)^{m-|S|}\eps^{|S|}\right)^{1/s}\left\|\frac{\Tr_{E^\A_S}(M^\bx_\ba)}{2^{|S|}}\right\|_s^\norm\right\}_S\]
 and
 \[ \left\{\left((1-\eps)^{m-|S|}\eps^{|S|}\right)^{1/r}\left\|((N^\by_\bb)_{S^c})^{\mathsf T}\right\|_r^\norm\right\}_S,\]
 since each summand in the preceding sum can be written as a product of two terms from these sequences. Now by Theorem~\ref{thm:erasure-hypercontractivity},
 \begin{align*}
\left(\sum_{S\subseteq [m]}(1-\eps)^{m-|S|}\eps^{|S|}\left(\left\|\frac{\Tr_{E^\A_S}(M^\bx_\ba)}{2^{|S|}}\right\|_s^\norm\right)^s\right)^{1/s} \leq \|M^\bx_\ba\|^\norm_r.
 \end{align*}
 By the same argument as in Lemma~\ref{lem:pointwise}, $\|M^\bx_\ba\|^\norm_r$ is at most $(\sfP^\era_{\bA|\bx}(\ba))^{1/r}$. Similarly,
 \begin{align*}
\left(\sum_{S\subseteq [m]}(1-\eps)^{m-|S|}\eps^{|S|}\left(\left\|((N^\by_\bb)_{S^c})^{\mathsf T}\right\|_r^\norm\right)^r\right)^{1/r} & = \left(\sum_{S\subseteq [m]}(1-\eps)^{m-|S|}\eps^{|S|}\left(\left\|(N^\by_\bb)_{S^c}\right\|_r^\norm\right)^r\right)^{1/r} \\
& \leq \left(\sum_{S\subseteq [m]}(1-\eps)^{m-|S|}\eps^{|S|}\frac{\Tr((N^\by_\bb)_{S^c})}{2^{m-|S|}}\right)^{1/r} \\
& = \left(\sfP^\era_{\bB|\by}(\bb)\right)^{1/r}
 \end{align*}
 where the last equality can be seen by writing out
 \[ \sfP^\era_{\bB|\by}(\bb) = \Tr\left[(I\otimes N^\by_\bb)(\Omega^\era_\eps)^{\otimes m}\right].\]
 Putting these into \eqref{eq:era-hyp-prod} we get
 \[ \sfP^\era_{\bA\bB|\bx\by}(\ba,\bb) = \Tr\left[(M^\bx_\ba\otimes N^\by_\bb)(\Omega^\era_\eps)^{\otimes m}\right] \leq \left(\sfP^\era_{\bA|\bx}(\ba)\sfP^\era_{\bB|\by}(\bb)\right)^{1/r}. \qedhere\]
\end{proof}

Using this pointwise bound lemma, we get the following analogues of Theorems~\ref{thm:parrep}, \ref{thm:u-parrep} and \ref{thm:chsh-parrep} for unital noise and erasure noise.
\begin{theorem}\label{thm:parrep-gen}
For a two-player non-local game $G=(\sfP,\clX\times\clY,\clA\times\clB,V)$, suppose
\[ \omega(G^n) \leq w_G^n.\]
Then its noisy quantum values under unital noise $\Delta^\uni_K$ and erasure noise $\Delta^\era_\eps$ satisfy
\[ \omega^\uni_K(G^n) \leq \left(2^{\|K\|_\infty\log(|\clA|\cdot|\clB|)}\cdot w_G\right)^{\frac{n}{1+\|K\|_\infty}}, \qquad \omega^\era_\eps(G^n) \leq \left(2^{\sqrt{1-\eps}\log(|\clA|\cdot|\clB|)}\cdot w_G\right)^{\frac{n}{1+\sqrt{1-\eps}}}. \]
\end{theorem}
\begin{theorem}\label{thm:u-parrep-gen}
Suppose $G$ is a unique game whose parallel-repeated classical value is upper bounded as $\omega(G^n) \leq w_G^n$. Then the noisy quantum values of $G^n$ under unital and erasure noise are upper bounded by
\begin{align*}
\omega^\uni_K(G^n) & \leq O(n^2)\cdot\left((\col_{1+\rho}(G))^{4\rho}w_G^{1-\rho}\right)^{\frac{n}{1+3\rho}}, \\
\omega^\era_\eps(G^n) & \leq O(n^2)\cdot\left((\col_{1+\rho}(G))^{4\rho}w_G^{1-\rho}\right)^{\frac{n}{1+3\rho}},
\end{align*}
where $\rho$ is the maximal correlation of the respective model: $\|K\|_\infty$ for unital noise and $\sqrt{1-\eps}$ for erasure noise.
\end{theorem}
\begin{theorem}\label{thm:chsh-parrep-gen}
For large enough $n$ and unital noise $\Delta^\uni_K$, write $\rho=\|K\|_\infty$. The unital quantum value of $\chsh^n$ is upper bounded by
\[ \omega^\uni_K(\chsh^n) \leq (\beta_\rho)^n,\]
where $\beta_{\rho}<\cos^2(\pi/8)$ as long as $\rho\leq0.36$.

For erasure noise $\Delta^\era_\eps$, the erased quantum value of $\chsh^n$ is upper bounded by
\[ \omega^\era_\eps(\chsh^n) \leq \gamma_\eps^n\]
where $\gamma_\eps < \cos^2(\pi/8)$ as long as $\eps\geq 0.87$.
\end{theorem}

\section{Lower bounds for communication with noisy entanglement}\label{sec:comm-lb}
\subsection{Separation between communication complexities with noisy entanglement and perfect randomness}\label{sec:noisy-c-ub}
In this section, we prove that there exists a relational problem with low $\Q^{\mathrm{noise}}$ but high $\R^\pub$. We first prove an upper bound on the classical value of the threshold parallel-repeated CHSH. This result and a result similar to the next one were also shown in \cite{JK26}, but we reproduce the proof since we use similar threshold parallel repetition theorems in other models.
\begin{lemma}\label{lem:chsh-thres}
The classical value of winning $t$ out of $n$ CHSH games, for $t=((1+\sqrt{5})/4+\eps)n$, is upper bounded by
\[ \omega(\chsh^{t/n}) \leq e^{-2\eps^2n}. \]
\end{lemma}
\begin{proof}
First, we argue that for $n$ CHSH games, the probability of winning every game in any subset $S$ is at most $\left(\frac{1+\sqrt{5}}{4}\right)^{|S|}$. Suppose this is not true for some subset $S$. Then we can give a strategy for $\chsh^{|S|}$ as follows: use the strategy for $\chsh^n$, embed the actual $|S|$ inputs for $\chsh^{|S|}$ into the subset $S$ of $n$, and sample the $n-|S|$ other inputs according to the CHSH distribution. The probability of winning $\chsh^{|S|}$ using this strategy is exactly the probability of winning in the subset $S$ of the original strategy for $\chsh^n$. The probability of winning in $S$ thus cannot be higher than $\left(\frac{1+\sqrt{5}}{4}\right)^{|S|}$ by Theorem~\ref{thm:chsh-n}.

For each coordinate $i\in[n]$, let $W_i$ be the indicator that the strategy wins the $i$-th CHSH game. The preceding argument shows that, for every $S\subseteq[n]$,
\[
\Pr\left[\bigwedge_{i\in S}W_i=1\right]
\leq
\left(\frac{1+\sqrt5}{4}\right)^{|S|}.
\]
Apply Theorem~\ref{thm:gen-chernoff} with $\delta=(1+\sqrt5)/4$ and $\gamma=\delta+\eps$. We obtain
\[
\Pr\left[\sum_{i=1}^nW_i\geq(\delta+\eps)n\right]
\leq e^{-2\eps^2n},
\]
which is exactly the claimed upper bound.
\end{proof}

\begin{lemma}\label{lem:chsh-r-ub}
Any public-coin communication protocol for $\chsh^{t/n}$ for $t=((1+\sqrt{5})/4+\eps)n$ and communication $C$ must satisfy
\[ \Pr_{\sfU}\left[V^{t/n}(\bx,\by,\ba,\bb)=1\right] \leq 2^C\cdot e^{-2\eps^2n},\]
where the probability is over the uniform input distribution of $n$-fold CHSH, as well as the internal randomness.
\end{lemma}
\begin{proof}
Suppose the upper bound does not hold for some protocol $\clP$. Then we can give a strategy $\mathcal{S}$ for $\chsh^{t/n}$ without communication with winning probability more than $e^{-2\eps^2n}$, contradicting Lemma~\ref{lem:chsh-thres}. The way this strategy works is that Alice and Bob will use their shared randomness to guess the transcript of the communication protocol.

If $\clP$ has $J$ rounds, with, say, Alice communicating first and Bob last, then the randomness in $\mathcal{S}$ is divided into blocks as $\mathbf{r} = r^\A_1r^\B_2\ldots r^\A_{J-1}r^\B_J$, where $r^\A_i$ and $r^\B_j$ are the same length as Alice or Bob's message in the $i$-th or $j$-th rounds. If Bob communicates in the $j$-th round of $\clP$, in $\mathcal{S}$ he will pretend that $r^\A_{j-1}$ is the message he got from Alice in the $(j-1)$-th round, and then check whether $r^\B_j$ equals the message he would send given this message from Alice, his input, and his previous messages. If it is not, he will record $\bot$ privately, and otherwise he will carry on to the next round of his communication. Alice will behave similarly. After going through all the rounds of $\clP$, if Alice or Bob have recorded $\bot$ at any point, they will output arbitrarily; otherwise, they will output according to $\clP$.

We will prove that
\[ \Pr_{\mathcal{S}}\left[V^{t/n}(\bx,\by,\ba,\bb)=1\right] \geq 2^{-C}\Pr_{\clP}\left[V^{t/n}(\bx,\by,\ba,\bb)=1\right].\]
We note that if neither Alice nor Bob has recorded $\bot$ at any point, $\mathbf{r}$ is equal to the transcript of $\clP$, and the probability that the output of $\mathcal{S}$ satisfies $V^{t/n}$ is the same as that in $\clP$. Since $\mathbf{r}$ is a uniformly random string, the probability that it is equal to a $C$-length transcript is exactly $2^{-C}$. Since the outputs of $\mathcal{S}$ satisfy the win condition at least when $\mathbf{r}$ is equal to the transcript (and $\clP$ satisfies the win condition), we have the required result.
\end{proof}

We are now ready to prove the separation between $\Q^{\mathrm{noise}}$ and $\R^\pub$ for $\mathrm{noise} \in\{\dep,\res,\era\}$. The separation for $\Q^\uni_K$ will require a bit more work.
\begin{theorem}\label{thm:Rpub-Qt}
For $\mathrm{noise} \in\{\dep,\res,\era\}$, let $p$ be the probability that the noise channel does nothing. For every $p\geq 0.875$, there exists a relational problem whose input and output sets are both $\{0,1\}^n\times\{0,1\}^n$ for which $\Q^{\mathrm{noise}}_p$ is $0$, but $\R^\pub$ is $\Omega(n)$.
\end{theorem}

\begin{proof}
The relation in question will be the worst-case version of $\chsh^{t/n}$, i.e., the inputs, outputs and win condition will be the same as $\chsh^{t/n}$ but we will require winning on every input pair. We use $\wchsh^{t/n}$ to denote this.

By Corollary~\ref{cor:mu-chsh}, there is a noisy quantum strategy for CHSH under all the relevant noise models that wins with worst-case probability $p\cos^2(\pi/8)+\frac{1-p}{2}$. This is greater than $\frac{1+\sqrt{5}}{4}$ when $p > \frac{\sqrt{5}-1}{\sqrt{2}} \approx 0.874032$. The range of $p$ in the theorem therefore satisfies the lower bound. Let $\eps = 0.875\times\cos^2(\pi/8)+\frac{1-0.875}{2} - \frac{1+\sqrt{5}}4$, and take $t=\left(\frac{1+\sqrt{5}}4+\frac\eps2\right)n$ in $\chsh^{t/n}$. The noisy quantum strategy for $\wchsh^{t/n}$ plays this noisy single-copy strategy independently using $n$ EPR pairs. By Theorem~\ref{thm:chernoff}, the probability of winning fewer than $t$ copies is at most $e^{-\eps^2n/2}$ for every $p \geq 0.875$. Therefore, for this $t$, $\Q^{\mathrm{noise}}_p(\wchsh^{t/n})=0$ for every $p\geq 0.875$.

On the other hand, by Lemma~\ref{lem:chsh-r-ub}, the average-case success probability of $\chsh^{t/n}$ for $t=\left(\frac{1+\sqrt{5}}4+\frac\eps2\right)n$ (over the uniform distribution $\sfU$) of any randomized strategy that communicates $C$ bits is at most $2^Ce^{-\eps^2n/2}$. Therefore, achieving average-case success probability at least $2/3$ requires $\Omega(n)$ communication. This means $\R^\pub(\chsh^{t/n},\sfU) = \Omega(n)$, and thus by Theorem~\ref{thm:yao}, $\R^\pub(\wchsh^{t/n}) = \Omega(n)$.
\end{proof}

In order to prove the analogue of Theorem~\ref{thm:Rpub-Qt} for unital noise, we will use Lemma~\ref{lem:wchsh-uni-lb} instead of Corollary~\ref{cor:mu-chsh}. Note that since additional randomness is required in the strategy for Lemma~\ref{lem:wchsh-uni-lb}, we will correspondingly get a $\Q^{\uni,\pub}_K$ protocol from this result, rather than a $\Q^\uni_K$ protocol.

\begin{theorem}\label{thm:Rpub-Quni}
Let $\Delta^\uni_K$ be a unital noise channel whose Bloch matrix $K$ has two largest singular values $s_1,s_2$. If $s_1^2+s_2^2\geq 1.528$, there exists a relational problem whose input and output sets are both $\{0,1\}^n\times\{0,1\}^n$ for which $\Q^{\uni,\pub}_K$ is $0$, but $\R^\pub$ is $\Omega(n)$.
\end{theorem}
\begin{proof}
We have $\omega^\uni_K(\chsh) > \frac{1+\sqrt{5}}4$ if $s_1^2+s_2^2 > (\sqrt{5}-1)^2 \approx 1.527864$. Therefore, by the same argument as in Theorem~\ref{thm:Rpub-Qt}, taking $\wchsh^{t/n}$ as the relation for $t=\left(\frac{1+\sqrt{5}}4+\frac\eps2\right)n$, where $\eps$ is the gap at $s_1^2+s_2^2=1.528$, gives the result. The communication model here is $\Q^{\uni,\pub}_K$, since Lemma~\ref{lem:wchsh-uni-lb} needs the random bits to achieve worst-case success probability $\frac12+\frac14\sqrt{s_1^2+s_2^2}$.
\end{proof}

\begin{remark}
The strategy in Lemma~\ref{lem:chsh-uni-lb} has success probability at least $\frac12+\frac{s_2^2}{2\sqrt{s_1^2+s_2^2}}$ on all inputs. Therefore, we could get a separation between $\Q^\uni_K$ and $\R^\pub$ with slightly worse dependence, if we used this fact. This needs $\frac12+\frac{s_2^2}{2\sqrt{s_1^2+s_2^2}} > \frac{1+\sqrt{5}}4$, implying $s_2^2 > \frac{\sqrt{5}-1}{2}\sqrt{s_1^2+s_2^2}$.
\end{remark}

\subsection{Separation between communication complexities with noisy and noiseless entanglement}
\begin{theorem}\label{thm:Q*-Qnoise}
For $\mu \leq 0.36$, there exists a relational problem whose input and output sets are both $\{0,1\}^n\times\{0,1\}^n$, for which $\Q^*$ is $0$ but $\Q^\dep_\mu$ is $\Omega(n)$.

A similar separation holds between $\Q^*$ and $\Q^\uni_K$ for unital noise with $\|K\|_\infty \leq 0.36$, and $\Q^*$ and $\Q^\era_\eps$ for erasure noise with $\eps \geq 0.87$.
\end{theorem}
\begin{proof}
By Theorem~\ref{thm:chsh-parrep}, for all $\mu \leq 0.36$,
\[ \omega^\dep_\mu(\chsh^n) \leq \left(\cos^2\frac{\pi}{8}-\eps\right)^n \]
for some $\eps=\Omega(1)$. Therefore, by an argument similar to that of Lemma~\ref{lem:chsh-thres}, for $t=\left(\cos^2\frac{\pi}{8}-\frac{\eps}{2}\right)n$
\[ \omega^\dep_\mu(\chsh^{t/n}) \leq e^{-\eps^2n/2}.\]
Moreover, by an argument similar to that of Lemma~\ref{lem:chsh-r-ub}, any noisy-entanglement protocol with visibility $\mu$, shared randomness, and $C$ bits of communication for $\chsh^{t/n}$ must satisfy
\[ \Pr_{\sfU}\left[V^{t/n}(\bx,\by,\ba,\bb)=1\right] \leq 2^C\cdot e^{-\eps^2n/2}.\]
Note that in the simulation of the communication protocol by a non-local game strategy in this case, Alice and Bob are producing messages corresponding to the protocol by measurements on their parts of the shared noisy states, rather than deterministically from the transcript and randomness. However, the same arguments still hold, since once the measurement outcomes are fixed, the protocol is deterministic, and local consistency for Alice and Bob implies global consistency of the transcript.

Now if we take $\wchsh^{t/n}$, with $t$ as we've fixed before, to be our relational problem, we can say that $\Q^{\dep,\pub}_\mu(\chsh^{t/n}, \sfU) = \Omega(n)$. Therefore, by Theorem~\ref{thm:yao},
\[ \Q^{\dep,\pub}_\mu(\wchsh^{t/n}) = \Omega(n)\]
for every $\mu \leq 0.36$.
On the other hand, by playing the noiseless quantum strategy for $\chsh$ $n$ times independently, by Theorem~\ref{thm:chernoff}, $t$ out of $n$ games are won with high probability. Therefore, $\Q^*(\wchsh^{t/n})=0$.

The results for unital noise and erasure noise follow similarly using Theorem~\ref{thm:chsh-parrep-gen}.
\end{proof}

\section{The noisy entanglement cost of communication protocols}\label{sec:ent-lb}

\subsection{A more resource-efficient protocol for estimating inner products}

The purpose of this subsection is to prove Theorem~\ref{thm:canonnebetter}. First, we have the following lemma which shows that two-way quantum communication protocols can be reduced to estimating real inner products.

\begin{lemma}[Reducing communication to real inner product] \label{lem:reductiontoIP}
    Fix an entanglement-assisted protocol with two-way classical communication and $n$-bit inputs that communicates $C$ bits and produces an output bit $b$. Take $\eps > 0$. For some
    \[
    D=O\!\left(\frac{n2^{2C}}{\eps^2}\right),
    \]
    there exist sub-unit vectors $\{\ket{v_x}\}_{x \in \{0,1\}^n}$ and $\{\ket{w_y}\}_{y \in \{0,1\}^n}$ in $\bbR^D$ such that, for all $x,y$,
    \[|\Pr[b=1|x,y] - 2^C \braket{v_x | w_y}| < \eps.\]
\end{lemma}
\begin{proof}
    We can assume without loss of generality that the initial shared entangled state is pure, and also that all measurements as well as the initial state are real-valued. Write it as $\ket{\Psi}_{E^\A E^\B} \in \bbR^m \otimes \bbR^m$. We can assume this vector is real by the standard embedding of complex vectors into real vectors. By Lemma~\ref{lem:clevebuhrman} (and its realification), there exist collections of $m \times m$ positive semidefinite matrices
    \[\{M^{x,z}\}_{x \in \{0,1\}^n, z \in \{0,1\}^C}, \{N^{y,z}\}_{y \in \{0,1\}^n, z \in \{0,1\}^C}\]
    with spectral norm at most one such that, for all $x,y \in \{0,1\}^n$,
    \[\Pr[b=1|x,y] = \sum_z \bra{\Psi} M^{x,z} \otimes N^{y,z} \ket{\Psi}\]
    Writing, for $x,y \in \{0,1\}^n$:
    \[\ket{v'_x} = \frac{1}{\sqrt{2^C}} \sum_z \ket{z} \otimes ((M^{x,z} \otimes I) \ket{\Psi})\]
    \[\ket{w'_y} = \frac{1}{\sqrt{2^C}}  \sum_z \ket{z} \otimes ((I \otimes N^{y,z}) \ket{\Psi})\]
    We see that $\Pr[b=1|x,y] = 2^C \braket{v'_x | w'_y}$. Moreover,
    \[
    \|v'_x\|_2^2
    =2^{-C}\sum_z
    \bra{\Psi}(M^{x,z})^2\otimes I\ket{\Psi}
    \leq1,
    \]
    because $0\preceq M^{x,z}\preceq I$ implies $(M^{x,z})^2\preceq I$. Similarly, $\|w'_y\|_2\leq1$.

    Apply Lemma~\ref{lem:jl-inner-products} to the union
    \[
    \{v'_x:x\in\{0,1\}^n\}\cup\{w'_y:y\in\{0,1\}^n\},
    \]
    which contains $2^{n+1}$ vectors, with approximation parameter $\eps/2^C$. It gives a target dimension
    \[
    D=O\!\left(\frac{\log(2^{n+2})}{(\eps/2^C)^2}\right)
    =O\!\left(\frac{n2^{2C}}{\eps^2}\right)
    \]
    and sub-unit vectors $v_x,w_y\in\bbR^D$ such that, for all $x,y$,
    \[|\braket{v'_x | w'_y} -\braket{v_x | w_y}| < \eps/2^C.\]
    Multiplying this inequality by $2^C$ proves the lemma.
\end{proof}

This shows that the simulation of entangled protocols with classical communication reduces to that of estimating inner products of real vectors, defined formally as follows.
\begin{definition} \label{def:ipn}
    For $D \in \mathbb{N}$ and $\eps, \delta > 0$, the problem $\IP_{D, \eps, \delta}$ is defined as follows. Alice and Bob are given sub-unit vectors $u, v \in \bbR^D$. They are to compute $\langle u,v \rangle$ to within additive precision $\eps$ with failure probability at most $\delta$.
\end{definition}

Theorem~\ref{thm:canonnebetter} will then follow from Lemma~\ref{lem:reductiontoIP} along with the following theorem.
\begin{theorem}\label{thm:ipsimulation}
Let $R^\A,R^\B$ be dependent finite-valued random variables, and write $\rho=\MC(R^\A;R^\B)>0$.
For any integer $d\geq1$ and $0<\eps,\delta<1$, there is a classical SMP protocol for $\IP_{D,\eps,\delta}$ with communication
\[
\frac{2^{O(d)}\log(2/\delta)}{\eps^6\rho^{6d}}
\]
that uses
\[
\frac{d\,2^{O(d)}D^{1/d}\log(2/\delta)}{\eps^6\rho^{6d}}
\]
independent copies of $(R^\A,R^\B)$. The constants implicit in $O$ depend only on the distribution of $(R^\A,R^\B)$, apart from the displayed dependence on $\rho$.
\end{theorem}

We first give the proof of Theorem~\ref{thm:canonnebetter} assuming Theorem~\ref{thm:ipsimulation}.
\begin{proof}[Proof of Theorem~\ref{thm:canonnebetter}]
Fix a constant target approximation error. Lemma~\ref{lem:reductiontoIP} reduces the acceptance probability of the original protocol to an inner product in dimension $D=O(n2^{2C})$, up to half of the target error. Because the inner product is multiplied by $2^C$, it suffices to estimate it to additive error $2^{-C}$ times the remaining constant error.

Apply Theorem~\ref{thm:ipsimulation} with this precision and constant failure probability. Its communication bound is $\rho^{-6d}2^{O(C+d)}$. Its sample bound has the additional factor
\[
dD^{1/d}=O\left(dn^{1/d}2^{2C/d}\right).
\]
Thus the protocol uses $\rho^{-6d}2^{O(C+d)}n^{1/d}$ source samples. The constants implicit in $O$ depend only on the distribution of $(R^\A,R^\B)$ and the target approximation error, apart from the displayed dependence on $\rho$.
\end{proof}

The protocol promised by this theorem is given in Figure~\ref{fig:IPprotocol}, and its parameter choices for $T$ and $k$ are given in Lemma~\ref{lem:IP-parameters}.

\begin{figure}
\begin{tcolorbox}[colback=white,colframe=black,sharp corners,boxrule=0.5pt] 
\textbf{Parameters.}
Natural numbers $m,d,k \geq 1$ and a cutoff $T > 0$.

\medskip

\noindent\textbf{Inputs.}

\begin{itemize}
    \item Letting $D = \binom{m}{d}$, Alice and Bob receive sub-unit vectors $u,v \in \bbR^D$.
    \item They respectively receive $\{R^\A_{i,j}\}_{i \in [k], j \in [m]}$ and $\{R^\B_{i,j}\}_{i \in [k], j \in [m]}$, where the pairs $(R^\A_{i,j},R^\B_{i,j})$ are i.i.d. copies of finite-valued random variables $(R^\A,R^\B)$.
\end{itemize}

\medskip
\noindent\textbf{Protocol.}

\begin{enumerate}

    \item Write $\rho=\MC(R^\A;R^\B)>0$, and take $f \in L^2(R^\A)$ and $g \in L^2(R^\B)$ achieving $\rho$ as in Definition~\ref{def:maximalcorrelation1}. After receiving their inputs, for $i \in [k]$, Alice and Bob compute
    \[Z_i^\A = \sum_{S \in \binom{[m]}{d}} u_{S} \left(\prod_{j \in S} f(R^\A_{i,j})\right),\]
    \[Z_i^\B = \sum_{S\in\binom{[m]}{d}} v_{S} \left(\prod_{j \in S} g(R^\B_{i,j})\right).\]
    \item For every $i$, they set:
    \[\tZ_i^\A = \min(T, \max(-T, Z_i^\A)),\]
    \[\tZ_i^\B = \min(T, \max(-T, Z_i^\B)).\]
    Using private randomness, they sample $a_i, b_i \in \{-1,1\}$ according to the following probabilities: 
    \[\Pr[a_i = 1|\tZ_i^\A] = \frac{\tZ_i^\A + T}{2T} \qquad \Pr[b_i = 1|\tZ_i^\B] = \frac{\tZ_i^\B + T}{2T},\]
    and send the $a_i, b_i$ to the referee. 
    \item The referee's estimate for $\langle u, v\rangle$ is:
    \[\frac{T^2}{k\rho^d} \sum_{i \in [k]} a_i b_i.\]
\end{enumerate}

\end{tcolorbox}
\caption{The SMP protocol for $\IP_{D, \eps, \delta}$ using noisy shared randomness}
\label{fig:IPprotocol}
\end{figure}

We first prove a preliminary lemma which will help us prove the main analysis Lemma~\ref{lem:IP-parameters}.
\begin{lemma} \label{lem:protocolok}
Fix $d\geq1$ and a finite-valued real random variable $X$ satisfying $\bbE[X]=0$ and $\bbE[X^2]=1$. For $m\geq d$, let $X_1,\ldots,X_m$ be i.i.d. copies of $X$. Let
\[
\sum_{S\in\binom{[m]}d}c_S^2\leq1,
\]
where the $c_S$ are real coefficients, and define the scalar random variable
\[
Z=\sum_{S\in\binom{[m]}d}c_S\prod_{i\in S}X_i.
\]
Then
\[
\lim_{T\to\infty}\sup_{m\geq d}\sup_{\sum_Sc_S^2\leq1}
\bbE\left[Z^2\mathbbm{1}[|Z|\geq T]\right]=0,
\]
where $\mathbbm{1}[E]$ is the indicator of the event $E$.
\end{lemma}

\begin{proof}
Expanding $Z^2$ and taking expectation shows why orthogonality enters here. If $S\neq S'$, then some index belongs to exactly one of $S,S'$, so the corresponding factor has mean zero and the cross term vanishes. Hence
\[
\bbE[Z^2]=\sum_Sc_S^2\leq1.
\]
Corollary~\ref{cor:finite-space-hypercontractivity} therefore gives
\[
\bbE[Z^4]
\leq
\left(2\sqrt{3}\right)^{4d}\bbE[X^4]^d.
\]
When the event $|Z|\geq T$ happens, we have $Z^2\leq Z^4/T^2$. Consequently,
\[
\bbE\left[Z^2\mathbbm{1}[|Z|\geq T]\right]
\leq
\frac{\bbE[Z^4]}{T^2}
\leq
\frac{\left(2\sqrt{3}\right)^{4d}\bbE[X^4]^d}{T^2}.
\]
This bound is uniform in $m$ and the coefficients $(c_S)_S$, and it converges to zero as $T\to\infty$.
\end{proof}

We can now complete the analysis of the protocol in Figure~\ref{fig:IPprotocol}, and subsequently the proof of Theorem~\ref{thm:ipsimulation}.
\begin{lemma} \label{lem:IP-parameters}
Fix finite-valued random variables $(R^\A,R^\B)$ with maximal correlation $\rho>0$, and let $f,g$ be the maximal-correlation witnesses used in Figure~\ref{fig:IPprotocol}. Set
\[
\HC(R^\A,R^\B)=2\sqrt{3}\max\left\{\bbE[f(R^\A)^4]^{1/4},\bbE[g(R^\B)^4]^{1/4}\right\}.
\]
This is a finite constant depending only on the fixed source and the chosen maximizing witnesses $f,g$.
For any $0<\eps,\delta<1$, take
\[
T=\frac{4\HC(R^\A,R^\B)^{2d}}{\eps\rho^d}, \qquad
k=\left\lceil\frac{2048\,\HC(R^\A,R^\B)^{8d}\log(2/\delta)}{\eps^6\rho^{6d}}\right\rceil.
\]
Then the estimate in Figure~\ref{fig:IPprotocol} has additive error at most $\eps$ with probability at least $1-\delta$.
\end{lemma}

\begin{proof}
Because $f(R^\A)$ and $g(R^\B)$ satisfy $\bbE[f(R^\A)]=\bbE[g(R^\B)]=0$, $\bbE[f(R^\A)^2]=\bbE[g(R^\B)^2]=1$, and $\bbE[f(R^\A)g(R^\B)]=\rho$, independence across $j$ for the source pairs $(R^\A_{i,j},R^\B_{i,j})$ gives
\begin{equation}\label{eq:ZAB-bound}
\bbE[Z_i^\A Z_i^\B]
=\rho^d\sum_{S\in\binom{[m]}d}u_{S}v_{S}
=\rho^d\langle u,v\rangle.
\end{equation}
Indeed, every term indexed by distinct sets $S\neq S'$ contains a zero-expectation factor that appears only once and therefore has expectation zero. The same orthogonality calculation gives
\begin{equation}\label{eq:Z2-bound}
\bbE[(Z_i^\A)^2]\leq1, \qquad \bbE[(Z_i^\B)^2]\leq1.
\end{equation}

By definition, $|\tZ_i^\A|,|\tZ_i^\B|\leq T$, so the rounding probabilities in Figure~\ref{fig:IPprotocol} lie in $[0,1]$. Clipping can only decrease absolute value, so
\[
\bbE[(\tZ_i^\A)^2]\leq\bbE[(Z_i^\A)^2]\leq1,
\qquad
\bbE[(\tZ_i^\B)^2]\leq\bbE[(Z_i^\B)^2]\leq1.
\]
Moreover, $Z_i^\A-\tZ_i^\A=0$ when $|Z_i^\A|<T$, whereas $|Z_i^\A-\tZ_i^\A|\leq|Z_i^\A|$ when $|Z_i^\A|\geq T$. The same statements hold on Bob's side. Consequently,
\begin{equation}\label{eq:Z2-indicator}
(Z_i^\A-\tZ_i^\A)^2
\leq (Z_i^\A)^2\mathbbm{1}[|Z_i^\A|\geq T],
\qquad
(Z_i^\B-\tZ_i^\B)^2
\leq (Z_i^\B)^2\mathbbm{1}[|Z_i^\B|\geq T].
\end{equation}
We now calculate:
\begin{align*}
& \left|\bbE[Z_i^\A Z_i^\B]-\bbE[\tZ_i^\A\tZ_i^\B]\right| \\
& \leq \left|\bbE[Z_i^\A(Z_i^\B-\tZ_i^\B)] + \bbE[\tZ_i^\B(Z_i^\A-\tZ_i^\A)]\right| \\
& \overset{\textcolor{red}{a}}{\leq} \bbE[(Z_i^\A)^2]^{1/2}\bbE[(Z_i^\B-\tZ_i^\B)^2]^{1/2}
+\bbE[(\tZ_i^\B)^2]^{1/2}\bbE[(Z_i^\A-\tZ_i^\A)^2]^{1/2}\\
& \overset{\textcolor{red}{b}}{\leq} \bbE[(Z_i^\B-\tZ_i^\B)^2]^{1/2}
+\bbE[(Z_i^\A-\tZ_i^\A)^2]^{1/2}\\
& \overset{\textcolor{red}{c}}{\leq} \bbE\left[(Z_i^\B)^2\mathbbm{1}[|Z_i^\B|\geq T]\right]^{1/2}
+\bbE\left[(Z_i^\A)^2\mathbbm{1}[|Z_i^\A|\geq T]\right]^{1/2} \\
& \overset{\textcolor{red}{d}}{\leq} \frac{(2\sqrt{3})^{2d}\bbE[f(R^\A)^4]^{d/2}}{T} + \frac{(2\sqrt{3})^{2d}\bbE[g(R^\B)^4]^{d/2}}{T},
\end{align*}
where \textcolor{red}{$a$} is the Cauchy--Schwarz inequality, \textcolor{red}{$b$} uses the second-moment bounds above, \textcolor{red}{$c$} uses \eqref{eq:Z2-indicator}, and \textcolor{red}{$d$} follows by taking square roots in Lemma~\ref{lem:protocolok}. By the definition of $\HC(R^\A,R^\B)$, each of the final two terms is at most $\HC(R^\A,R^\B)^{2d}/T$. With the stated choice of $T$, we therefore have
\[ \left|\bbE[Z_i^\A Z_i^\B]-\bbE[\tZ_i^\A\tZ_i^\B]\right| \leq \eps\rho^d/2.\]
Using \eqref{eq:ZAB-bound}, this gives
\begin{equation}\label{eq:tZAB-bound}
\left|\frac{1}{\rho^d}\bbE[\tZ_i^\A\tZ_i^\B] - \langle u, v\rangle\right| \leq \eps/2. 
\end{equation}

Conditional on $\tZ_i^\A$ and $\tZ_i^\B$, $a_i,b_i$ are independent and satisfy
\[
\bbE[Ta_i|\tZ_i^\A]=\tZ_i^\A,
\qquad
\bbE[Tb_i|\tZ_i^\B]=\tZ_i^\B.
\]
Thus conditional independence and the tower property give
\[
\bbE\left[\frac{T^2}{\rho^d}a_ib_i\right]
=\frac{1}{\rho^d}\bbE[\tZ_i^\A\tZ_i^\B].
\]
The source blocks and private rounding randomness are independent across $i$, so the products $a_ib_i$ are independent and lie in $[-1,1]$. Theorem~\ref{thm:chernoff} therefore gives
\[
\Pr\left[
\left|
\frac{T^2}{k\rho^d}\sum_{i=1}^k a_ib_i
-\frac{1}{\rho^d}\bbE[\tZ_1^\A\tZ_1^\B]
\right|>\frac\eps2
\right]
\leq
2\exp\left(-\frac{k\eps^2\rho^{2d}}{8T^4}\right) \leq \delta
\]
for the stated choice of $T$ and $k$. Finally, \eqref{eq:tZAB-bound} gives
\[ \Pr\left[
\left|
\frac{T^2}{k\rho^d}\sum_{i=1}^k a_ib_i - \langle u, v\rangle\right| > \eps\right] \leq \delta. \qedhere\]
\end{proof}

\begin{proof}[Proof of Theorem~\ref{thm:ipsimulation}]
Set $m=d\lceil D^{1/d}\rceil$. The standard bound $\binom md\geq(m/d)^d$ gives $\binom md\geq D$. Alice and Bob pad their input vectors with zeros to dimension $\binom md$ and apply the protocol in Figure~\ref{fig:IPprotocol}. Padding does not change their inner product.

The protocol sends two bits for each $i\in[k]$, so its communication is $2k$. It uses one source sample for each pair $(i,j)\in[k]\times[m]$, for a total of $km$ samples. Since $D\geq1$,
\[
m=d\lceil D^{1/d}\rceil\leq2dD^{1/d}.
\]
The communication and sample bounds now follow from Lemma~\ref{lem:IP-parameters}.
\end{proof}

\subsection{Noisy entanglement lower bound for Equality}
We now turn to proving Theorem~\ref{thm:lowerboundnumberofcopies}, which gives a lower bound on the amount of noisy shared randomness needed to compute the Equality function with constant communication. We first have the following simple packing lemma. 
\begin{lemma} \label{lem:diagonalpackingbound}
Let $k,q \in \mathbb{N}$ and take $\eps, R, T > 0$. Take $0 \leq a_1,...,a_q \leq 1$ and set $\Lambda = \diag(a_1,...,a_q)$. 
Suppose there exist collections of vectors $\{v_l\}_{l \in [k]}, \{w_l\}_{l \in [k]}$ in $\bbC^q$, all of norm at most $R$, such that:
\begin{enumerate}
    \item For all $l$, 
    \[|\braket{v_l | \Lambda | w_l}| \geq T + \eps\]
    \item For all $l \neq l'$, 
    \[|\braket{v_l | \Lambda | w_{l'}}| \leq T.\]
\end{enumerate}
Then, setting
\[q' = \#\left\{i \in [q] : a_i > \frac{\eps}{4R^2}\right\},\]
we have:
    \[k \leq \left(1+\frac{4R^2}{\eps}\right)^{2q'}.\]
\end{lemma}
\begin{proof}
We can assume that the $a_i$ are sorted in decreasing order without loss of generality. Write:
\begin{align*}
\Lambda_1 & = \diag(a_1, \dots, a_{q'}, 0, \dots,0) \\
\Lambda_2 & = \diag(0, \dots,0, a_{q'+1}, \dots, a_{q}) \\
P & = \diag(\underbrace{1,\ldots,1}_{q'},
      \underbrace{0,\ldots,0}_{q-q'}).
\end{align*}
By definition, $\Lambda_1 = \Lambda-\Lambda_2$, and $\Lambda_1=\Lambda_1P=P\Lambda_1$. Therefore, for any vectors $u,u'$, $\langle u',\Lambda_1u\rangle=\langle\Lambda_1u',Pu\rangle$. The Cauchy--Schwarz inequality then gives $|\langle u',\Lambda_1u\rangle|\leq\|\Lambda_1u'\|_2\|Pu\|_2$. Applying this with $u'=v_l$ and $u=w_l-w_{l'}$ for $l \neq l'$, we get
    \begin{align*}
        \|Pw_l - Pw_{l'}\|_2 &\geq \frac{|\braket{v_l | \Lambda_1 | (w_l - w_{l'})}|}{R}\\
        & \overset{\textcolor{red}{a}}{\geq} \frac{|\braket{v_l | \Lambda | (w_l - w_{l'})}| - |\braket{v_l | \Lambda_2 | (w_l - w_{l'})}|}{R}\\
        & \overset{\textcolor{red}{b}}{\geq} \frac{|\braket{v_l | \Lambda | w_l}| - |\braket{v_l | \Lambda | w_{l'}}| - \|v_l\|_2 \|\Lambda_2(w_l - w_{l'})\|_2}{R}\\
        & \overset{\textcolor{red}{c}}{\geq} \frac{T+\eps-T - \frac{\eps}{4R^2}\|v_l\|_2 \|w_l - w_{l'}\|_2}{R}\\
        & \overset{\textcolor{red}{d}}{\geq} \frac{\eps - \frac\eps{4R^2}\cdot R\cdot 2R}{R}\\
        & = \frac{\eps}{2R}.
    \end{align*}
    Here \textcolor{red}{$a$} is the triangle inequality, and \textcolor{red}{$b$} uses the triangle inequality on the first term and the Cauchy--Schwarz inequality on the second. Step \textcolor{red}{$c$} uses the two hypotheses on $\Lambda$ and the bound $\|\Lambda_2\|_\infty \leq \eps/(4R^2)$. Finally, \textcolor{red}{$d$} uses $\|v_l\|_2,\|w_l\|_2,\|w_{l'}\|_2\leq R$ and the triangle inequality.
    
    This calculation shows that $\{Pw_l/R\}_{l \in [k]}$ is a set of sub-unit vectors in $\bbC^{q'}$ whose pairwise distances are at least $\eps/(2R^2)$. The conclusion follows from Lemma~\ref{lem:packingbound}.
\end{proof}

We can complete the proof of Theorem~\ref{thm:lowerboundnumberofcopies}.
\begin{proof}[Proof of Theorem~\ref{thm:lowerboundnumberofcopies}]
    Let $q_\A=\mathrm{rank}(\Psi_{E^\A})$, $q_\B=\mathrm{rank}(\Psi_{E^\B})$, and $q_0=\max\{q_\A^2,q_\B^2\}$. Let $C=O(1)$ be a universal bound on the communication of the protocols in the family. We can assume that the protocols all communicate exactly $C$ bits without loss of generality. Fix $n\in\mathbb{N}$, write $m=m(n)$ for the number of copies of $\Psi$ used to compute $\EQ_n$, and let $\rho=\MC(\Psi)$. Let $\{M^{x,z}\}_{x \in \{0,1\}^n, z \in \{0,1\}^C}$ and $\{N^{y,z}\}_{y \in \{0,1\}^n, z \in \{0,1\}^C}$ be the operators given by Lemma~\ref{lem:clevebuhrman} for this protocol.

    Apply Theorem~\ref{thm:efronstein} and extend its singular-vector families to orthonormal bases of the two weighted operator spaces. If the two spaces have different dimensions, pad the shorter coefficient vector with zeros. In this way, the one-copy correlation form is represented on $\bbC^{q_0}$ by a diagonal matrix with entries
    \[
    1=\rho_1\geq\rho_2\geq\cdots\geq\rho_{q_0}\geq0,
    \qquad \rho_2=\rho,
    \]
    where the extra entries introduced by padding are zero. Taking tensor products gives padded coefficient vectors $v_{x,z},w_{y,z}\in\bbC^{q_0^m}$ for which
    \begin{align*}
    \Pr[b=1|x,y]
    &= \sum_{z \in \{0,1\}^C}\Tr\!\left(\Psi^{\otimes m}((M^{x,z})^\dagger \otimes N^{y,z})\right) \\
    &= \sum_{\substack{z \in \{0,1\}^C \\ \bs \in [q_0]^m} }
    \left(\prod_{i \in [m]} \rho_{s_i} \right)
    (v_{x,z})^*_{\bs}(w_{y,z})_{\bs}.
    \end{align*}
    Hence, if we set:
    \[\lambda_{\bs} = \prod_{i \in [m]} \rho_{s_i},\qquad v_x = \bigoplus_{z \in \{0,1\}^C} v_{x,z},\qquad w_y = \bigoplus_{z \in \{0,1\}^C} w_{y,z},\]
    \[\Lambda = I_{2^C} \otimes \diag(\lambda_{\bs})_{\bs \in [q_0]^m},\]
    we get that:
    \[\Pr[b=1|x,y] = \braket{v_x | \Lambda | w_y}\]
    where $b$ is the output of the protocol on inputs $(x,y)$.
    
    Since $0\preceq M^{x,z},N^{y,z}\preceq I$, orthonormality in the weighted inner products gives
    \begin{align*}
    \|v_{x,z}\|_2^2
    &=\langle M^{x,z},M^{x,z}\rangle_{\Psi_{E^\A}^{\otimes m}}
    =\Tr\!\left(\Psi_{E^\A}^{\otimes m}(M^{x,z})^2\right)\leq1,\\
    \|w_{y,z}\|_2^2
    &=\langle N^{y,z},N^{y,z}\rangle_{\Psi_{E^\B}^{\otimes m}}
    =\Tr\!\left(\Psi_{E^\B}^{\otimes m}(N^{y,z})^2\right)\leq1.
    \end{align*}
    Therefore,
    \[ \|v_x\|_2^2 = \sum_{z\in\{0,1\}^C}\|v_{x,z}\|_2^2 \leq 2^C, \qquad \|w_y\|_2^2 \leq 2^C\]
    for all $x,y$.

    We now use the correctness of the protocol. We take the output bit $b=1$ to mean that the inputs are equal. Bounded-error correctness therefore gives $\Pr[b=1|x,x]\geq\frac23$ for every $x\in\bits^n$, whereas $\Pr[b=1|x,y]\leq\frac13$ whenever $x\neq y$. Equivalently,
\[
\braket{v_x|\Lambda|w_x}\geq\frac23,
\qquad
\braket{v_x|\Lambda|w_y}\leq\frac13
\quad\text{for }x\neq y.
\]
We can therefore apply Lemma~\ref{lem:diagonalpackingbound} to the $2^n$ pairs of vectors indexed by $x\in\bits^n$, with $R=2^{C/2}, T=\frac13, \eps=\frac13$.
The threshold for a diagonal entry of $\Lambda$ to be counted by the packing lemma is $\frac{\eps}{4R^2} = \frac{1}{12\cdot2^C}$.
Set
\[
N_{\mathrm{large}}
=
\#\left\{
\bs\in[q_0]^m:
\lambda_{\bs}>\frac{1}{12\cdot2^C}
\right\}.
\]
Each $\lambda_{\bs}$ occurs once in every one of the $2^C$ blocks of $\Lambda$. Thus the number of diagonal entries of $\Lambda$ above the packing threshold is $2^CN_{\mathrm{large}}$. Lemma~\ref{lem:diagonalpackingbound} gives
\[
2^n
\leq
\left(1+12\cdot2^C\right)^{2^{C+1}N_{\mathrm{large}}}.
\]
Taking logs yields
\begin{equation}\label{eq:large-modes-lower}
N_{\mathrm{large}}
\geq
\frac{n}{2^{C+1}\log(1+12\cdot2^C)}.
\end{equation}

It remains to upper bound $N_{\mathrm{large}}$ in terms of $m$. For $\bs\in[q_0]^m$, let
\[
|\bs|
=
\#\{i\in[m]:s_i\neq1\}.
\]
Since $\rho_1=1$, and $\MC(\Psi)=\rho_2 \geq \rho_3 \geq \ldots$, $\lambda_{\bs}\leq\rho_2^{|\bs|}$. If $\lambda_{\bs}>1/(12\cdot2^C)$, then
\[
|\bs|\log(1/\rho_2)
<
C+\log 12
<
C+4.
\]
Set
\[
L
=
\left\lceil
\frac{C+4}{\log(1/\rho)}
\right\rceil.
\]
Since $C$ and $\rho$ are constants, $L$ is a constant. Every sequence counted by $N_{\mathrm{large}}$ has at most $L$ entries different from $1$. If exactly $t$ entries differ from $1$, there are $\binom mt$ ways to choose their positions and $(q_0-1)^t$ ways to choose their values. Consequently,
\[
N_{\mathrm{large}}
\leq
\sum_{t=0}^L\binom mt(q_0-1)^t
\leq
(L+1)(q_0-1)^Lm^L.
\]

Combining this with the lower bound in \eqref{eq:large-modes-lower} gives
\[
\frac{n}{2^{C+1}\log(1+12\cdot2^C)}
\leq
(L+1)(q_0-1)^Lm^L.
\]
Therefore,
\[
m
\geq
\left(
\frac{n}{2^{C+1}\log(1+12\cdot2^C)(L+1)(q_0-1)^L}
\right)^{1/L}
=
\Omega(n^{1/L}).
\]
Since $C,q_0,L$ are constants, this proves $m=n^{\Omega(1)}$.
\end{proof}

\section*{Acknowledgments}
We thank Penghui Yao, Honghao Fu and Minglong Qin for helpful discussions. The second author acknowledges the support of the Natural Sciences and Engineering Research Council of Canada (NSERC) grants ALLRP-578455-2022 and RGPIN-2023-03731.
\paragraph{Statement of AI use.} During the development of this work, the authors used ChatGPT 5.6 Sol (at High and Extra High effort) and Codex as research and writing assistants. These tools were used to explore proof strategies, check intermediate calculations, identify possible gaps, and generate computational code. The AI assistants materially affected the proofs in Sections~\ref{sec:noisyentanglement}, \ref{sec:parrep}, and Appendix~\ref{app:chsh-cert}; substantive human guidance was involved in getting the results of Sections~\ref{sec:parrep}. All AI-generated material was critically reviewed and corrected by the authors. The authors independently verified the mathematical claims, proofs, numerical results, and citations appearing in the final manuscript and take full responsibility for its contents.

\bibliographystyle{alpha}
\bibliography{ref}

\appendix

\section{Certified upper bound on a collision value of CHSH}\label{app:chsh-cert}

In this section, we carry out the optimization from Lemma~\ref{lem:chsh-col} for $p=1.36=\frac{34}{25}$, i.e., $\mu=0.36$. We first prove some lemmas about the maximizers of the optimization problem.

\begin{lemma}\label{lem:interior}
Every maximizer of $\|T_\chsh z\|_2^2$ subject to
$\sum_i z_i^p=1$ has all four coordinates strictly positive.
\end{lemma}

\begin{proof}
Suppose first, for example, that $z_0=0$ and $z_2+z_3>0$.  Replace
$z_0$ by $\eps>0$, leaving the other coordinates unchanged.  Before restoring the constraint this changes the
objective
\[ \frac12\left(\sum_{i=0}^3z_i^2 + (z_0+z_1)(z_2+z_3)\right)\]
by
\[
 \frac{z_2+z_3}{2}\eps+\frac{\eps^2}{2}.
\]
The $p$th power of the $\ell_p$ norm changes from $1$ to
$1+\eps^p$.  Thus the perturbed vector must be multiplied by
$(1+\eps^p)^{-1/p}$ to restore the constraint, and the squared
$\ell_2$ objective is multiplied by
\[
 (1+\eps^p)^{-2/p}
 =
 1-\frac2p\eps^p+O(\eps^{2p}).
\]
Because $p>1$, the strictly positive term of order $\eps$ dominates
the loss of order $\eps^p$ caused by renormalization.  The normalized
perturbation therefore has a strictly larger objective for all sufficiently
small $\eps>0$.

The same argument applies to a zero coordinate in either pair whenever the
sum of the coordinates in the opposite pair is positive.  If the entire
opposite pair is zero, the constraint implies that at least one coordinate in
the same pair is positive.  Turning on either coordinate in the opposite pair
then produces a strictly positive linear cross term, and the same
renormalization argument applies. Hence no boundary point can maximize the
objective.
\end{proof}

At an interior stationary point, the Lagrange equations, after absorbing a
positive constant into $\lambda$, are
\begin{align}
 z_0+\frac{z_2+z_3}{2}&=\lambda z_0^{p-1},
 &z_1+\frac{z_2+z_3}{2}&=\lambda z_1^{p-1},\label{eq:lagrange-a}\\
 z_2+\frac{z_0+z_1}{2}&=\lambda z_2^{p-1},
 &z_3+\frac{z_0+z_1}{2}&=\lambda z_3^{p-1}.\label{eq:lagrange-b}
\end{align}
Multiplying these equations by $z_0,z_1,z_2,z_3$, respectively, and
summing gives
\[
 2\|T_\chsh z\|_2^2=\lambda\sum_i z_i^p.
\]
In particular, at a normalized stationary point,
\begin{equation}\label{eq:norm-half-lambda}
 \|T_\chsh z\|_2^2=\frac \lambda2.
\end{equation}

\begin{lemma}\label{lem:classification}
Up to swapping coordinates inside either of the pairs $(z_0,z_1)$ and $(z_2,z_3)$, swapping the two pairs (which leaves the objective function unchanged), and
rescaling, every positive stationary point of the optimization in Lemma~\ref{lem:chsh-col} has one of the following forms:
\begin{enumerate}
 \item $(u,u,u,u)$;
 \item $(u,u,t,1)$ with $t>1$;
 \item $(t,1,t,1)$ with $t>1$.
\end{enumerate}
\end{lemma}

\begin{proof}
If both of the pairs $(z_0,z_1)$ and $(z_2,z_3)$ are constant (meaning $z_0=z_1$ and $z_2=z_3$), write the vector as $(u,u,v,v)$.  Equations
\eqref{eq:lagrange-a}--\eqref{eq:lagrange-b} give
\[
 u+v=\lambda u^{p-1}=\lambda v^{p-1},
\]
so $u=v$.

If exactly one pair is constant, the symmetries and homogeneity allow us to
write the vector as $(u,u,t,1)$ with $t>1$ (we can always take $z_3=1$, since we only care about solutions up to scaling).

It remains to consider the case in which both pairs are non-constant.
Adding the first two Lagrange equations and then the last two shows that
\[
 z_0^{p-1}+z_1^{p-1}=z_2^{p-1}+z_3^{p-1}.
\]
Subtracting the two equations within the first pair and within the second
pair gives
\[
 z_0-z_1
 =
 \lambda\bigl(z_0^{p-1}-z_1^{p-1}\bigr),
 \qquad
 z_2-z_3
 =
 \lambda\bigl(z_2^{p-1}-z_3^{p-1}\bigr).
\]
Both pairs are non-constant, so the two differences on the right are
nonzero.  Dividing and then writing
$z_i=(z_i^{p-1})^{1/(p-1)}$ yields
\[
 \lambda
 =
 \frac{
 (z_0^{p-1})^{1/(p-1)}-(z_1^{p-1})^{1/(p-1)}
 }{
 z_0^{p-1}-z_1^{p-1}
 }
 =
 \frac{
 (z_2^{p-1})^{1/(p-1)}-(z_3^{p-1})^{1/(p-1)}
 }{
 z_2^{p-1}-z_3^{p-1}
 }.
\]

We now explain why the preceding common sum and common divided difference
determine the two pairs up to ordering.  Consider any positive pair whose sum
is $2m$.  After ordering the pair, write it as $m+d,m-d$ with
$0<d<m$.  Its divided difference for the function
$x\mapsto x^{1/(p-1)}$ is
\[
 \frac{
 (m+d)^{1/(p-1)}-(m-d)^{1/(p-1)}
 }{2d}
 =
\frac1{2d}
 \int_{m-d}^{m+d}
 \frac1{p-1}x^{1/(p-1)-1}\,dx.
\]
Differentiating this expression with respect to $d$ gives
\[
 \frac1d\left[
 \frac{
 \frac1{p-1}(m-d)^{1/(p-1)-1}
 +
 \frac1{p-1}(m+d)^{1/(p-1)-1}
 }2
 -
 \frac1{2d}
 \int_{m-d}^{m+d}
 \frac1{p-1}x^{1/(p-1)-1}\,dx
 \right].
\]
Here
\[
 \frac1{p-1}x^{1/(p-1)-1}
 =
 \frac{25}{9}x^{16/9}
\]
is strictly convex on $(0,\infty)$.  The strict trapezoidal inequality for
a strictly convex function says that its average over
$[m-d,m+d]$ is strictly smaller than the average of its two endpoint
values.  Therefore the displayed derivative is positive, so the divided
difference is strictly increasing in $d$.

Apply this to the two pairs $\bigl(z_0^{p-1},z_1^{p-1}\bigr), \bigl(z_2^{p-1},z_3^{p-1}\bigr)$. They have the same sum, and their divided differences are both equal to
$\lambda$.  The strict monotonicity just proved forces the same value of
$d$ and hence
\[
 \{z_0^{p-1},z_1^{p-1}\}
 =
 \{z_2^{p-1},z_3^{p-1}\}.
\]
Since $x\mapsto x^{p-1}$ is injective on $(0,\infty)$,
\[
 \{z_0,z_1\}=\{z_2,z_3\}.
\]
Ordering within the two pairs consistently and scaling the smaller
coordinate to $1$ gives $(t,1,t,1)$ with $t>1$.
\end{proof}

We now bound the objective on each of the three stationary-point families in
Lemma~\ref{lem:classification}.  Two of the bounds use rigorous one-variable
interval calculations.  We first explain exactly what is meant by this.

\paragraph{Exact rational interval arithmetic.}
Every decimal endpoint appearing below is interpreted as an exact rational
number; for example, $78.4530$ means $784530/10000$.  Ordinary floating-point
arithmetic is not used to establish any inequality.

The only non-rational operations that occur are positive rational powers.
Suppose that $x=A/B$ is a nonnegative rational number and that we need to
enclose $x^{a/b}$, where $a$ and $b$ are positive integers.  The verification
program fixes $N=90$ and finds, by integer binary search, the largest integer
$m$ such that
\[
 m^bB^a\leq 2^{Nb}A^a.
\]
It follows, using only this integer comparison, that
\[
 \frac{m}{2^N}\leq x^{a/b}<\frac{m+1}{2^N}.
\]
If equality holds in the first comparison, the two endpoints can of course
be taken equal.  Thus every fractional power is enclosed between dyadic
rational numbers whose validity is checked using exact integer arithmetic.

The program then applies the usual outward-rounded interval rules.  For
example, if $x\in[x_-,x_+]$ and $y\in[y_-,y_+]$, then
\[
 x+y\in[x_-+y_-,x_++y_+],
\]
and the product is enclosed by taking the minimum and maximum of the four
endpoint products.  Division is multiplication by
$[y_-,y_+]^{-1}=[1/y_+,1/y_-]$ whenever $0<y_-\leq y_+$.  Since the powers
used below are increasing on the nonnegative real numbers, an interval power
is enclosed by applying the preceding dyadic construction to its two
endpoints.  Negative rational powers are enclosed by first bounding the
corresponding positive power and then taking reciprocals.

For an interval of possible values of $t$, these rules produce a rational
interval containing every possible value of the relevant left-hand side.
If the lower endpoint of the resulting interval is positive, the expression
is positive throughout the entire $t$-interval; if its upper endpoint is
negative, the expression is negative throughout.  When neither conclusion
follows, the program bisects the $t$-interval into two rational subintervals
and repeats the calculation.  In this way it constructs a finite rational
partition on which every asserted sign is rigorous.

There is one small refinement for
\[
 \lambda=\frac{t-1}{t^{p-1}-1}.
\]
For $0<p-1<1$, the denominator divided by $t-1$ is the secant slope of the
strictly concave function $t\mapsto t^{p-1}$ between $1$ and $t$.  This
secant slope is decreasing in $t$, so $\lambda$ is increasing.  The program
therefore encloses $\lambda$ on a $t$-interval by evaluating its two endpoints
with directed rational rounding.  This is substantially tighter than treating
the two occurrences of $t$ in the quotient as independent intervals.

\begin{lemma}\label{lem:Tchsh-uniform-case}
If a stationary point has the form $(u,u,u,u)$, then
\begin{equation}\label{eq:uniform}
 \frac{\|T_\chsh z\|_2^2}{\|z\|_p^2}
 =4^{1-2/p}=4^{-8/17}<0.520809.
\end{equation}
\end{lemma}

\begin{proof}
The ratio is homogeneous, so the value of the positive number $u$ does not
matter.  Since every row of $T_\chsh$ sums to $1$,
\[
 T_\chsh(u,u,u,u)=(u,u,u,u).
\]
Consequently,
\[
 \|T_\chsh(u,u,u,u)\|_2^2=4u^2,
 \qquad
 \|(u,u,u,u)\|_p^2=(4u^p)^{2/p}=4^{2/p}u^2.
\]
Dividing gives $4^{1-2/p}=4^{-8/17}$.  The exact rational-power enclosure
described above gives
\[
 4^{-8/17}<0.520808006<0.520809.
\qedhere \]
\end{proof}

\begin{lemma}\label{lem:Tchsh-one-constant-pair-case}
If a positive stationary point has, up to the symmetries and rescaling in
Lemma~\ref{lem:classification}, the form $(u,u,t,1)$ with $t>1$, then
\begin{equation}\label{eq:one-pair-bound}
 \frac{\|T_\chsh z\|_2^2}{\|z\|_p^2}<0.550769.
\end{equation}
\end{lemma}

\begin{proof}
Substituting $(u,u,t,1)$ into the Lagrange equations gives
\[
 u+\frac{t+1}{2}=\lambda u^{p-1},
 \qquad
 t+u=\lambda t^{p-1},
 \qquad
 1+u=\lambda.
\]
Subtracting the third equation from the second gives
\[
 t-1=\lambda(t^{p-1}-1).
\]
Since $t>1$, both sides of the denominator below are positive, and hence
\[
 \lambda=\frac{t-1}{t^{p-1}-1},
 \qquad
 u=\lambda-1=\frac{t-1}{t^{p-1}-1}-1.
\]
Substitution into the first Lagrange equation shows that $t$ must satisfy
\begin{equation}\label{eq:one-pair-root}
 \frac{t-1}{t^{p-1}-1}
 \left(\frac{t-1}{t^{p-1}-1}-1\right)^{p-1}
 -
 \left(\frac{t-1}{t^{p-1}-1}-1\right)
 -
 \frac{t+1}{2}
 =0.
\end{equation}
We first deal analytically with the ranges close to $1$ and close to infinity.
This leaves a compact interval bounded away from the removable singularity at
$t=1$ for the interval calculation.

\begin{claim}\label{clm:one-pair-near-one}
The left-hand side of \eqref{eq:one-pair-root} is strictly positive for every
$1<t\leq2$.
\end{claim}

\begin{proof}
Strict concavity of $t^{p-1}=t^{9/25}$ gives the strict tangent-line bound
\[
 t^{p-1}-1<(p-1)(t-1)
\]
for every $t>1$.  Therefore
\[
 \lambda=\frac{t-1}{t^{p-1}-1}>\frac1{p-1}=\frac{25}{9},
 \qquad
 u=\lambda-1>\frac{16}{9}.
\]

Consider the part of the left-hand side of \eqref{eq:one-pair-root} that
depends on $\lambda$:
\[
 \lambda(\lambda-1)^{p-1}-(\lambda-1).
\]
Its derivative with respect to $\lambda$ is
\[
 (\lambda-1)^{p-1}
 +(p-1)\lambda(\lambda-1)^{p-2}
 -1.
\]
For $\lambda\geq25/9$, one has $\lambda-1\geq16/9>1$.  The first term
in this derivative is therefore greater than $1$, while the second term is
positive.  Hence the displayed expression is strictly increasing in
$\lambda$ on the entire range under consideration.  It follows that
\[
 \lambda(\lambda-1)^{p-1}-(\lambda-1)
 >
 \frac{25}{9}\left(\frac{16}{9}\right)^{9/25}
 -\frac{16}{9}.
\]
The last quantity is greater than $3/2$.  Indeed, this is equivalent to
\[
 \left(\frac{16}{9}\right)^{9/25}>\frac{59}{50},
\]
or, after raising both sides to the power $25/9$,
\[
 \frac{16}{9}>\left(\frac{59}{50}\right)^{25/9}.
\]
Because $25/9=2+7/9$ and $x^{7/9}$ is strictly concave,
\[
 \left(\frac{59}{50}\right)^{25/9}
 =
 \left(\frac{59}{50}\right)^2
 \left(1+\frac{9}{50}\right)^{7/9}
 <
 \left(\frac{59}{50}\right)^2\frac{57}{50}
 =
 \frac{198417}{125000}
 <
 \frac{16}{9}.
\]
Finally, $(t+1)/2\leq3/2$ for $1<t\leq2$.  Thus the positive part of
the left-hand side of \eqref{eq:one-pair-root} is strictly greater than
$3/2$, while the final term subtracted from it is at most $3/2$.  The whole
left-hand side is therefore strictly positive.
\end{proof}

\begin{claim}\label{clm:one-pair-large-t}
The left-hand side of \eqref{eq:one-pair-root} is strictly negative for every
$t\geq1000$.
\end{claim}

\begin{proof}
Since $u=\lambda-1<\lambda$, the left-hand side of
\eqref{eq:one-pair-root} is strictly smaller than
\[
 \lambda^p-\frac t2.
\]
Moreover,
\[
 \lambda
 =
 \frac{t-1}{t^{p-1}-1}
 =
 t^{2-p}\frac{1-t^{-1}}{1-t^{1-p}}
 <
 \frac{t^{2-p}}{1-t^{1-p}}.
\]
It follows that
\[
 \frac{2\lambda^p}{t}
 <
 2t^{p(2-p)-1}(1-t^{1-p})^{-p}
 =
 2t^{-(p-1)^2}(1-t^{1-p})^{-p}.
\]
Both $t^{-(p-1)^2}$ and $(1-t^{1-p})^{-p}$ are decreasing for $t>1$:
the first assertion is immediate, while $t^{1-p}$ decreases, so
$1-t^{1-p}$ increases and its negative power decreases.  The right-hand
side is therefore largest on $[1000,\infty)$ at $t=1000$.  The exact
rational-power procedure described above gives
\[
 2(1000)^{-81/625}
 \bigl(1-(1000)^{-9/25}\bigr)^{-34/25}
 <0.920<1.
\]
Hence $2\lambda^p/t<1$, or equivalently $\lambda^p<t/2$.  The left-hand
side of \eqref{eq:one-pair-root} is consequently strictly negative.
\end{proof}

It remains to examine $2\leq t\leq1000$.  A single exact rational interval
calculation, performed by the adaptive subdivision procedure described above,
encloses the sign change as follows:
\[
 \begin{array}{c|c}
 \text{range of }t&
 \text{certified sign of the left-hand side of
 \eqref{eq:one-pair-root}}\\
 \hline
 {[2,78.4530]}&\text{positive}\\
 {[78.4532,1000]}&\text{negative}.
 \end{array}
\]
The small interval $(78.4530,78.4532)$ is deliberately left unresolved:
an interval containing a zero cannot have a strictly positive or strictly
negative enclosure.  Claims~\ref{clm:one-pair-near-one} and
\ref{clm:one-pair-large-t}, together with this calculation, prove that every
solution with $t>1$ lies in
\[
 78.4530<t<78.4532.
\]
This is stronger than merely finding a numerical zero near $78.4531$: the
sign certificates rule out any additional solutions outside the displayed
bracket.

Since $\lambda$ is increasing in $t$, directed rational evaluation of
\[
 u=\frac{t-1}{t^{p-1}-1}-1
\]
at the endpoints of the bracket gives
\[
 19.3338<u<19.3341.
\]
More precisely, the certified enclosure is
\[
 19.3339328<u<19.3339618.
\]

We now convert the unnormalized stationary point into the scale-invariant
objective value.  Its $\ell_p$ normalization factor is
\[
 \bigl(2u^p+t^p+1\bigr)^{1/p}.
\]
In general, if an unnormalized stationary vector is divided by a positive
factor $s$, then the linear left-hand sides of its Lagrange equations are
divided by $s$, whereas its terms $z_i^{p-1}$ are divided by $s^{p-1}$.
Thus its Lagrange multiplier is multiplied by $s^{p-2}$.  Here this gives
the normalized multiplier
\[
 \lambda\bigl(2u^p+t^p+1\bigr)^{(p-2)/p}
 =
 \lambda\bigl(2u^p+t^p+1\bigr)^{-8/17}.
\]
The third unnormalized stationarity equation gives $\lambda=1+u$.
Equation~\eqref{eq:norm-half-lambda} applied to the normalized point therefore
gives
\[
 \frac{\|T_\chsh z\|_2^2}{\|z\|_p^2}
 =
 \frac{1+u}
 {2\bigl(2u^p+t^p+1\bigr)^{8/17}}.
\]
Finally, apply the same outward-rounded rational interval operations to this
expression over the full certified intervals for $t$ and $u$.  The resulting
upper endpoint is less than $0.550768129$, and hence the ratio is strictly
less than $0.550769$, as required.
\end{proof}

\begin{lemma}\label{lem:Tchsh-two-nonconstant-pairs-case}
If a positive stationary point has, up to the symmetries and rescaling in
Lemma~\ref{lem:classification}, the form $(t,1,t,1)$ with $t>1$, then
\begin{equation}\label{eq:two-pair-bound}
 \frac{\|T_\chsh z\|_2^2}{\|z\|_p^2}<0.547454.
\end{equation}
\end{lemma}

\begin{proof}
For $(t,1,t,1)$, the Lagrange equations reduce to
\[
 \frac{3t+1}{2}=\lambda t^{p-1},
 \qquad
 \frac{t+3}{2}=\lambda.
\]
Eliminating $\lambda$ gives
\begin{equation}\label{eq:two-pair-root}
 \frac{3t+1}{t+3}-t^{p-1}=0.
\end{equation}
As in the preceding case, we first handle the two tails analytically.

\begin{claim}\label{clm:two-pair-tails}
The left-hand side of \eqref{eq:two-pair-root} is strictly positive for
$1<t\leq1.001$ and strictly negative for $t\geq24$.
\end{claim}

\begin{proof}
The derivative of the left-hand side is
\[
 \frac8{(t+3)^2}-(p-1)t^{p-2}.
\]
For $1<t\leq1.001$,
\[
 \frac8{(t+3)^2}
 \geq
 \frac8{(4.001)^2}
 >
 \frac9{25}
 \geq
 \frac9{25}t^{-16/25}
 =
 (p-1)t^{p-2}.
\]
Thus the derivative is positive on this interval.  The left-hand side of
\eqref{eq:two-pair-root} vanishes at $t=1$, so it is strictly positive for
$1<t\leq1.001$.

For $t\geq24$,
\[
 \frac{3t+1}{t+3}<3,
\]
whereas
\[
 t^{p-1}\geq24^{9/25}>3.
\]
The last strict inequality is equivalent, after raising both positive sides
to the $25$th power, to the exact integer inequality $24^9>3^{25}$.
Therefore the left-hand side of \eqref{eq:two-pair-root} is strictly negative
for every $t\geq24$.
\end{proof}

On the remaining compact interval, one exact rational interval calculation
encloses the nonuniform sign change:
\[
 \begin{array}{c|c}
 \text{range of }t&
 \text{certified sign of the left-hand side of
 \eqref{eq:two-pair-root}}\\
 \hline
 {[1.001,12.5360]}&\text{positive}\\
 {[12.5363,24]}&\text{negative}.
 \end{array}
\]
Together with Claim~\ref{clm:two-pair-tails}, this proves that every
nonuniform solution satisfies
\[
 12.5360<t<12.5363.
\]

For the unnormalized vector $(t,1,t,1)$, the second stationarity equation
gives $\lambda=(t+3)/2$, and the $p$th power of its $\ell_p$ norm is
$2(t^p+1)$.  Repeating the normalization calculation from
Lemma~\ref{lem:Tchsh-one-constant-pair-case} and then applying
\eqref{eq:norm-half-lambda} gives
\[
 \frac{\|T_\chsh z\|_2^2}{\|z\|_p^2}
 =
 \frac{t+3}{4\bigl(2(t^p+1)\bigr)^{8/17}}.
\]
Outward-rounded rational interval evaluation over
$12.5360<t<12.5363$ gives an upper endpoint smaller than
$0.547453699$, which proves \eqref{eq:two-pair-bound}.
\end{proof}

Finally, we prove the required norm bound.
\begin{lemma}\label{lem:Tchsh-bound}
For $p=34/25$, the CHSH collision operator satisfies
\[
 \|T_\chsh\|_{p\to2}^2\leq0.550769.
\]
\end{lemma}

\begin{proof}
By Lemma~\ref{lem:interior}, the maximum defining
$\|T_\chsh\|_{p\to2}^2$ is attained at a positive stationary point.
Lemma~\ref{lem:classification} shows that, up to symmetries and rescaling,
every such point belongs to one of the three families considered above.
The corresponding bounds are
\[
 0.520809,\qquad 0.550769,\qquad 0.547454
\]
by Lemmas~\ref{lem:Tchsh-uniform-case},
\ref{lem:Tchsh-one-constant-pair-case}, and
\ref{lem:Tchsh-two-nonconstant-pairs-case}, respectively.  The largest of
these is $0.550769$, proving the result.
\end{proof}

Since $2/p-1=8/17$, we get from the definition of $\col_p(\chsh)$ and Lemma~\ref{lem:Tchsh-bound},
\[ (\col_{34/25}(\chsh))^2 < 2^{8/17}\times 0.550769\]
and taking the square root,
\[ \col_{34/25}(\chsh) < 0.873606.\]

\end{document}